\documentclass[a4paper]{article} 
\usepackage{amssymb,amsmath,amsthm,verbatim}
\usepackage[dvips,final]{graphicx}
\usepackage{psfrag}
\usepackage{natbib}
\usepackage{hyperref}

\newtheorem{theorem}{Theorem}[subsection]
\newtheorem{lemma}[theorem]{Lemma}

\theoremstyle{definition}

  {\end{list}}
\newcounter{easyenum}

  {\end{list}}
\newcounter{mediumenum}

  {\end{list}}
\newcounter{hardenum}

\newtheorem{prop}[theorem]{Proposition}

\newcommand{\R}{\mathbb{R}}
\newcommand{\Prob}{\mathbb{P}}
\renewcommand{\Pr}{\Prob}
\renewcommand{\bf}{\bfseries}

\newcommand{\E}{\mathbb{E}}

\newcommand{\w}{\omega}

\newcommand{\eps}{\varepsilon}

\newcommand{\indic}[1]{\boldsymbol{1}_{\{\ensuremath{#1}\}}}

\newcommand{\ie}{i.e.}
\newcommand{\eg}{e.g.}

\usepackage{multicol}
\usepackage{color}

\newcommand{\re}{\mathbb{R}} %
\newcommand{\td}{\mathrm{d}} %
\newcommand{\ub}{\ensuremath{\overline{b}}} %
\newcommand{\lb}{\ensuremath{\underline{b}}} %
\newcommand{\uv}{\ensuremath{\overline{v}}} %
\newcommand{\lv}{\ensuremath{\underline{v}}} %
\newcommand{\dbp}{\mathbf{1}_{\overline{S}_T\geq \ub,\,
    \underline{S}_T\leq \lb}} %
\newcommand{\dbpnonstrict}{\mathbf{1}_{\overline{S}_T > \ub,\,
    \underline{S}_T < \lb}} %
\newcommand{\dbmp}{\mathbf{1}_{\overline{S}_T\geq \ub,\,
    \underline{S}_T\geq \lb}} %
\newcommand{\rdbmp}{\mathbf{1}_{\overline{S}_T\leq \ub\textrm{ or
    }\underline{S}_T\leq \lb}} %
\newcommand{\uh}{\overline{H}} %
\newcommand{\lh}{\underline{H}} %
\newcommand{\uH}{\uh} %
\newcommand{\lH}{\lh} %
\newcommand{\sS}{\overline{S}} %
\newcommand{\iS}{\underline{S}} %

\newcommand{\Pc}{\mathcal{P}}
\newcommand{\Lin}{\mathrm{Lin}}

\title{Robust hedging of double touch barrier options}

\author{A.~M.~G.~Cox\thanks{e-mail:
        \texttt{A.M.G.Cox@bath.ac.uk};
        web: \texttt{www.maths.bath.ac.uk/$\sim$mapamgc/}}\\
        Dept.\ of Mathematical Sciences\\
        University of Bath\\
        Bath BA2 7AY, UK
 \and Jan Ob\l \'oj\thanks{e-mail:
        \texttt{jobloj@imperial.ac.uk}; web:
        \texttt{www.imperial.ac.uk/people/j.obloj/}
        \newline Research supported by a Marie Curie Intra-European
Fellowship within the $6^{th}$ European Community Framework Programme.}\\
        Department of Mathematics\\
        Imperial College London\\
        London SW7 2AZ, UK
}

\date{\today}

\begin{document}
\maketitle
\begin{abstract}
  We consider model-free pricing of digital options, which pay out if
  the underlying asset has crossed both upper and lower barriers. We
  make only weak assumptions about the underlying process (typically
  continuity), but assume that the initial prices of call options with
  the same maturity and all strikes are known. Under such
  circumstances, we are able to give upper and lower bounds on the
  arbitrage-free prices of the relevant options, and further, using
  techniques from the theory of Skorokhod embeddings, to show that
  these bounds are tight. Additionally, martingale inequalities are
  derived, which provide the trading strategies with which we are able
  to realise any potential arbitrages. We show that, depending of the
  risk aversion of the investor, the resulting hedging strategies can
  outperform significantly the standard delta/vega-hedging in presence
  of market frictions and/or model misspecification.

\end{abstract}

\section{Introduction}

In the standard approach to pricing and hedging, one postulates a
model for the underlying, calibrates it to the market prices of
liquidly traded vanilla options and then uses the model to derive
prices and associated hedges for exotic over-the-counter products.
Prices and hedges will be correct only if the model describes
perfectly the real world, which is not very likely. The model-free
approach uses market data to deduce bounds on the prices consistent
with \emph{no-arbitrage} and the associated super- and sub-
replicating strategies, which are robust to model misspecification. In
this work we adopt such an approach to derive model-free prices and
hedges for digital double barrier options.

The methodology, which we now outline, is based on solving the
\emph{Skorokhod embedding problem} (SEP).  We assume \emph{no
  arbitrage} and suppose we know the market prices of calls and puts
for all strikes at one maturity $T$. We are interested in pricing an
exotic option with payoff given by a path-dependent functional
$O(S)_T$. The example we consider here is a
\emph{digital double touch} barrier option struck at $(\lb,\ub)$ which
pays $1$ if the stock price reaches both $\lb$ and $\ub$ before
maturity $T$.  Our aim is to construct a model-free super-replicating
strategy of the form
\begin{equation}
\label{eq:intro_superrepl}
O(S)_T\leq F(S_T)+N_T,
\end{equation}
where $F(S_T)$ is the payoff of a finite portfolio of European puts
and calls and $N_T$ are gains from a self-financing trading strategy
(typically forward transactions). Furthermore, we want
\eqref{eq:intro_superrepl} to be tight in the sense that we can
construct a market model which matches the market prices of calls and
puts and in which we have equality in \eqref{eq:intro_superrepl}.  The
initial price of the portfolio $F(S_T)$ is then the least upper bound
on the price of the exotic $O(S)_T$ and the right hand side of
\eqref{eq:intro_superrepl} gives a simple super-replicating strategy
at that cost. There is an analogous argument for the lower bound
and an analogous sub-replicating strategy.

In fact, in order to construct \eqref{eq:intro_superrepl}, we first
construct the market model which induces the upper bound on the price
of $O(S)_T$ and hence will attain equality in
\eqref{eq:intro_superrepl}.  To do so we rely on the theory of
Skorokhod embeddings (cf.\ Ob\l\'oj \cite{Obloj:04b}). We assume
\emph{no arbitrage} and consider a market model in the risk-neutral
measure so that the forward price process $(S_t:t\leq T)$ is a
martingale\footnote{Equivalently, under a simplifying assumption of
  zero interest rates $S_t$ is simply the stock price process.}.  It
follows from Monroe's theorem \cite{Monroe:78} that $S_t=B_{\rho_t}$,
for a Brownian motion $(B_t)$ with $B_0=S_0$ and some increasing
sequence of stopping times $\{\rho_t:t\leq
T\}$ (possibly relative to an enlarged
filtration). Knowing the market prices of calls and puts for all
strikes at maturity $T$ is equivalent to knowing the distribution
$\mu$ of $S_T$ (cf.~\cite{BreedenLitzenberger:78}). Thus, we can see
the stopping time $\rho=\rho_T$ as a solution to the SEP for
$\mu$. Conversely, let $\tau$ be a solution to the SEP for $\mu$,
i.e.\ $B_\tau\sim\mu$ and $(B_{t\land \tau}:t\geq 0)$ is a uniformly
integrable martingale. Then the process $\tilde{S}_t:=B_{\tau\land
  \frac{t}{T-t}}$ is a model for the stock-price process consistent
with the observed prices of calls and puts at maturity T. In this way,
we obtain a correspondence which allows us to identify market models
with solutions to the SEP and vice versa.  In consequence, to estimate
the fair price of the exotic option $\E O(S)_T$, it suffices to bound
$\E O(B)_\tau$ among all solutions $\tau$ to the SEP. More precisely,
if $O(S)_T=O(B)_{\rho_T}$ a.s., then we have
\begin{equation}\label{eq:exotic_bounds}
  \inf_{\tau: B_\tau\sim\mu} \E O(B)_\tau\leq \E O(S)_T\leq \sup_{\tau: B_\tau\sim\mu} \E O(B)_\tau,
\end{equation}
where all stopping times $\tau$ are such that $(B_{t \wedge \tau})_{t \ge 0}$ is uniformly integrable. Once we
compute the above bounds and the stopping times which achieve them, we
usually have a good intuition how to construct the super- (and sub-)
replicating strategies \eqref{eq:intro_superrepl}.

A more detailed description of the SEP-driven methodology outlined
above can be found in Ob\l\'oj \cite{Obloj:EQF}. The idea of
\emph{no-arbitrage} bounds on the prices goes back to Merton
\cite{Merton:73}. The methods for robust pricing and hedging of
options sketched above go back to the works of Hobson
\cite{Hobson:98b} (lookback option) and Brown, Hobson and Rogers
\cite{Brown:01b} (single barrier options). More recently, Dupire
\cite{Dupire:05} investigated volatility derivatives using the SEP and
Cox, Hobson and Ob\l\'oj \cite{CHO:08} designed pathwise inequalities
to derive price range and robust super-replicating strategies for
derivatives paying a convex function of the local time.

Unlike in previous works, e.g.\ Brown, Hobson and Rogers
\cite{Brown:01b}, we don't find a unique inequality
\eqref{eq:intro_superrepl} for a given barrier option. Instead we find
that depending on the market input (i.e.\ prices of calls and puts)
and the pair of barriers different strategies may be optimal. We
characterise all of them and give precise conditions to decide which
one should be used. This new difficulty is coming from the dependence
of the payoff on both the running maximum and minimum of the
process. Solutions to the SEP which maximise or minimise
$\Pr(\sup_{u\leq\tau}B_u\geq \ub,\inf_{u\leq \tau}B_\tau\leq \lb)$
 have not been developed previously and they
are introduced in this paper. As one might suspect, they are
considerably more involved that the ones by Perkins \cite{Perkins:86}
or Az\'ema and Yor \cite{AzemaYor:79} exploited by Brown, Hobson and
Rogers \cite{Brown:01b}.

From a practical point of view, the no-arbitrage price bounds which we
obtain are too wide to be used for pricing. However, our super- or
sub- hedging strategies can still be used. Specifically, suppose an
agent sells a double touch barrier option $O(S)_T$ for a premium
$p$. She can then set up our superhedge \eqref{eq:intro_superrepl}
for an initial premium $\overline{p}>p$. At maturity $T$ she holds
$H=-O(S)_T+F(S)_T+N_T+p-\overline{p}$ which on average is worth zero,
$\E H=0$, but is also bounded below: $H\geq p-\overline{p}$. In
reality, in the presence of model uncertainty and market frictions,
this can be an appealing alternative to the standard delta/vega
hedging. Indeed, our numerical simulations in Section
\ref{sec:extra_numerics} show that in the presence of transaction
costs a risk averse agent will generally prefer the hedging
strategy we construct to a (daily monitored) delta/vega-hedge.

The paper is structured as follows. First we present the setup: our
assumptions and terminology and explain the types of double barriers
considered in this and other papers. Then in Section \ref{sec:dt} we
consider digital double touch barrier option mentioned above. We first
present super- and sub- replicating strategies and then prove in
Section \ref{sec:dt_pricing} that they induce tight model-free bounds
on the admissible prices of the double touch options. In Section
\ref{sec:extra} we reconsider our assumptions and investigate some
applications. Specifically, in Section \ref{sec:extra_strikes} we
consider the case when calls and puts with only a finite number of
strikes are observed and in Section \ref{sec:extra_jumps} we discuss
discontinuities in the price process $(S_t)$.  In Section
\ref{sec:extra_numerics} we present a numerical investigation of the
performance of our super- and sub- hedging strategies.  Section
\ref{sec:proofs} contains the proofs of main theorems. Additional
figures are attached in Section \ref{sec:figures}.

\subsection{Setup}

In what follows $(S_t)_{t \ge 0}$ is the forward price
process. Equivalently, we can think of the underlying with zero
interest rates, or an asset with zero cost of
carry. In particular, our results can be directly applied
in Foreign Exchange markets for currency pairs from economies with
similar interest rates. Moving to the spot market with non-zero
interest rates is not immediate as our barriers become time-dependent.

We assume that $(S_t)_{t \ge 0}$ has continuous paths. We comment in Section
\ref{sec:extra_jumps} how this assumption can be removed or weakened
to a requirement that given barriers are crossed continuously.  We fix
a maturity $T>0$, and assume we observe the initial spot price $S_0$
and the market prices of European calls for all strikes $K>0$ and
maturity $T$:
\begin{equation}
\label{eq:market_input}
\Big(C(K): K\geq 0\Big),
\end{equation}
which we call the \emph{market input}. For simplicity we assume that
$C(K)$ is twice differentiable and strictly convex on
$(0,\infty)$. Further, we assume that we can enter a forward
transaction at no cost. More precisely, let $\rho$ be a stopping time
relative to the natural filtration of $(S_t)_{t\leq T}$ such that
$S_\rho=\ub$. Then the portfolio corresponding to selling a forward at
time $\rho$ has final payoff $(\ub-S_T)\mathbf{1}_{\rho\leq T}$ and we
assume its initial price is zero. The initial price of a portfolio
with a constant payoff $K$ is $K$. We denote $\mathcal{X}$ the set of
all calls, forward transactions and constants and $\Lin(\mathcal{X})$
is the space of their finite linear combinations, which is precisely
the set of portfolios with given initial market prices.  For
convenience we introduce a pricing operator $\Pc$ which, to a
portfolio with payoff $X$ at maturity $T$, associates its initial
(time zero) price, e.g.\ $\Pc K=K$, $\Pc(S_T-K)^+=C(K)$ and
$\Pc(\ub-S_T)\mathbf{1}_{t\geq \rho}=0$. We also assume $\Pc$ is
linear, whenever defined. Initially, $\Pc$ is only given on $\Lin(\mathcal{X})$. One of
the aims of the paper is to understand extensions of $\Pc$ which do
not introduce arbitrage to $\Lin(\mathcal{X}\cup\{Y\})$, for double touch
barrier derivatives $Y$.  Note that linearity of $\Pc$ on $\Lin(\mathcal{X})$ implies call-put parity holds
 and in consequence we also know the market prices of all European put options with maturity $T$:
$$P(K):=\Pc (K-S_T)^+=K-S_0+C(K).$$

Finally, we assume the market admits \emph{no arbitrage} (or
quasi-static arbitrage) in the sense that any portfolio of initially
traded assets with a non-negative payoff has a non-negative price:
 \begin{equation}
 \label{eq:no_arbitrage}
 \forall X\in \Lin(\mathcal{X}): X\geq 0\Longrightarrow \Pc X\geq 0.
 \end{equation}
 As we do not have any probability measure yet, by $X\geq 0$ we mean
 that the payoff is non-negative for any continuous non-negative stock
 price path $(S_t)_{t\leq T}$.

By a \emph{market model} we mean a filtered probability space
$(\Omega,\mathcal{F},(\mathcal{F}_t),\Pr)$ with a continuous
$\Pr$-martingale $(S_t)$ which matches the market input
\eqref{eq:market_input}.  Note that we consider the model under the
risk-neutral measure and the pricing operator is just the expectation
$\Pc=\E$. Saying that $(S_t)$ matches the market input is equivalent
to saying that it starts in the initial spot $S_0$ a.s.\ and that
$\E(S_T-K)^+=C(K)$, $K>0$.  This in turn is equivalent to knowing the
distribution of $S_T$ (cf.~ \cite{BreedenLitzenberger:78,
  Brown:01b}). We denote this distribution $\mu$ and often refer to it
as the \emph{law of $S_T$ implied by the call prices}. Our regularity
assumptions on $C(K)$ imply that
 \begin{equation}\label{eq:market_law}
 \mu(\td K)=C''(K), \ K>0,
 \end{equation}
so that $\mu$ has a positive density on $(0,\infty)$. We could relax
 this assumption and take the support of $\mu$ to be any interval $[a,b]$. Introducing atoms would complicate our formulae (essentially without introducing new difficulties).

The running maximum and minimum of the price process are denoted
respectively $\overline{S}_t=\sup_{u\leq t}S_u$ and
$\underline{S}_t=\inf_{u\leq t}S_u$. We are interested in this paper
in derivatives whose payoff depends both on $\sS_T$ and $\iS_T$. It is
often convenient to express events involving the running maximum and
minimum in terms of the first hitting times $H_x=\inf\{t: S_t=x\}$,
$x\geq 0$. As an example, note that $\dbp=\mathbf{1}_{H_{\ub}\lor
  H_{\lb}\leq T}$.

We use the notation $a<<b$ to indicate that $a$ is \emph{much smaller}
than $b$ -- this is only used to give intuition and is not formal. The
minimum and maximum of two numbers are denoted $a\land b=\min\{a,b\}$
and $a\lor b=\max\{a,b\}$ respectively, and the positive part is denoted $a^+=a\lor 0$.

\subsection{Types of digital double barrier options}

This paper is part of a larger project of
describing model-independent pricing and hedging of all digital double
barrier options. As mentioned above, in this paper we consider double
touch barrier options. Given barriers $\lb<S_0<\ub$ there are $8$
digital barrier options which pay $1$ depending conditional on the
underlying crossing/not-crossing the barriers\footnote{Naturally,
  there are further $4$ `degenerate' options which only involve one of
  the barriers -- these were treated in \cite{Brown:01b} as mentioned
  above.}. Naturally they come in pairs, \eg{} $\dbmp = 1-\rdbmp$
a.s.\footnote{We assume the distribution of $S_T$ has no atoms and
  this implies that distributions of $\sS_T$ and $\iS_T$ also have no
  atoms, cf.\ Cox and Ob\l\'oj \cite{CoxObloj:range}.}. In consequence
a model-free super- and sub- hedge of a double touch/no-touch barrier
option $\dbmp$ implies respectively a model-free sub- and super- hedge
of a barrier option with payoff $\rdbmp$ (and vice
versa). Furthermore, there are two double touch/no-touch options with
payoffs $\dbmp$ and $\mathbf{1}_{\overline{S}_T\leq
  \ub,\,\underline{S}_T\leq \lb}$ but by symmetry it suffices to
consider one of them. In consequence to understand model-free pricing
and hedging of all digital double barrier options it is enough to
consider three options
\begin{itemize}
\item the \emph{double touch} option with payoff $\dbp$,
\item the \emph{double touch/no-touch} option with payoff $\dbmp$,
\item the \emph{double no-touch} option with payoff $\mathbf{1}_{\overline{S}_T\leq \ub,\,\underline{S}_T\geq \lb}$.
\end{itemize}
The present papers deals with the first type, while the second and third types are treated in Cox and Ob\l\'oj \cite{CoxObloj:range} and Cox and Ob\l\'oj \cite{CoxObloj:2tnt}.

\section{Model-free pricing and hedging}
\label{sec:dt}

We investigate now model-free pricing and hedging of double touch option which pays $1$ if
and only if the stock price goes above $\ub$ and below $\lb$ before
maturity: $\dbp$. We present simple quasi-static super- and sub- replicating strategies which prove to be optimal (i.e.\ replicating) in some market model.
\subsection{Superhedging}
\label{sec:dt_super}
\psfrag{lb}{$\lb$} \psfrag{ub}{$\ub$} \psfrag{K1}{$K_1$}
\psfrag{K2}{$K_2$} \psfrag{K3}{$K_3$} \psfrag{K4}{$K_4$}
\psfrag{K}{$K$} \psfrag{1}{$1$} \psfrag{legend1}{Initial
  portfolio}\psfrag{legend2}{Portfolio at $t=H_{\lb}<H_{\ub}\land
  T$}\psfrag{legend3}{Portfolio at $t=H_{\ub}>H_{\lb}$}
\psfrag{legend4}{Portfolio at $t=H_{\ub}<H_{\lb}\land
  t$}\psfrag{legend5}{Portfolio at
  $t=H_{\lb}>H_{\ub}$}\psfrag{legendh1}{Portfolio after $H_{\lb}$}%
 We present here four super-replicating strategies. All our strategies have the same simple
structure: we buy an initial portfolio of calls and puts and when the
stock price reaches $\lb$ or $\ub$ we buy or sell forward
contracts. Naturally our goal is not only to write a super-replicating
strategy but to write \emph{the smallest super-replicating strategy}
and to do so we have to choose judiciously the parameters. As we will
see in Section \ref{sec:dt_pricing}, for a given pair of barriers
$\lb,\ub$ exactly one of the super-replicating strategies will induce a
tight bound on the derivative's price. We will provide an explicit
criterion determining which strategy to use.
\smallskip\\ \emph{$\uh^{I}$: superhedge for $\lb<<S_0<\ub$.}\\ We buy
$\alpha$ puts with strike $K\in (\lb,\infty)$ and when the stock price
reaches $\lb$ we buy $\beta$ forward contracts, see Figure
\ref{fig:H1}. The values of $\alpha,\beta$ are chosen so that the
final payoff on $(0,K)$, provided the stock price has reached $\lb$,
is constant and equal to $1$. One easily computes that
$\alpha=\beta=(K-\lb)^{-1}$.
\begin{figure}
  \begin{center}
    \includegraphics[width=14cm]{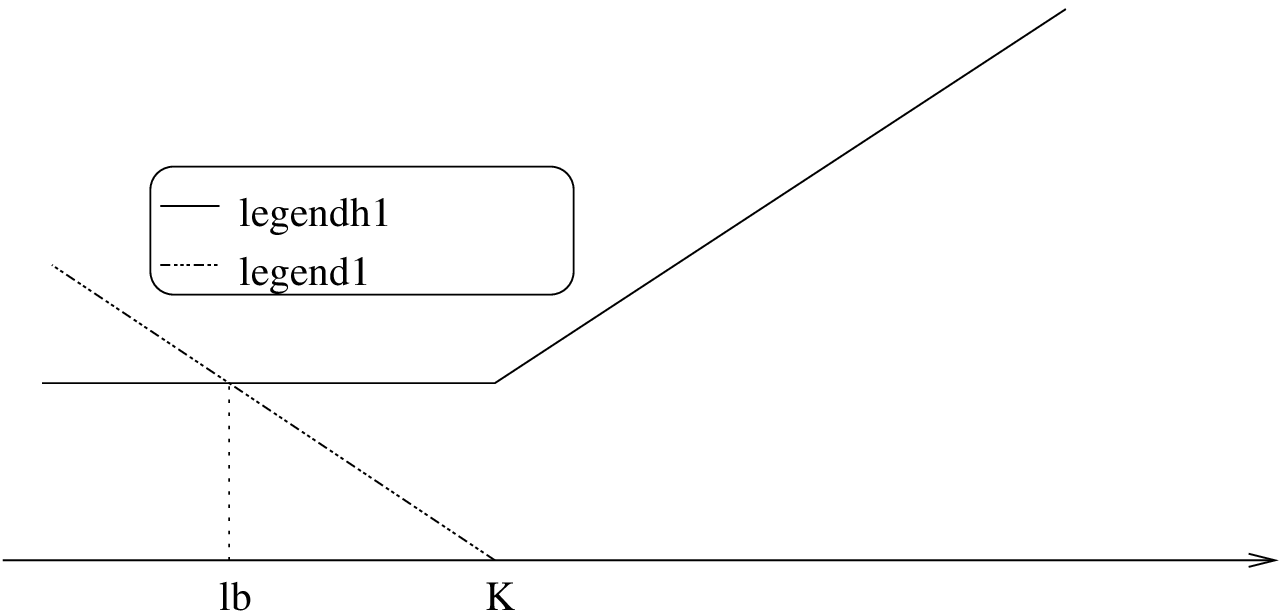}
    \caption{\label{fig:H1}Superhedge $H^{I}$}
  \end{center}\end{figure}
Formally, the super-replication follows from the following inequality
\begin{equation}
  \label{eq:H1}
  \dbp\leq
  \frac{(K-S_T)^+}{K-\lb}+\frac{S_T-\lb}{K-\lb}\mathbf{1}_{\underline{S}_T\leq
    \lb}=:\uh^{I}(K),
\end{equation}
where the last term corresponds to a forward contract entered into, at
no cost, when $S_t=\lb$. Note that $\mathbf{1}_{\underline{S}_T\leq
  \lb}=\mathbf{1}_{H_{\lb}\leq T}$.
\smallskip\\ \emph{$\uh^{II}$: superhedge for $\lb<S_0<<\ub$.}\\ This
is a mirror image of $\uh^{I}$: we buy $\alpha$ calls with strike
$K\in (0,\ub)$ and when the stock price reaches $\ub$ we sell
$\beta$ forward contracts.  The values of $\alpha,\beta$ are chosen so
that the final payoff on $(K,\infty)$, provided the stock price
reached $\ub$, is constant and equal to $1$. One easily computes that
$\alpha=\beta=(\ub-K)^{-1}$. Formally, the super-replication follows
from the following inequality
\begin{equation}
  \label{eq:H2}
  \dbp\leq
  \frac{(S_T-K)^+}{\ub-K}+\frac{\ub-S_T}{\ub-K}\mathbf{1}_{\overline{S}_T\geq
    \ub}=:\uh^{II}(K).
\end{equation}
\emph{$\uh^{III}$: superhedge for $\lb<<S_0<<\ub$.}\\ This superhedge
involves a static portfolio of $4$ calls and puts and at most $4$
dynamic trades. The choice of parameters is judicious which makes the
strategy the most complex to describe.  Choose
\begin{equation}\label{eq:orderingH3}
  0<K_4<\lb<K_3<K_2<\ub<K_1
\end{equation}
and buy $\alpha_i$ calls with strike $K_i$, $i=1,2$ and $\alpha_j$
puts with strike $K_j$, $j=3,4$. If the stock price reaches $\ub$
without having hit $\lb$ before, that is when $H_{\ub}<H_{\lb}\land
T$, sell $\beta_1$ forward. If $H_{\lb}<H_{\ub}\land T$, at $H_{\lb}$
buy $\beta_2$ forwards.  When the stock price reaches $\lb$ after
having hit $\ub$, that is when $H_{\ub}<H_{\lb}\leq T$, buy
$\beta_3=\alpha_3+\beta_1$ forwards. Finally, if $H_{\lb}<H_{\ub}\leq
T$, sell $\beta_4=\alpha_2+\beta_2$ forwards. The choice of $\beta_3$
and $\beta_4$ is such that the final payoff after hitting $\ub$ and
then $\lb$ (resp.\ $\lb$ and then $\ub$) is constant and equal to $1$
on $[K_4,K_3]$ (resp.\ $[K_2,K_1]$). We now proceed to impose
conditions which determine other parameters. A pictorial
representation
of the superhedge is given in Figure \ref{fig:H3}.\\
\begin{figure}
  \begin{center}
    \includegraphics[width=14cm]{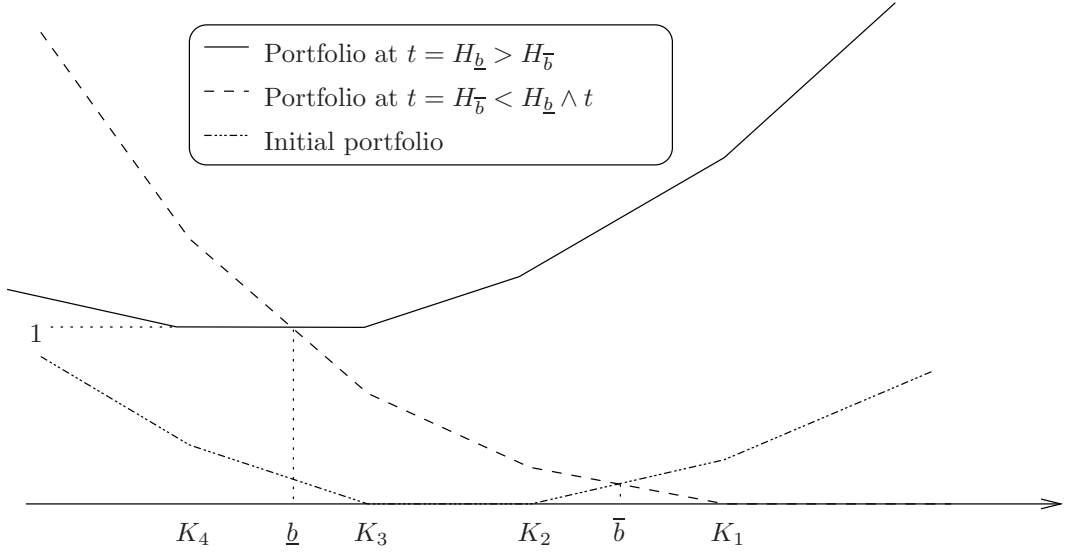}
    \caption{\label{fig:H3}Superhedge $H^{III}$}
  \end{center}
\end{figure}
Note that the initial payoff on $[K_3,K_2]$ is zero.  After hitting
$\ub$ and before hitting $\lb$ the payoff should be zero on
$[K_1,\infty)$ and equal to $1$ at $\lb$. Likewise, after hitting
$\lb$ and before hitting $\ub$, the payoff should be zero on $[0,K_4]$
and equal to $1$ at $\ub$. This yields $6$ equations
\begin{equation}
  \label{eq:conditionH3}
  \left\{ \begin{array}{l}
      \alpha_1+\alpha_2-\beta_1=0\\
      \alpha_2(K_1-K_2)-\beta_1(K_1-\ub)=0\\
      \alpha_3(K_3-\lb)-\beta_1(\lb-\ub)=1
    \end{array} \right. 
  \quad \left\{ \begin{array}{l}
      \alpha_3+\alpha_4-\beta_2=0\\
      \alpha_3(K_3-K_4)+\beta_2(K_4-\lb)=0\\
      \alpha_2(\ub-K_2)+\beta_2(\ub-\lb)=1
    \end{array} \right. . 
\end{equation}
The superhedging strategy corresponds to an a.s.\ inequality
\begin{equation}\label{eq:H3}
  \begin{split}
    \dbp& \leq \alpha_1(S_T-K_1)^++\alpha_2(S_T-K_2)^+
    +\alpha_3(K_3-S_T)^+ +\alpha_4(K_4-S_T)^+\\
    &\quad {}- \beta_1(S_T-\ub)\mathbf{1}_{H_{\ub}<H_{\lb}\land
      T}+\beta_2(S_T-\lb)\mathbf{1}_{H_{\lb}<H_{\ub}\land T}\\
    &\quad {}+\beta_3(S_T-\lb)\mathbf{1}_{H_{\ub}<H_{\lb}\leq
      T}-\beta_4(S_T-\ub)\mathbf{1}_{H_{\lb}<H_{\ub}\leq T}\\
    &\quad \quad {} =:\uh^{III}(K_1,K_2,K_3,K_4),
  \end{split}
\end{equation}
where the parameters, after solving \eqref{eq:conditionH3}, are given
by
\begin{equation}
  \label{eq:valuesH3}\begin{split}
    &\alpha_3=\frac{(K_1-K_2)(\lb-K_4)(\ub-\lb) - (K_1-\ub) (\ub-K_2)
      (\lb-K_4)}{(K_1-K_2) (K_3-K_4) (\ub-\lb)^2 - (K_3-\lb) (K_1-\ub)
      (\ub-K_2)(\lb-K_4)}\\
    &\left\{ \begin{array}{l}
        \alpha_1=\big(1-\alpha_3\frac{K_3-K_4}{\lb-K_4} (\ub-\lb)\big)
        (K_1-\ub)^{-1}\\
        \alpha_2=\big(1-\alpha_3\frac{K_3-K_4}{\lb-K_4} (\ub-\lb)\big)
        (\ub-K_2)^{-1}\\
        \alpha_4=\frac{K_3-\lb}{\lb-K_4}\alpha_3
      \end{array} \right. 
    \quad \left\{ \begin{array}{l}
        \beta_1=\alpha_1+\alpha_2\\
        \beta_2=\alpha_3+\alpha_4
      \end{array} \right. . 
  \end{split}
\end{equation}
Using \eqref{eq:orderingH3} one can verify that $\alpha_3$ and
$\alpha_1$ are non-negative and thus also $\alpha_2$ and $\alpha_4$,
and all $\beta_1,\dots,\beta_4$. 

\emph{$\uh^{IV}$: superhedge for $\lb<S_0<\ub$.}\\ Choose
$0<K_2<\lb<S_0<\ub<K_1$. The initial portfolio is composed of
$\alpha_1$ calls with strike $K_1$, $\alpha_2$ puts with strike $K_2$,
$\alpha_3$ forward contracts and $\alpha_4$ in cash. If we hit $\ub$
before hitting $\lb$ we sell $\beta_1$ forwards, and if we hit $\lb$
before hitting $\ub$ we buy $\beta_2$ forwards. The payoff of the
portfolio should be zero on $[K_1,\infty)$ (resp.\ $[0,K_2]$) and
equal to $1$ at $\lb$ (resp.\ $\ub$) in the first (resp.\ second)
case. Finally, when we hit $\lb$ after having hit $\ub$ we buy
$\beta_3$ forwards, and when we hit $\ub$, having hit $\lb$, we sell
$\beta_4$ forwards. In both cases the final payoff should then be
equal to $1$ on $[K_2,K_1]$, see Figure \ref{fig:H4}. 
\begin{figure}
  \begin{center}
    \includegraphics[width=14cm]{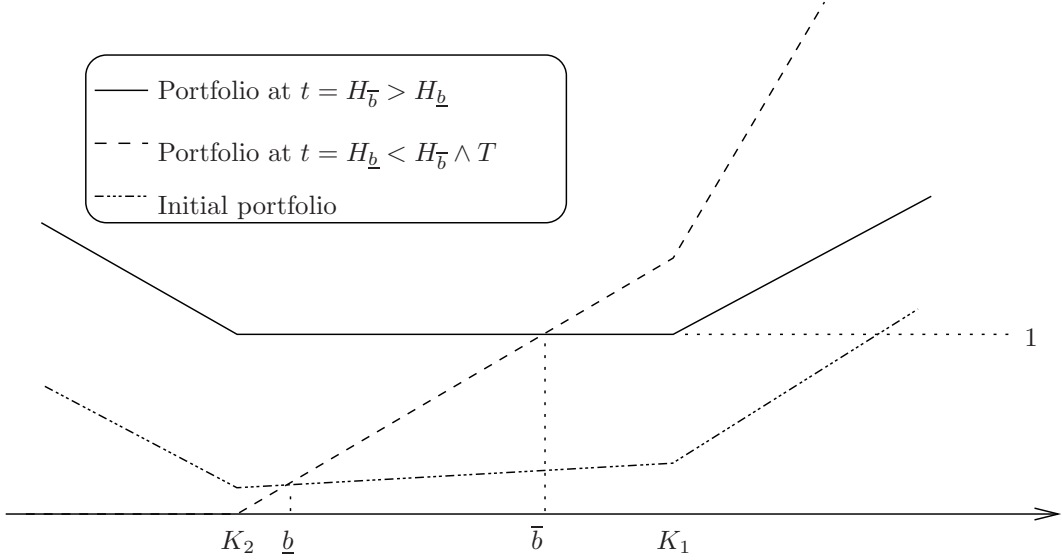}
    \caption{\label{fig:H4}Superhedge $H^{IV}$}
  \end{center}\end{figure}
The superhedging strategy corresponds to the following
a.s.\ inequality
\begin{equation}\label{eq:H4}
  \begin{split}
    \dbp\leq
    &\alpha_1(S_T-K_1)^++\alpha_2(K_2-S_T)^++\alpha_3(S_T-S_0)+\alpha_4\\
    &- \beta_1(S_T-\ub)\mathbf{1}_{H_{\ub}<H_{\lb}\land
      T}+\beta_2(S_T-\lb)\mathbf{1}_{H_{\lb}<H_{\ub}\land T}\\
    &+\beta_3(S_T-\lb)\mathbf{1}_{H_{\ub}<H_{\lb}\leq
      T}-\beta_4(S_T-\ub)\mathbf{1}_{H_{\lb}<H_{\ub}\leq T}\\
    &=:\uh^{IV}(K_1,K_2),
  \end{split}
\end{equation}
where, working out the conditions on $\alpha_i,\beta_i$, the
parameters are
\begin{equation}
  \label{eq:valuesH4}
  \left\{ \begin{array}{l}
      \alpha_1=1/(K_1-\lb)\\
      \alpha_2=1/(\ub-K_2)\\
      \alpha_3=\frac{(K_1-\lb)-(\ub-K_2)}{(K_1-\lb)(\ub-K_2)}\\
      \alpha_4=\frac{\lb\ub-K_1K_2}{(K_1-\lb)(\ub-K_2)}+\alpha_3S_0
    \end{array} \right. 
  \quad \left\{ \begin{array}{l}
      \beta_1=\alpha_1+\alpha_3=1/(\ub-K_2)\\
      \beta_2=\alpha_2-\alpha_3=1/(K_1-\lb)\\
      \beta_3=\alpha_1=1/(K_1-\lb)\\
      \beta_4=\alpha_2=1/(\ub-K_2)
    \end{array} \right. . 
\end{equation}

\subsection{Subhedging}
\label{sec:dt_sub}
We present now three constructions of subhedges which will turn out to
be the best (i.e.\ the most expensive) model-free subhedges depending
on the relative distance of barriers to the spot.  We note however
that there is also a fourth (trivial) subhedge, which has payoff zero
and corresponds to an empty portfolio. In fact this will be the most
expensive subhedge when $\lb<<S_0<<\ub$ and we can construct a market
model in which both barriers are never hit. Details will be given in
Theorem \ref{thm:sub_prices}.
        
\emph{$\lh_I$: subhedge for $\lb<S_0<\ub$.}\\ Choose
$0 < K_2 < \lb < S_0 < \ub < K_1$. The initial portfolio will contain a cash amount, a forward,
calls with 5 different strikes and additionally will also include two digital
options, which pay off $\pounds 1$ provided $S_T$ is above a specified
level. Figure~\ref{fig:H4} demonstrates
graphically the hedging strategy, and we note the effect of the
digital options is to provide a jump in the payoff at the points \lb,
\ub. 

\begin{figure}
  \begin{center}
    \includegraphics[width=14cm]{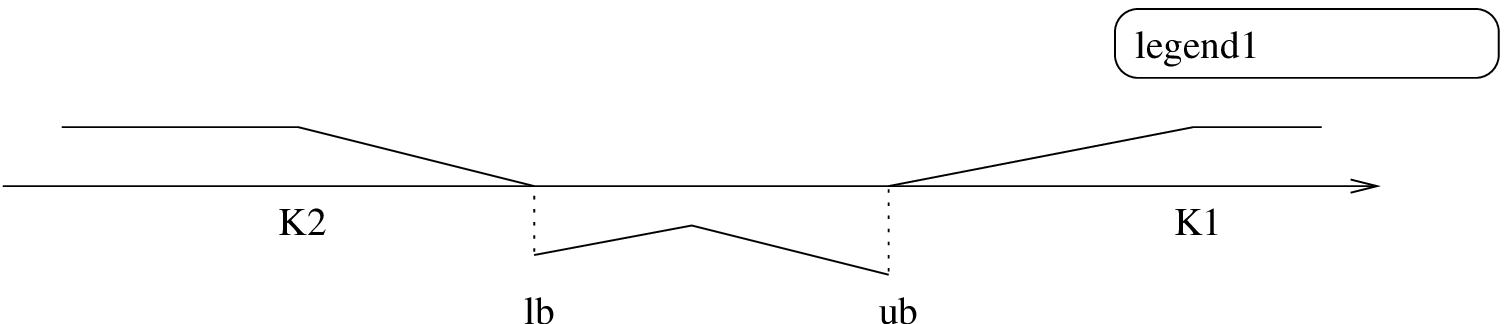}\smallskip

    \includegraphics[width=14cm]{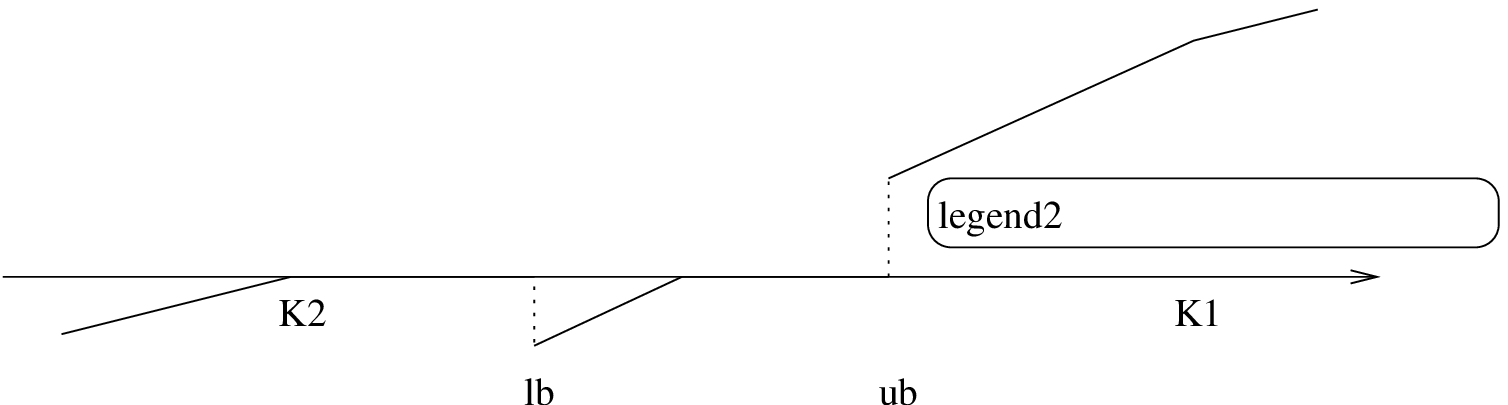}\smallskip

    \includegraphics[width=14cm]{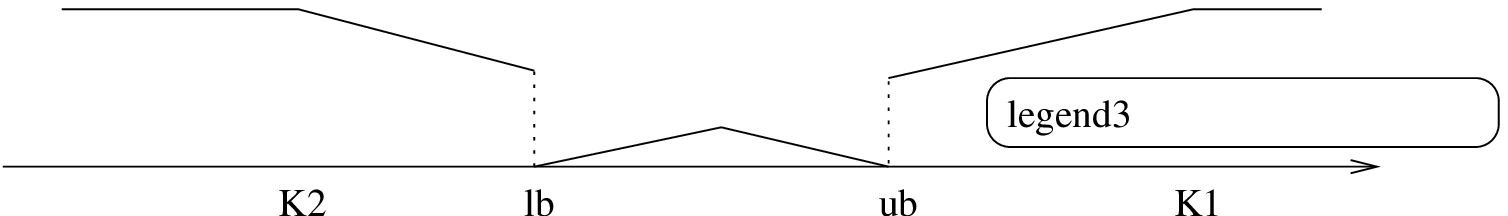}
    \caption{\label{fig:sH}Subhedge $\lh_{I}$.}
  \end{center}\end{figure}

As in the previous cases, the optimality of the construction will
follow from an almost-sure inequality. The relevant inequality is now:
\begin{equation} \label{eqn:asineqsH1}
  \begin{split}
    \dbp & \ge \alpha_0 + \alpha_1(S_T-S_0) -\alpha_2 (S_T-K_2)^+ +
    \alpha_3 (S_T - \lb)^+ - \alpha_3(S_T-K_3)^+\\
    &\quad {} + \alpha_3(S_T-\ub)^+ - (\alpha_3-\alpha_2)(S_T-K_1)^+
    -\gamma_1 \indic{S_T >\lb} + \gamma_2 \indic{S_T \ge \ub}\\
    &\quad {} + (\alpha_2-\alpha_1)(S_T-\lb)
    \indic{H_{\lb}<H_{\ub}\wedge T} - \alpha_2(S_T-\ub) \indic{H_{\lb}
      < H_{\ub} <T}\\
    &\quad {} - (\alpha_3-\alpha_2+\alpha_1)(S_T-\ub) \indic{H_{\ub} <
      H_{\lb} \wedge T} + (\alpha_3-\alpha_2)(S_T-\lb)
    \indic{H_{\ub}<H_{\lb}<T}
  \end{split}. 
\end{equation}
Specifically, we can see that the hedging strategy consists of a
portfolio which contains cash $\alpha_0$, $\alpha_1$ forwards, is
short $\alpha_2$ calls at strike $K_2$ etc. The novel terms here are
the digital options; we note further that the digital options can be
considered also as the limit of portfolios of calls (see for example
\cite{BowieCarr:94}). In our context, we can use their limiting
argument to write: $\Pc \indic{S_T\geq\ub}= -C'(K)$.

The strategy to be followed is then: initially, run to either \lb{} or
\ub{}; supposing that \lb{} is hit first, we buy $(\alpha_2 -
\alpha_1)$ forward, then if we later hit \ub{}, we sell forward
$\alpha_2$ units of the underlying. A similar strategy is followed if
\ub{} is hit first. As previously, the structure imposes some
constraints on the parameters. The relevant constraints are:
\begin{eqnarray}
  0 & = & \alpha_0 + \alpha_1(\lb-S_0) - \alpha_2(\lb -
  K_2) \label{eq:subh1}\\
  0 & = & \alpha_0 + \alpha_1(\ub-S_0) -
  \alpha_2(\ub - K_2) + \alpha_3(K_3-\lb) -
  \gamma_1+\gamma_2 \label{eq:subh2}\\
  1 & = & \alpha_0 + \alpha_1(K_2 -
  S_0) + (\alpha_2-\alpha_1)(K_2 -\lb) -
  \alpha_2(K_2-\ub) \label{eq:subh3}\\
  1 & = & \alpha_0 + \alpha_1(K_2-S_0) - (\alpha_3 -\alpha_2
  +\alpha_1)(K_2 -\ub) +
  (\alpha_3-\alpha_2)(K_2-\lb) \label{eq:subh4}\\
  \gamma_1 & = & (K_3 - \lb) \alpha_3 \label{eq:subh5}\\
  \gamma_2 & = & (\ub - K_3) \alpha_3 \label{eq:subh6}\\
  \frac{K_3-\lb}{K_1-\lb} & = & \frac{\ub - K_3}{\ub-K_2} \label{eq:subh7}
\end{eqnarray}
The equations \eqref{eq:subh1} and \eqref{eq:subh2} arise from the
constraint that initially the payoff is zero at \lb{}, \ub{};
constraints \eqref{eq:subh3} and \eqref{eq:subh4} come from the
constraint that the final payoff is 1 at $K_2$ when both barrier are hit (in either order); 
\eqref{eq:subh5} and \eqref{eq:subh6} represent the fact
that, in the intermediate step, at $K_3$ the gap at \lb{}
(resp. \ub{}) is the size of the respective digital option. The final
constraint, \eqref{eq:subh7} follows from noting that $K_3$ is the
intersection point of the lines from $(\lb,0)$ to $(K_1,1)$ and from
$(K_2,1)$ to $(\ub,0)$. Note that it follows that the initial payoff on $(0,K_1)$ and $(K_2,\infty)$ are co-linear, or that the finial payoff in $K_1$ is $1$ when both barriers are hit.

The given equations can be solved to deduce:
\renewcommand{\arraystretch}{1.5}
\begin{equation}
  \label{eqn:subh1soln}
  \left\{
    \begin{array}{rcl}
      \alpha_0 & = & \frac{S_0 (K_1 + K_2 - \ub - \lb) + \ub \lb - K_1
        K_2}{(\ub-K_2)(K_1 - \lb)}\\
      \alpha_1 & = & \frac{K_1 + K_2 - \ub - \lb}{(\ub-K_2)(K_1 - \lb)}\\
      \alpha_2 & = & \frac{1}{\ub - K_2}\\
      \alpha_3 & = & \frac{\ub - K_2 + K_1 - \lb}{(\ub-K_2)(K_1 - \lb)}
    \end{array}
  \right.  \quad \left\{
    \begin{array}{rcl}
      K_3 & = & \frac{\ub K_1 - \lb K_2}{\ub - K_2 - \lb + K_1}\\
      \gamma_1 & = & \frac{\ub - \lb}{\ub - K_2}\\
      \gamma_2 & = & \frac{\ub - \lb}{K_1 - \lb}
    \end{array}
  \right. 
\end{equation}
\renewcommand{\arraystretch}{1} We note from the above that $\alpha_2,
\alpha_3, \gamma_1$ and $\gamma_2$ are all (strictly) positive;
further, it can be checked that the quantities $(\alpha_3 - \alpha_2),
(\alpha_2 - \alpha_1), (\alpha_3 - \alpha_2 + \alpha_1)$ are all
positive. It follows that the construction holds for all choices of
$K_1, K_2$ with $K_2 < \lb$, and $K_1 > \ub$. 

For future reference, we define $\lh_I(K_1, K_2)$ to be the random
variable given by the right hand side of \eqref{eqn:asineqsH1}, where
the coefficients are given by the solutions of
\eqref{eq:subh1}--\eqref{eq:subh7}. 

\emph{$\lh_{II}$: subhedge for $\lb<S_0<<\ub$.}\\ While
the above hedge can be considered to be the `typical' subhedge for the
option, there are also two further cases that need to be considered
when the initial stock price, $S_0$, is much closer to one of the
barriers than the other. The resulting subhedge will share many of the
features of the previous construction, however the main difference
concerns the behaviour in the tails; we now have the hedge taking the
value one in the tails only under one of the possible ways of knocking
in (specifically, in the case where $\lb < S_0 << \ub$, we get
equality in the tails only when $\lb$ is hit first.) 

\begin{figure}[t]
  \begin{center}
    \includegraphics[width=14cm]{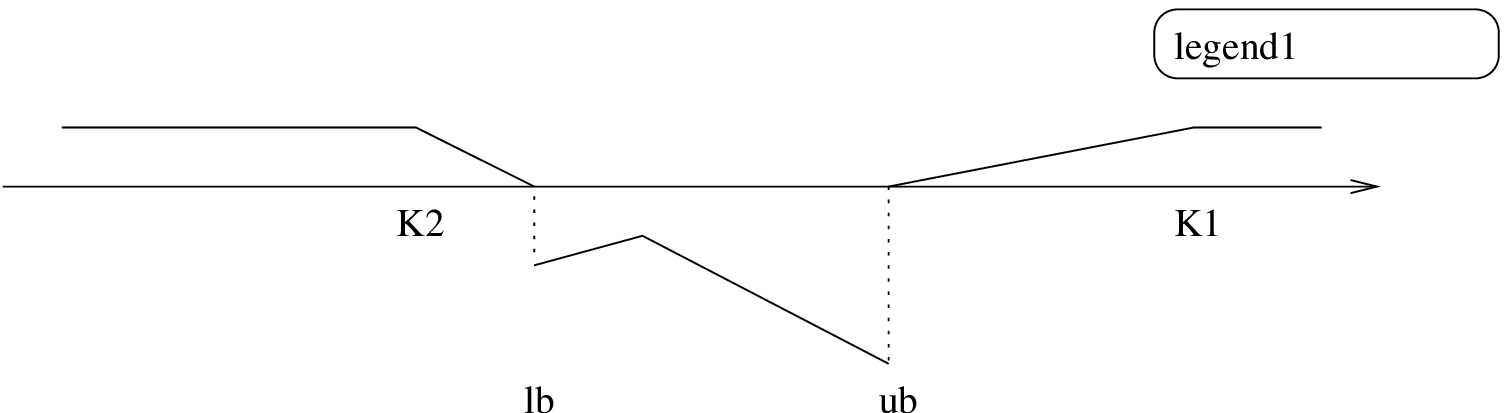}\vspace{-1.5cm}

    \includegraphics[width=14cm]{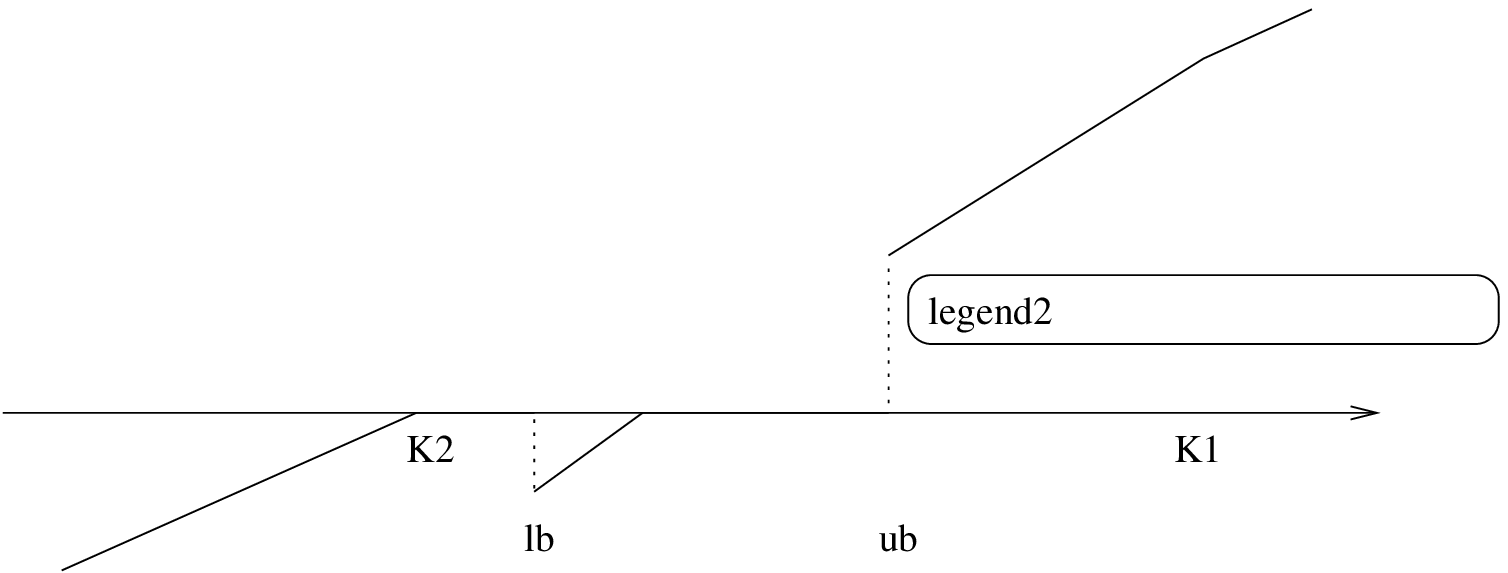}\smallskip

    \includegraphics[width=14cm]{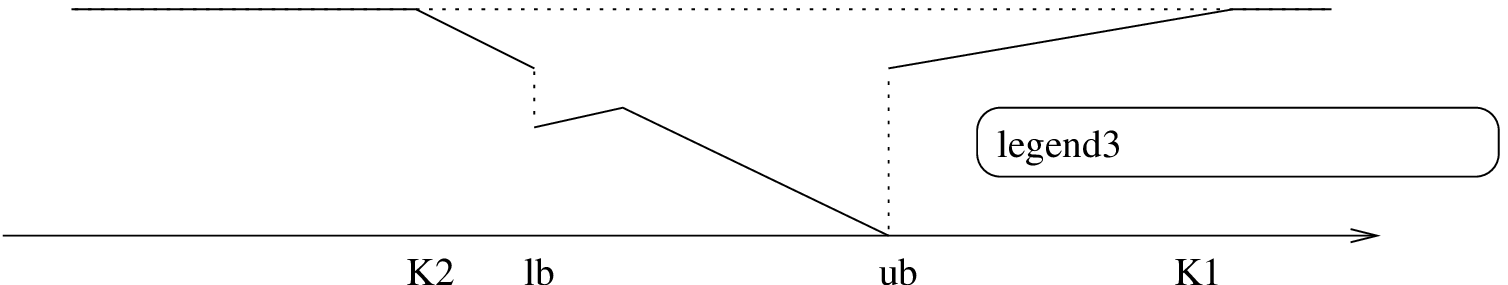}\smallskip

    \includegraphics[width=14cm]{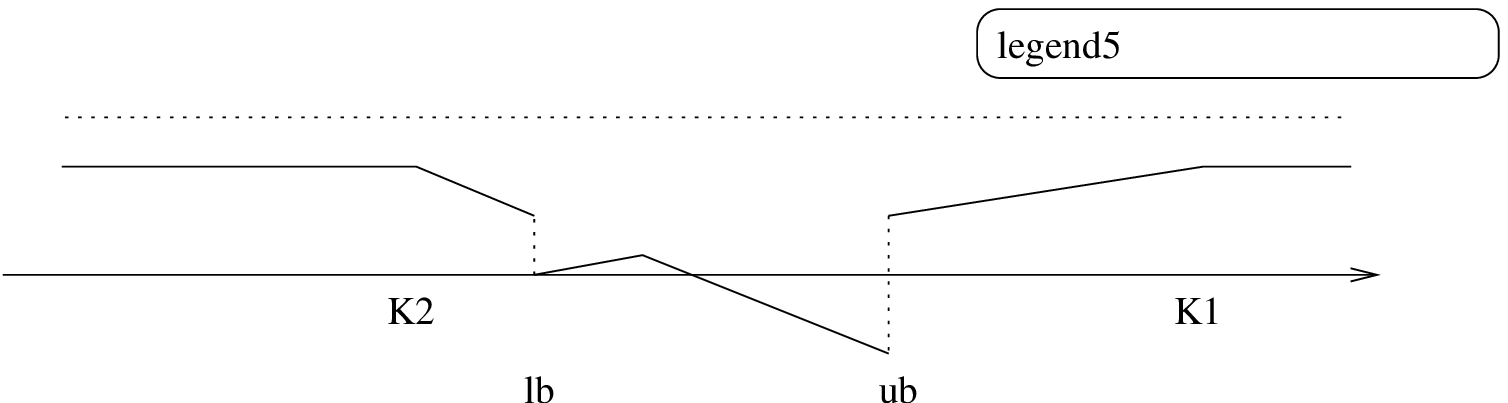}
    \caption{\label{fig:sH2}Subhedge $\lh_{II}$}
  \end{center}\end{figure}

A graphical representation of the construction is given in
Figure~\ref{fig:sH2}. In this case, rather than specifying only $K_1$
and $K_2$, we also need to specify $K_3 \in (\lb,\ub)$ satisfying:
\[
\frac{\ub-K_3}{\ub-K_2} \ge \frac{K_3-\lb}{K_1-\lb}
\]
which implies the function is larger just above $\lb$ than just below
$\ub$. This can be rearranged to get:
\[
K_3 \le \ub \frac{K_1-\lb}{(K_1 - \lb) + (\ub - K_2)} + \lb \frac{\ub
  -K_2}{(K_1 - \lb) + (\ub - K_2)}
\]

The actual inequality we use remains the same as in the previous case
\eqref{eqn:asineqsH1}, as do some of the constraints:
\begin{eqnarray}
  0 & = & \alpha_0 + \alpha_1(\lb-S_0) - \alpha_2(\lb -
  K_2) \label{eq:subh1b}\\
  0 & = & \alpha_0 + \alpha_1(\ub-S_0) -
  \alpha_2(\ub - K_2) + \alpha_3(K_3-\lb) -
  \gamma_1+\gamma_2 \label{eq:subh2b}\\
  1 & = & \alpha_0 + \alpha_1(K_2
  - S_0) + (\alpha_2-\alpha_1)(K_2 -\lb) -
  \alpha_2(K_2-\ub) \label{eq:subh3b}\\
  1 & = & \alpha_0 + \alpha_1(\lb-S_0) + \alpha_2 (K_2 - K_1 + \ub - \lb) + \alpha_3 (K_3 +
  K_1 - \lb - \ub) - \gamma_1 + \gamma_2 \label{eq:subh4b}\\
  \gamma_1 & = & (K_3 - \lb) \alpha_3 \label{eq:subh5b}\\
  \gamma_2 & = & (\ub - K_3) \alpha_3 \label{eq:subh6b}
\end{eqnarray}
\eqref{eq:subh1b} and \eqref{eq:subh2b} refer still to having an
initial payoff of 0 at \ub{} and \lb{}, \eqref{eq:subh5b} and
\eqref{eq:subh6b} also still relate the size of the digital options to
the slopes. The change is in the constraints \eqref{eq:subh3b} and
\eqref{eq:subh4b} which now ensure that the function at $K_1$ and
$K_2$, after hitting first \lb{} and then \ub{}, takes the value 1. We
note that, in the previous example, where \eqref{eq:subh7} held, these
are in fact equivalent to \eqref{eq:subh3} and \eqref{eq:subh4}; the
fact that \eqref{eq:subh7} no longer holds means that we need to be
more specific about the constraints. 

The solutions to the above are now: \renewcommand{\arraystretch}{1.5}
\begin{equation}
  \label{eqn:subh2soln}
  \left\{
    \begin{array}{rcl}
      \alpha_0 & = & \frac{(\ub \lb + S_0 K_3)(K_1 - K_2) - (\lb K_1 + S_0
        \ub) (K_3 - K_2) - (\ub K_2 + S_0 \lb)(K_1 - K_3)}{(\ub- \lb)(K_1 -
        K_3)(\ub - K_2)}\\
      \alpha_1 & = & \frac{K_3 (K_1 - K_2) - \lb (K_1 -
        K_3) - \ub (K_3 - K_2)}{(\ub- \lb)(K_1 - K_3)(\ub - K_2)}\\
      \alpha_2 & = & \frac{1}{\ub - K_2}\\
      \alpha_3 & = & \frac{K_1 - K_2}{(K_1 - K_3)(\ub - K_2)}
    \end{array}
  \right.  \left\{
    \begin{array}{rcl}
      \gamma_1 & = & \frac{(K_3 - \lb) ( K_1 - K_2)}{(\ub - K_2) ( K_1 -
        K_3)}\\
      \gamma_2 & = & \frac{(\ub - K_3) ( K_1 - K_2)}{(\ub - K_2) (
        K_1 - K_3)}
    \end{array}\right. 
\end{equation}
\renewcommand{\arraystretch}{1}

As before, we write $\lh_{II}(K_1,K_2,K_3)$ for the random variable on
the right hand side of \eqref{eqn:asineqsH1} where the constants are
chosen as the solutions to the above equations. 

\emph{$\lh_{III}$: subhedge for $\lb<<S_0<\ub$.}\\ The
third case here is the corresponding version of the above where we
have a large value of $K_3$, specifically,
\[
K_3 \ge \ub \frac{K_1-\lb}{(K_1 - \lb) + (\ub - K_2)} + \lb \frac{\ub
  -K_2}{(K_1 - \lb) + (\ub - K_2)}
\]
and we need to modify equations \eqref{eq:subh3b} and
\eqref{eq:subh4b} appropriately:
\begin{eqnarray*}
  1 & = & \alpha_0 + \alpha_1(\ub - S_0) + (\alpha_3 - \alpha_2) (\ub -
  \lb)\\
  1 & = & \alpha_0 + \alpha_1 (\ub - S_0) + \alpha_2(K_2 - K_1 + \lb -
  \ub) + \alpha_3(K_3 + K_1 - 2 \lb) - \gamma_1 + \gamma_2
\end{eqnarray*}
and the solutions are now:
\renewcommand{\arraystretch}{1.5}
\begin{equation}
  \label{eqn:subh3soln}
  \left\{
    \begin{array}{rcl}
      \alpha_0 & = & \frac{(\ub \lb + S_0 K_3)(K_1 - K_2) - (\lb K_1 + S_0
        \ub) (K_3 - K_2) - (\ub K_2 + S_0 \lb)(K_1 - K_3)}{(\ub- \lb)(K_3 -
        K_2)(K_1-\lb)}\\
      \alpha_1 & = & \frac{K_3 (K_1 - K_2) - \lb (K_1 -
        K_3) - \ub (K_3 - K_2)}{(\ub- \lb)(K_3 - K_2)(K_1 - \lb)}\\
      \alpha_2 & = & \frac{K_1 - K_3}{(K_3 - K_2)(K_1 - \lb)}\\
      \alpha_3 & = & \frac{K_1 - K_2}{(K_3 - K_2)(K_1 - \lb)}
    \end{array}
  \right.  \left\{
    \begin{array}{rcl}
      \gamma_1 & = & \frac{(K_3 - \lb) ( K_1 - K_2)}{(K_3 - K_2)(K_1 -
        \lb)}\\
      \gamma_2 & = & \frac{(\ub - K_3) ( K_1 - K_2)}{(K_3 -
        K_2)(K_1 - \lb)}
    \end{array}\right. 
\end{equation}
\renewcommand{\arraystretch}{1}
As before, we write $\lh_{III}(K_1,K_2,K_3)$ for the random variable on
the right hand side of \eqref{eqn:asineqsH1} where the constants are
chosen as the solutions to the above equations.


\subsection{Pricing}
\label{sec:dt_pricing}
Consider the double touch digital barrier option with the payoff
$\dbp$. As an immediate consequence of the superhedging strategies
described in Section \ref{sec:dt_super} we get an upper bound on the
price of this derivative:
\begin{prop} 
\label{prop:dt_upperbound}
Given the market input \eqref{eq:market_input}, no-arbitrage \eqref{eq:no_arbitrage}
in the class of portfolios $\Lin(\mathcal{X}\cup \{\dbp\})$ implies the following inequality between the prices
\begin{equation}\label{eq:generalbound_dbp}
    \Pc\dbp\leq \inf\big\{\Pc \uh^{I}(K),\Pc\uh^{II}(K'),\Pc
    \uh^{III}(K_1,K_2,K_3,K_4),\Pc\uh^{IV}(K_1,K_4)\big\},
  \end{equation}
  where the infimum is taken over $K>\lb$, $K'<\ub$ and
  $0<K_4<\lb<K_3<K_2<\ub<K_1$, and where
  $\uh^{I},\uh^{II},\uh^{III},\uh^{IV}$ are given by
  \eqref{eq:H1},\eqref{eq:H2},\eqref{eq:H3} and \eqref{eq:H4}
  respectively.
\end{prop}
The purpose of this section is to show that given the law of $S_T$ and
the pair of barriers $\lb,\ub$ we can determine explicitly which
superhedges and with what strikes the infimum on the RHS of
\eqref{eq:generalbound_dbp} is achieved. We present formal criteria but we also use labels, e.g.\ `$\lb<<S_0<\ub$', which provide an intuitive classification. An example is considered later in Figure \ref{fig:types_hedges}. Furthermore, we will show
that we can always construct a market model in which the infimum in \eqref{eq:generalbound_dbp} is the actual price
of the double barrier option and therefore exhibit the model-free
least upper bound for the price of the derivative. 
Subsequently, an analogue reasoning for subhedging and the lower bound is presented.

Let $\mu$ be the market implied law of $S_T$ given by \eqref{eq:market_law}. 
The barycentre of $\mu$ associates to a non-empty Borel set $\Gamma\subset \re$ the mean of $\mu$ over $\Gamma$ via
\begin{equation}\label{eq:barycentre}
  \mu_B(\Gamma)=\frac{\int_\Gamma u\mu(\td u)}{\int_\Gamma \mu(\td u)}. 
\end{equation}
For $w<\lb$ and $z>\ub$ let $\rho_-(w)>\lb$ and $\rho_+(z)<\ub$ be the
unique points such that the intervals $[w,\rho_-(w)]$ and
$[\rho_+(z),z]$ are centered respectively around $\lb$ and $\ub$, that
is
\begin{equation}\label{eq:def_rho}
  \left\{ \begin{array}{ll} \rho_-:[0,\lb]\to[\lb,\infty)&\textrm{
        defined via }
      \mu_B([w,\rho_-(w)])=\lb,\\ \rho_+:[\ub,\infty)\to[0,\ub]&\textrm{
        defined via } \mu_B([\rho_+(z),z])=\ub. 
    \end{array} \right. 
\end{equation}
Note that $\rho_\pm$ are decreasing and well defined as
$\mu_B([0,\infty))=S_0\in (\lb,\ub)$. We need to define two more
functions:
\begin{equation}\label{eq:def_gamma}
  \left\{ \begin{array}{ll} \gamma_+(w)\geq \ub&\textrm{ defined via }
      \mu_B\big([0,w]\cup[\rho_+(\gamma_+(w)),\gamma_+(w)]\big)=\lb,\;
      w\leq \lb,\\
      \gamma_-(z)\leq \lb &\textrm{ defined via }
      \mu_B\big([\gamma_-(z),\rho_-(\gamma_-(z))]\cup[z,\infty)\big)=\ub,\;
      z\geq \ub,\\
    \end{array} \right. 
\end{equation}
so that $\gamma_+(\cdot)$ is increasing, $\gamma_-(\cdot)$ is
decreasing, and:
\[
\gamma_+(w) \downarrow \ub \mbox{ as } w \downarrow 0, \gamma_-(z)
\uparrow \lb \mbox{ as } z \uparrow \infty. 
\]
Note that $\gamma_+$ is defined on $[0,w_0]$ where $w_0=\lb\land
\sup\{w<\lb: \gamma_+(w)<\infty\}$ and similarly $\gamma_-$ is defined
on $[z_0,\infty]$. We are now ready to state our main theorem. 

\begin{theorem}\label{thm:upper_price}
  Let $\mu$ be the law of $S_T$ inferred from the prices of vanillas via \eqref{eq:market_law}
  and consider the double barrier derivative paying $\dbp$ for a fixed pair
  of barriers $\lb<S_0<\ub$.  Then exactly one of the following is true
  \begin{enumerate}
  \item[\fbox{I}] `$\lb<<S_0<\ub$':

    There exists $z_0 > \ub$ such that\footnote{note that here and
      subsequently, we use $\uparrow \infty$ and $\downarrow 0$ as
      meaning only the case where the increasing/decreasing sequence
      is itself finite/strictly positive, so that in \eqref{eqn:gmcond}, we
      strictly mean
      \[
      \gamma_-(z) \to 0 \mbox{ as } z \downarrow z_0, \mbox{ and }
      \gamma_-(z) > 0 \mbox{ for } z > z_0
      \]
    }
    \begin{equation} \label{eqn:gmcond}
      \gamma_-(z) \downarrow 0 \mbox{ as } z \downarrow z_0, \mbox{ and }
      \rho_-(0) \le \ub. 
    \end{equation}
    Then there is a market model in which $\E\dbp=\E
    \uh^{I}(\rho_-(0))=\frac{C(\rho_-(0))}{\rho_-(0)-\lb}$. 

  \item[\fbox{II}] `$\lb<S_0<<\ub$':

    There exists $w_0 < \lb$ such that
    \[
    \gamma_+(w) \uparrow \infty \mbox{ as } w \uparrow w_0, \mbox{ and }
    \rho_+(\infty) \ge \lb. 
    \]
    Then there is a market model in which $\E\dbp=\E
    \uh^{II}(\rho_+(\infty))=\frac{C(\rho_+(\infty))}{\ub-\rho_+(\infty)}$. 

  \item[\fbox{III}] `$\lb<<S_0<<\ub$':

    There exists $0 \leq w_0 \le \lb$ such that $\gamma_-(\gamma_+(w_0)) =
    w_0$ and $\rho_-(w_0) \le \rho_+(\gamma_+(w_0))$.\\ Then there is a
    market model in which
    \begin{equation}
      \begin{split}
        \E\dbp&=\E\uh^{III}\big(\gamma_+(w_0),\rho_+(\gamma_+(w_0)),\rho_-(w_0),w_0\big)\\
        =\alpha_1C\big(&\gamma_+(w_0)\big)+\alpha_2C\big(\rho_+(\gamma_+(w_0))\big)
        +\alpha_3P\big(\rho_-(w_0)\big)+\alpha_4P\big(w_0\big),
      \end{split}
    \end{equation}
    where $\alpha_i$ are given in \eqref{eq:valuesH3}. 
  \item[\fbox{IV}] `$\lb<S_0<\ub$':

    We have $\ub < \rho_-(0), \lb > \rho_+(\infty)$ and $\rho_+(\rho_-(0)) <
    \rho_-(\rho_+(\infty))$.\\ Then there is a market model in which
    \begin{equation}
      \begin{split}
        \E\dbp&=\E\uh^{IV}\big(\rho_-(0),\rho_+(\infty)\big)\\
        &=\alpha_1C\big(\rho_-(0)\big)+\alpha_2P\big(\rho_+(\infty)\big)
        +\alpha_4,
      \end{split}
    \end{equation}
    where $\alpha_i$ are given in \eqref{eq:valuesH4}. 
  \end{enumerate}
\end{theorem}

We present now the analogues of Proposition \ref{prop:dt_upperbound} and Theorem \ref{thm:upper_price} for the subhedging case. Whereas above we find an upper bound on the price of the derivative, in this case we will construct a lower bound. 

\begin{prop}
\label{prop:dt_lowerbound}
Given the market input \eqref{eq:market_input}, no-arbitrage \eqref{eq:no_arbitrage}
in the class of portfolios $\Lin(\mathcal{X}\cup \{\dbp,\indic{S_T >\lb},\indic{S_T \geq\ub}\})$ implies the following inequality between the prices
  \begin{equation}\label{eq:generalbound_dbp2}
    \Pc\dbp\geq \sup\left\{\Pc \lh_{I}(K_1,K_2),\Pc\lh_{II}(K_1,K_2,K_3),\Pc \lh_{III}(K_1,K_2,K_3),0\right\},
  \end{equation}
  where the supremum is taken over $0<K_2<\lb<K_3<\ub<K_1$ and
  $\lh_{I},\lh_{II},\lh_{III}$ are given by \eqref{eqn:asineqsH1} and
  the solutions to the relevant set of equations:
  \eqref{eqn:subh1soln}, \eqref{eqn:subh2soln} and
  \eqref{eqn:subh3soln}.
\end{prop}

Again, an important aspect of \eqref{eq:generalbound_dbp2} is that
we can in fact show that the bound is tight --- that is, given a set
of call prices, there exists a process under which equality is
attained. Recall that under no-arbitrage the prices of digital calls are essentially specified by our market input via $\Pc \indic{S_T\geq\ub}= -C'(K)$.

In order to classify the different states, we make the following
definitions. Let $\mu$ be the law of $S_T$ implied by the call
prices. Fix $\lb < S_0 < \ub$, and, given $v \in [\lb,\ub]$, define:
\begin{eqnarray}
  \psi(v) & = & \inf \left\{ z \in [0, \lb]: \int_{(z,\lb) \cup (v,\ub)} u \,
    \mu(\td u) + \ub \left( \frac{\ub-S_0}{\ub - \lb} - \mu((z,\lb)\cup
      (v,\ub)) \right) = \lb \frac{\ub - S_0}{\ub - \lb} \right.\label{eq:psi_def_sub}\\ 
      &&\qquad \left. \mbox{ and } \mu((z,\lb) \cup (v,\ub)) \le \frac{\ub-S_0}{\ub -
      \lb} \right\}\nonumber \\
      \theta(v) & = & \sup \left\{ z \ge \ub:
    \int_{(\lb,v) \cup (\ub,z)} u \, \mu(\td u) + \lb \left\{
      \frac{S_0-\lb}{\ub - \lb} - \mu((\lb,v) \cup (\ub,z)) \right\} = \ub
    \frac{S_0 - \lb}{\ub - \lb}\right.\label{eq:theta_def_sub}\\ 
    &&\qquad \left. \mbox{ and } \mu((\lb,v) \cup (\ub,z)) \le \frac{S_0- \lb}{\ub - \lb} \right\}\nonumber
\end{eqnarray}
where we use the convention $\sup\{\emptyset\} = -\infty, \inf
\{\emptyset\} = \infty$. 

In particular, the definition of $\psi$ ensures that, on the set where
$\psi(v) \neq \infty$, we can run all the mass initially from
$S_0$ to $\{\lb, \ub\}$ and then embed from $\lb$ to $(\psi(v),\lb)
\cup (v,\ub)$ and a compensating atom at $\ub$ with the remaining
mass. Note further that the functions $\psi$ and $\theta$ are both
decreasing on the sets $\{v \in [\lb,\ub] :\psi(v) < \infty\}$ and
$\{v\in [\lb,\ub] :\theta(v) > -\infty\}$, which are both closed
intervals. Specifically, we will be interested in the region where
both the functions allow for a suitable embedding; define
\begin{eqnarray}
  \uv & = & \min \left\{\sup\{v \in [\lb,\ub] :\psi(v) <
    \infty\}, \sup\{v \in [\lb,\ub]: 
    \theta(v)>-\infty\}\right\},\nonumber\\
  \lv & = & \max \left\{\inf\{v\in [\lb,\ub]:\psi(v) < \infty\}, \inf\{v \in [\lb,\ub]: \theta(v)>-\infty\}\right\},\ \textrm{ and}\label{eqn:vvK_3}\\
  \kappa(v) &=& \ub \frac{\theta(v) - \lb}{\theta(v) - \lb + \ub - \psi(v)} +
  \lb \frac{\ub - \psi(v)}{\theta(v) - \lb + \ub - \psi(v)}\ ,\nonumber
\end{eqnarray}
where $\sup\{\emptyset\}=-\infty$ and $\inf\{\emptyset\}=\infty$.


\begin{theorem}\label{thm:sub_prices}
Let $\mu$ be the law of $S_T$ inferred from the prices of vanillas via \eqref{eq:market_law}
  and consider the double barrier derivative paying $\dbp$ for a fixed pair
  of barriers $\lb<S_0<\ub$, and recall \eqref{eq:psi_def_sub}-\eqref{eqn:vvK_3}. Then exactly one of the following is true
  \begin{enumerate}
  \item[\fbox{I}] `$\lb<S_0<\ub$':\\
    We have $\uv \ge \lv$ and there exists $v_0\in [\lv,\uv]$ such that $\kappa(v_0) = v_0$. Then there exists a market model in which:
    \begin{eqnarray} \label{eqn:subhedge1}
      \E \dbp & = & \E \lh_{I}(\theta(v_0),\psi(v_0))\nonumber\\
      & = & \alpha_0 +
      \alpha_2(C(\theta(v_0)) - C(\psi(v_0))) + \gamma_2 D(\ub) - \gamma_1
      D(\lb) \\
      &&{}+ \alpha_3\left[C(\lb)+C(\ub) -
        C(v_0)-C(\theta(v_0))\right]\nonumber
    \end{eqnarray}
    where $D(x)$ is the price of a digital option with payoff $\indic{S_T
      \ge x}$, and the values of $\alpha_0, \alpha_2, \alpha_3, \gamma_1,
    \gamma_2$ are given by \eqref{eqn:subh1soln}. 
  \item[\fbox{II}] `$\lb<S_0<<\ub$':\\
  We have $\uv \ge \lv$ and $\uv < \kappa(\uv)$. Then
    there exists a market model in which:
    \begin{eqnarray}\label{eqn:subhedge2}
      \E \dbp & = & \E \lh_{II}(\theta(\uv),\psi(\uv),\uv)\nonumber\\
       & = &
      \alpha_0 + \alpha_2(C(\theta(\uv)) - C(\psi(\uv))) + \gamma_2 D(\ub) -
      \gamma_1 D(\lb)\\ &&{}+ \alpha_3\left[C(\lb)+C(\ub) -
        C(\uv)-C(\theta(\uv))\right]\nonumber
    \end{eqnarray}
    where $D(x)$ is the price of a digital option with payoff $\indic{S_T
      \ge x}$, and the values of $\alpha_0, \alpha_2, \alpha_3, \gamma_1,
    \gamma_2$ are given by \eqref{eqn:subh2soln}. 
  \item[\fbox{III}] `$\lb<<S_0<\ub$':\\
    We have $\uv \ge \lv$ and $\lv > \kappa(\lv)$. Then there exists a market model in which:
    \begin{eqnarray*}
      \E \dbp & = & \E \lh_{III}(\theta(\lv),\psi(\lv),\lv)\nonumber\\ & = &
      \alpha_0 + \alpha_2(C(\theta(\lv)) - C(\psi(\lv))) + \gamma_2 D(\ub) -
      \gamma_1 D(\lb)\\ &&{}+ \alpha_3\left[C(\lb)+C(\ub) -
        C(\lv)-C(\theta(\lv))\right]\nonumber
    \end{eqnarray*}
    where $D(x)$ is the price of a digital option with payoff $\indic{S_T
      \ge x}$, and the values of $\alpha_0, \alpha_2, \alpha_3, \gamma_1,
    \gamma_2$ are given by \eqref{eqn:subh3soln}. 
  \item[\fbox{IV}] `$\lb<<S_0<<\ub$':\\
  We have $\uv<\lv$. Then there exists a market model in which $\E \dbp =0$.
  \end{enumerate}
  Furthermore, in cases $\fbox{I}-\fbox{III}$ we have $\lv =
  \inf\{v \in [\lb,\ub]:\psi(v) < \infty\} \le \sup\{v \in [\lb,\ub]:
  \theta(v)>-\infty\}= \uv$.
\end{theorem}

\section{Applications and Practical Considerations}
\label{sec:extra}

\subsection{Finitely many strikes}
\label{sec:extra_strikes}

One important practical aspect where reality differs from the
theoretical situation described above concerns the availability of
calls with arbitrary strikes. Generally, calls will only trade at a
finite set of strikes, $0 = x_0 \le x_1 \le \ldots \le x_N$ (with $x_0
= 0$ corresponding to the asset itself). It is then natural to ask:
how does this affect the hedging strategies introduced above? In full
generality, this question results in a rather large number of
`special' cases that need to be considered separately (for example,
the case where no strikes are traded above $\ub$, or the case where
there are no strikes traded with $\lb < K < \ub$). In addition, there
are differing cases, dependent on whether the digital options at $\lb$
and $\ub$ are traded. Consequently, we will not attempt to give a
complete answer to this question, but we will consider only the cases
where there are `comparatively many' traded strikes, and assume that
digital calls are not available to trade. Furthermore, we will apply
the theorems of previous sections to measures with atoms. It should be
clear how to do this, but a formal treatment would be rather lengthy
and tedious, with some extra care needed when the atoms are at the
barriers.  For that reason we only state the results of this section
informally.

Mathematically, the presence of atoms in the measure $\mu$ means that
the call prices are no longer twice differentiable. The function is
still convex, but we now have possibly differing left and right
derivatives for the function. The implication for the call prices is
the following:
\[
\mu([x,\infty)) = -C'_-(K), \mbox{ and } \mu((x,\infty)) = -C'_+(K).
\]
In particular, atoms of $\mu$ will correspond to `kinks' in the call
prices.

The first remark to make in the finite-strike case is that, if we
replace the supremum/infimum over strikes that appear in expressions
such as \eqref{eq:generalbound_dbp} and \eqref{eq:generalbound_dbp2}
by the supremum/infimum over traded strikes, then the arguments that
conclude that these are lower/upper bounds on the price are still
valid. The argument only breaks down when we wish to show that these
are the best possible bounds. To try to replace the latter, we now
need to consider which models might be possible under the given call
prices. Our approach will be based on the following type of argument:
\begin{enumerate}
\item suppose that using only calls and puts with traded strikes we
  may construct $\tilde{H}^i$, for $i \in \{I, \ldots, IV\}$, such
  that $\tilde{H}^i \ge \uH^i$ as a function of $S_T$;
\item suppose further that we can find an admissible call price
  function $C(K)$, $K>0$, which agrees with the traded prices, and
  such that in the market model
  $(\Omega,\mathcal{F},(\mathcal{F}_t),\Pr^* )$ associated by Theorem
  \ref{thm:upper_price} with the upper bound
  \eqref{eq:generalbound_dbp} we have $\tilde{H}^i(S_T) = \uH^i(S_T)$,
  $\Pr^* $-almost surely;
\end{enumerate}
then the smallest upper bound on the price of a digital double touch
barrier option is the cost of the cheapest portfolio
$\tilde{H}^i$. This is fairly easy to see: clearly the price is an
upper bound on the price of the option, since $\tilde{H}^i$
superhedges, and under $\Pr^*$ this upper bound is attained. Indeed,
by assumption on $\tilde{H}^i$, in the market model associated with
$\Pr^*$ we have $\Pc \tilde{H}^i=\E^* \tilde{H}^i=\E^*
\uH^{i}=\Pc\uH^{i}$.  Consequently, by Theorem~\ref{thm:upper_price},
the price of the traded portfolio $\tilde{H}^{i}$ and the price of the
digital double touch barrier option are equal under the market model
$\Pr^*$.  Note that in (ii) above it is in fact enough to have
$\tilde{H}^i(S_T) = \uH^i(S_T)$ just for the $\uH^i$ which attains
equality in \eqref{eq:generalbound_dbp}.

\begin{figure}
  \psfrag{xj}{$x_j$}
  \psfrag{xj2}{$x_{j+1}$}
  \centering
  \includegraphics{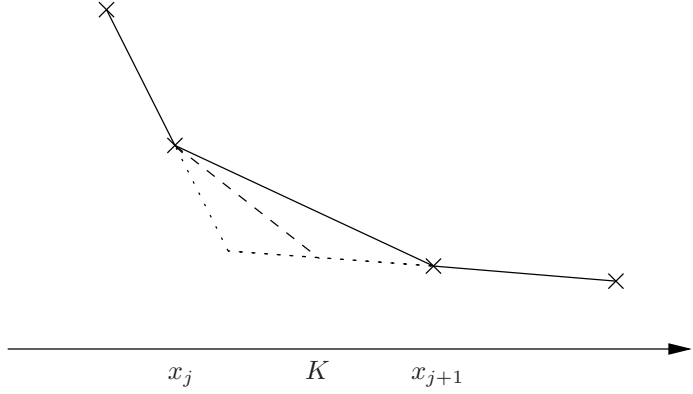}
  \caption{Possible call price surfaces as a function of the strike. 
    The crosses indicate the prices of traded calls, the solid line
    corresponds to the upper bound, which corresponds to placing all
    possible mass at traded strikes. The lower bound is indicated by
    the dotted line, and the dashed line indicates the surface we will
    choose when we wish to minimise the call price at $K$. In this
    case, we note that there will be mass at $K$ and $x_j$, but not at
    $x_{j+1}$. There is a second case where $K$ is below the kink in
    the dotted line, when the resulting surface would place mass at
    $K$ and $x_{j+1}$, but not at $x_j$.} 
\end{figure}

We now wish to understand the possible models that might
correspond to a given set of call prices $\left\{C(x_i); 0 \le i \le
  N\right\}$. Simple arbitrage constraints (see \eg{}
\cite[Theorem~3.1]{DavisHobson:06}\footnote{We suppose also that our
  call prices do not exhibit what is termed here as weak arbitrage})
require that the call prices at other strikes (if traded) imply that
the function $C(K)$ is convex and decreasing. This allows
us to deduce that, for $K$ such that $x_j < K < x_{j+1}$ for some $j$,
we must have
\begin{eqnarray}
  C(K) & \le & C(x_j) \frac{x_{j+1}-K}{x_{j+1} - x_j} +
  C(x_{j+1})\frac{K-x_j}{x_{j+1} - x_j} \label{eq:upper1} \\
  C(K) & \ge & C(x_{j}) + \frac{C(x_{j}) - C(x_{j-1})}{x_{j} -
    x_{j-1}}(K-x_{j}) \label{eq:lower1}\\
  C(K) & \ge & C(x_{j+1}) - \frac{C(x_{j+2}) - C(x_{j+1})}{x_{j+2} -
    x_{j+1}}(x_{j+1}-K)\label{eq:lower2}
\end{eqnarray}
These inequalities therefore provide upper and lower bounds on the
call price at strike $K$, and it can be seen that the upper
bound and lower bound are tight by choosing suitable models: in the
upper bound, the corresponding model places all mass of the law of
$S_T$ at the strikes $x_i$; in the lower bounds, the larger of the
two possible terms can be attained with a law that places mass at
$K$, and at other $x_i$'s except, in \eqref{eq:lower1}, at $x_{j}$,
and in \eqref{eq:lower2}, at $x_{j+1}$. Moreover, provided that there
are sufficiently many traded calls between two strikes $K, K'$, we can
(for example) choose a law that attains the maximum possible call
price at $K$, and the minimum possible price at $K'$. We will assume
that this property holds for all `relevant' points in the sequel. 

\begin{figure}
  \psfrag{xj}{$x_j$}
  \psfrag{xj2}{$x_{j+1}$}
  \centering
  \includegraphics{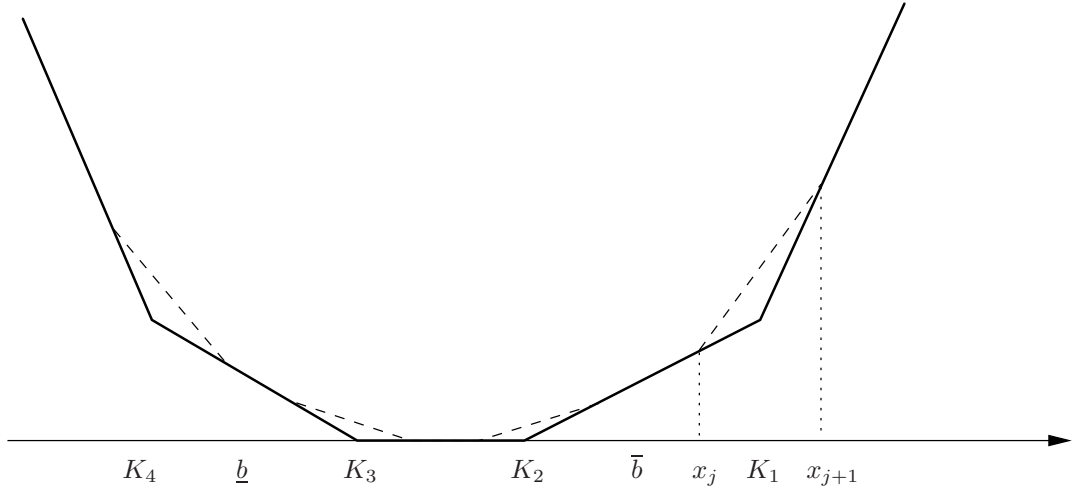}
  \caption{\label{fig:uH3finstrike} An optimal superhedge $\uH^{III}$ in the case
    where only finitely many strikes are traded. The lower (solid)
    payoff denotes the optimal construction under the chosen extension
    of call prices to all strikes, and the upper (dashed) payoff denotes the
    payoff actually constructed. Note that, for example, $x_j$ is the
    largest traded strike below $K_1$, and $x_{j+1}$ is the smallest
    traded strike greater than $K_1$.} 
\end{figure}

Consider firstly the case where we wish to superhedge the double touch
option. We will consider the model $\Pr^*$ which corresponds (through
Theorem~\ref{thm:upper_price}) to the call prices obtained by linearly
interpolating the prices at $x_0, \ldots, x_n$ --- in particular,
$C(K)$ is the maximal possible value it may take under the assumption
of no arbitrage. The key idea is to consider the portfolio depicted in
Figure~\ref{fig:uH3finstrike}, where the smaller payoff is the optimal
payoff under the assumption that all strikes may be traded, and the
upper payoff is one that may be constructed using only strikes that
are actually traded. Then although the upper payoff is strictly larger
between $x_j$ and $x_{j+1}$ say, where $x_j < K_1 < x_{j+1}$, the
points at which this occurs are not points at which the asset will
finish, since under $\Pr^*$, the law of $S_T$ is supported only by the
points $x_i$. Consequently, both the upper and lower payoff have
almost surely the same payoff under $\Pr^*$, and therefore the same
expectation and price --- in particular, the upper portfolio is a
superhedging portfolio, which has the same price as the lower
portfolio, which is the smallest upper bound for the double touch
under $\Pr^*$, and this payoff is in turn equal to the payoff of the
double touch under $\Pr^*$. Since we have a superhedge for all models,
and a model under which the superhedge is a hedge, we must have the
least upper bound.  We note additionally that the same choice of
$C(K)$ and the same $\Pr^*$ will work in a similar manner for the
other hedges $\uh^I, \uh^{II}$ and $\uh^{IV}$.

Consider now the lower bound. To keep things simple we begin by
altering slightly the problem: rather than the payoff $\dbp$, we
consider a subhedge of the option with payoff $\dbpnonstrict$. Then
\eqref{eqn:asineqsH1} still holds\footnote{Technically, for
  \eqref{eqn:asineqsH1} to still hold, we actually need to modify the
  hitting times so that we consider the entrance times of
  e.g. $(0,\lb)$ rather than the hitting time of $\lb$, and be able to
  trade forward with strike $\lb$ at this stopping time. In practice,
  this will not be crucial, since for example continuity of the asset
  price means that this entrance time may be suitably approximated,
  and it will commonly be the case that the two stopping times are in
  fact equal.}  with the new term $\dbpnonstrict$ on the left-hand
side provided we modify the digital options on the right-hand side to
$\mathbf{1}_{\{S_T \ge \lb\}}$ and $\mathbf{1}_{\{S_T > \ub\}}$. So we
may still consider the optimal portfolio (in all three cases) as being
short a collection of calls at strikes $K_1, K_2$ and $K_3$, long
calls at $\lb$ and $\ub$, and holding digital options at each of these
points. Intuitively, we should look for a model which will maximise
the cost of the calls at $K_1, K_2$ and $K_3$, and minimise the cost
at $\ub$ and $\lb$, as well as maximising the cost of a digital call
at $\lb$ and minimising the cost of the digital call at $\ub$. The
former conditions correspond to choosing the call prices which give
the upper bound \eqref{eq:upper1} --- so we choose the call price
which linearly interpolates $C(x_i)$ except when $x_i \le \lb \le
x_{i+1}$, and $x_i \le \ub \le x_{i+1}$. In the latter cases, we wish
to minimise the call price, so we choose the prices corresponding to
the appropriate lower bound \eqref{eq:lower1} or \eqref{eq:lower2},
which have a kink at $\lb$ and at (exactly) one of its two adjacent
traded strikes and likewise for $\ub$. We note that the prices of the
digital calls (which are either minus the left gradient or minus the
right gradient of the call prices at the barrier) will also now be
optimised when they trade in exactly the forms specified above (that
is, the digital call at $\lb$ only pays out if the asset is greater
than or equal to $\lb$, while the call at $\ub$ will pay out if the
asset is strictly larger than $\ub$ at maturity).

\begin{figure}
  \includegraphics{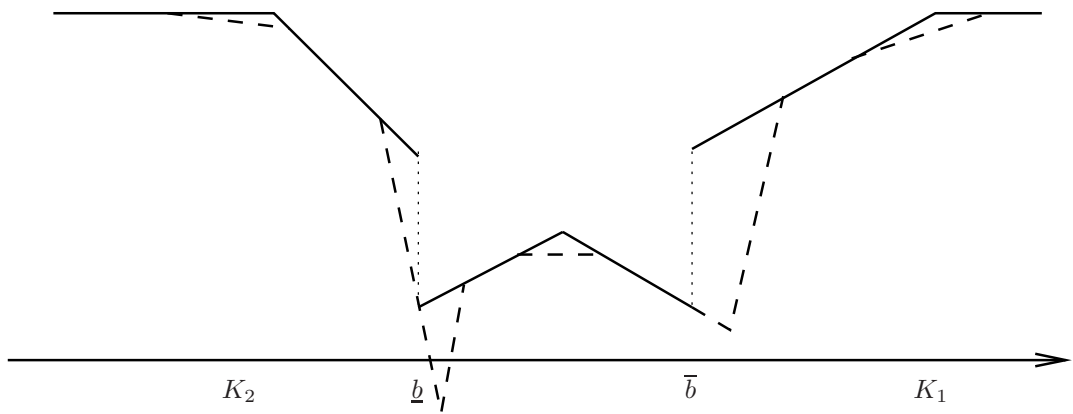}
  \centering
  \caption{\label{fig:lH1finstrike} An optimal subhedge in the case
    where only finitely many strikes are traded. The upper (solid)
    payoff denotes the optimal construction under the chosen extension
    of call prices to all strikes, and the lower (dashed) payoff
    denotes the payoff actually constructed. Observe that there are
    two possible constructions at the discontinuity --- the choice of
    which type of construction is optimal will depend on whether the
    optimal lower bound on the call price at $\lb$ or $\ub$
    corresponds to the bound \eqref{eq:lower1} or \eqref{eq:lower2}}
\end{figure}

The above procedure specifies uniquely a complete set of call prices
$C(K)$, which match the market input and which are our candidate for
the smallest lower bound among the possible models. Then we note that
a construction similar to that given in Figure~\ref{fig:lH1finstrike}
will work --- the main difference to the superhedge case is that, at
the discontinuity, there are two possible cases that need to be
considered, and the optimal subhedge will depend on behaviour of
$C(K)$ at the strikes adjacent to $\lb$ and $\ub$. More precisely, the
portfolio given by the dotted line in Figure~\ref{fig:lH1finstrike}
corresponds to the case when $C(K)$ has no kinks at the traded strikes
to the immediate right of both barriers (i.e.~$S_T$ has no atoms at
these strikes). The other three possibilities are straightforward
modifications. The argument then proceeds as above: in the model given
by Theorem \ref{thm:sub_prices} subhedge $\lH_i$ achieves equality in
\eqref{eq:generalbound_dbp2} (modified to account for the changes from
inequalities to strict inequalities and {\it vice versa}) and the
portfolio constructed in Figure~\ref{fig:lH1finstrike} is
almost-surely equal\footnote{Assuming that the optimal $K_3$ chosen is
  separated from $\lb$ and $\ub$ by a traded call price} to $\lH_i$,
so that they must have the same price. The resulting subhedge,
constructed using only calls and puts with traded strikes, is
therefore a hedge under the chosen model, and therefore the optimal
lower bound.

Note, however, the behaviour of our construction at the barriers: the
optimal model has atoms at the barriers which will be embedded by mass
which has already hit the other barrier. With the modified payoff,
$\dbpnonstrict$, this will still give equality in the subhedge, but
this would not have been the case with the original model. One could
now consider the option $\dbp$ by approximating with a payoff of the
form $\mathbf{1}_{\overline{S}_T> \ub-\eps,\, \underline{S}_T<
  \lb+\eps}$, which would (in the optimal embedding) place atoms `just
inside the barriers' --- it is this behaviour of the `optimal'
construction that required us to consider the modified payoff in this
case.

\subsection{Jumps in the underlying}
\label{sec:extra_jumps}
Throughout the paper, we have assumed that $(S_t)_{t \ge 0}$ has
continuous paths. In fact we can relax this assumption
considerably. First of all it is relatively simple to see that if we
only assume that barriers $\lb,\ub$ are crossed in a continuous manner
then all of our results remain true. Secondly, if we make no
continuity assumptions\footnote{We always suppose the processes have
  c\`adl\`ag paths.} then all our
superhedges still work -- jumping over the barrier only makes the
appropriate forward transaction more profitable. In contrast, our
subhedges do not work and lower bounds on prices are trivially zero
and are attained in the model where $S_t=S_0$ for all $t<T$ and then
it jumps to the final position $S_T$.

\subsection{Hedging comparisons}
\label{sec:extra_numerics}

In practice, one would expect that the prices we derive as upper and
lower bounds using the techniques of this paper will be rather wide,
and well outside typical bid/ask spreads. Consequently, the use of
these techniques as a method for pricing is unlikely to be successful.
However the technique also provides superhedging and subhedging
strategies that may be helpful. Consider a trader who has sold a
double barrier option (at a price determined by some model perhaps),
and who wishes to hedge the resulting risk. In a Black-Scholes world,
the trader could remove the risk from his position by delta-hedging
the short position. However, there are a number of practical
considerations that would interfere with such an approach:
\begin{description}
\item[Discrete Hedging:] A notable source of errors in the hedge will
  be the fact that the hedging portfolio cannot be continuously
  adjusted; rather the delta of the position might be adjusted on a
  periodic basis, resulting in an inexact hedge of the position. While
  including a gamma hedge could improve this, the hedge will never be
  perfect. In addition, there is an organisational cost (and risk) to
  setting up such a hedging operation that might be important.
\item[Transaction Costs:] A second consideration is that each trade
  will incur a certain level of transaction costs. These might be tiny for delta-hedging but
  not so anymore for vega-hedging. To minimise the
  total transaction costs, the trader would like to be able to trade
  as infrequently as possible. Of course, this means that there will
  be a necessary trade-off with the discretisation errors incurred
  above.
\item[Model Risk] The final concern for the trader would be: am I
  hedging with the correct model? Using an incorrect model will of
  course result in systematic hedging errors due to \eg{} incorrectly
  estimated volatility, but could also lead to large losses should the
  model fail to incorporate structural effects such as jumps. A
  delta-hedge would typically be improved with a vega-hedge which in
  turn raises the issue of transaction costs as mentioned above.
\end{description}
It would appear that the constructions developed in Section
\ref{sec:dt} may be able to address some of these
issues: there is no need for regular recalculation of the Greeks of
the position, although the breaching of the barrier still requires
monitoring; since there are only a small number of transactions, it
seems likely that the transaction costs may be reduced; our hedge has
been derived using model-free techniques, so that we will still be
hedged even if the market does not behave according to our initial
model, and behaviour such as jumps (at least for the upper bound of
the double touch) will not affect this.

A further consideration that is likely to be of importance to a hedger
is the likely distribution of the returns. Under the hedging strategy
suggested above, so that the trader has sold the option using the
`correct' price (plus a small profit), and set up the superhedging
strategy suggested at a higher price, on average the trader will come
out even, as he will if he delta hedges. His comparison between the
approaches would then come down to the respective risk involved in the
different hedges. For a delta/vega hedge, this is typically symmetric about
zero, however the superhedge will be very asymmetric as it is bounded below. One might
expect that a large number of paths will hedge `correctly,' resulting
in a loss close to the difference between the price the option was sold at,
and the price that the hedge was bought at, and the remaining paths
will do better, when the superhedge strictly dominates. If the trader
is particularly worried about the possible tail of his trading losses
as a measure of risk, this strict cut-off could be very
advantageous. The delta/vega hedge, on the other hand, has the appeal of having a lower variance
of hedging errors.

Of course a variety of such strategies (typically known as static,
semi-static or robust) have been suggested in the literature, under a
variety of more or less restrictive assumptions on the price process,
and mostly for single barrier options, and variants such as knock-out
calls. We have already mentioned the paper \cite{Brown:01b} which makes very
limited restrictions on the underlying price process. More restrictive
is the work of \cite{BowieCarr:94}, and subsequent papers
\cite{CarrChou:97,CarrEllisGupta:98}. Here the authors assume that the
volatility satisfies a symmetry assumption, and as a consequence, one
can for example hedge a knock-out call with the barrier above the
strike by holding the vanilla call, and being short a call at a
certain strike above the barrier. By the assumption on the volatility,
whenever the underlying hits the barrier, both calls have the same
value, and the position may be closed out for zero value. A related
technique is due to \cite{DermanErgenerKani:95}, and followed up by
\cite{AndersenAndreasenEliezer:02} and \cite{Fink:03}. The idea here
is to use other traded options to make the value of the hedging
portfolio equal to zero along the barrier when liquidated. In the
simplest form, a portfolio of calls above the barrier at different
strikes and/or maturities are purchased so that the portfolio value at
selected times before maturity is zero. Extensions allow this idea to
be used for stochastic volatility, and even to cover jumps, at the
expense of needing possibly a very large portfolio of options. More
recently, work of \cite{Nalholm:06} unifies both these approaches, and
allows a fairly general set of asset dynamics, as does
\cite{GieseMaruhn:07}, where the authors find an optimal portfolio
by setting up an optimisation problem. Note however that all
these strategies assume a known model for the underlying, and also
that the hedging assets will be liquid enough for the portfolio to be
liquidated at the price specified under the model. In addition, since
some of the hedging portfolios can involve a large number of options,
it is not clear that the static hedges here will be efficient at
resolving the issues of transaction costs (and operational
simplicity) or model risk. A numerical investigation of the
performance of a range of these options in practice has been conducted
in \cite{Engelmann:07}, and similar investigations for a different
class of static hedges appear in
\cite{DavisSchachermayerTompkins:01b,Tompkins:97}.

There are also more classical, theoretical approaches to the problem
of hedging where the problem is considered in an incomplete market
(without which, of course, a perfect hedge would be possible). In this
situation (for example \cite{FollmerSchweizer:90}) one wants to solve
an optimal control problem where the aim is to minimise the `risk' of
the hedging error, where `risk' is interpreted suitably (perhaps with
regards to a utility function or a risk measure). One can further
modify the approach to restrict the class of trading strategies
available (\eg{} \cite{CetinJarrowProtter:04}). More recently, in a
combination of the static and dynamic approaches,
\cite{IlhanJonssonSircar:06} have considered the problem of risk
minimisation over an initial static portfolio and a dynamic trading
strategy in the underlying.

The most notable difference between the studies described above, and the
ideas of the previous sections is that we make very little
modelling assumptions on the underlying, whereas most of the approaches
listed require a single model to be specified, with respect to which
the results will then be optimal. In particular, these techniques are
unable to say anything about hedging losses should the assumed model
not actually be correct.

Of course, the criteria under which we have constructed our hedges ---
that they are the smallest model free superhedging strategy, or the
greatest model-free subhedging strategy --- do not necessarily mean
that the behaviour of the hedges in `normal' circumstances will be
particularly suitable: we would expect that the hedge would perform
best in extreme market conditions, however in order for it to be
suitable as a hedge against model risk, one would also want the
performance of the hedge to generally be reasonable. To see how this
strategy compares, we will now consider some Monte Carlo based
comparisons with the standard delta/vega-hedging techniques. 
We only look at the double touch options treated in Section \ref{sec:dt}. The comparisons will take the following form:
\begin{enumerate}
\item We choose the Heston model for the `true' underlying asset, and compute
  the time-0 call prices under this model at a range of strike prices,
  and the time-0 price of a double barrier option;
\item We compute the optimal super- and sub- hedges for the digital double touch barrier option based on the observed call prices, and suppose that the hedger purchases these portfolios using the cash
  received from the buyer (and borrowing/investing the difference between the portfolios);
\item For comparison purposes, we also hedge the option using a
  suitable delta/vega hedge with daily updating (for comparison purposes,
  the hitting of the barrier in both cases is also monitored on a
  daily basis).
\end{enumerate}

In the numerical examples, we assume that the underlying process is
the Heston stochastic volatility model (\cite{Heston:93})
\begin{equation}
\label{eq:heston}
\left\{
  \begin{array}{rclcl}
    dS_t&=&\sqrt{v_t}S_tdW^1_t, & & S_0=S_0,\ v_0=\sigma_0\\
    dv_t&=&\kappa(\theta-v_t)dt+\xi\sqrt{v_t}dW^2_t, &&
    d\langle W^1,W^2\rangle_t=\rho dt,
  \end{array}
\right.
\end{equation}
with parameters 
\begin{equation}
\label{eq:heston_par}
S_0=100,\ \sigma_0=0.5,\ \kappa=0.6,\ \theta=1,\ \xi=1.3\textrm{ and }\rho=0.15\ .
\end{equation}
Transactions in $S_t$ carry a $0.5\%$ transaction cost and buying or
selling call/put options carries a $1\%$ transaction cost. The
delta/vega hedge is constructed using the Black-Scholes delta of the
option, but using the at-the-money implied volatility assuming that
the call prices are correct (\ie{}\ they follow the Heston
model). While not perfect as a hedge, empirical evidence
\cite{DumasFlemingWhaley:98} or \cite{EngelmannFenglerSchwendner:06}
suggests that the hedge is reasonable even without the vega component,
although it is also the case that more sophisticated methods should
result in an improvement of this benchmark.

We consider a short and a long position in a digital double touch
barrier option with payoff $\dbp$ for $\ub=117$ and $\lb=83$ and
compare hedging performance of our quasi-static super- and sub- hedges
and the standard delta/vega hedges running 40000 Monte Carlo
simulations\footnote{The resulting hedging errors were mean-adjusted
  as in Tompkins \cite{Tompkins:97} for consistency. The adjustments
  are of order 0.001 and have no qualitative influence on our
  results.}.

The cumulative distributions of hedging errors are given in the upper graphs in Figure \ref{fig:hegerr_117_83}.
\setlength{\columnsep}{0cm}
\begin{figure}[htbp]
\begin{multicols}{2}{
\includegraphics[height=6.5cm]{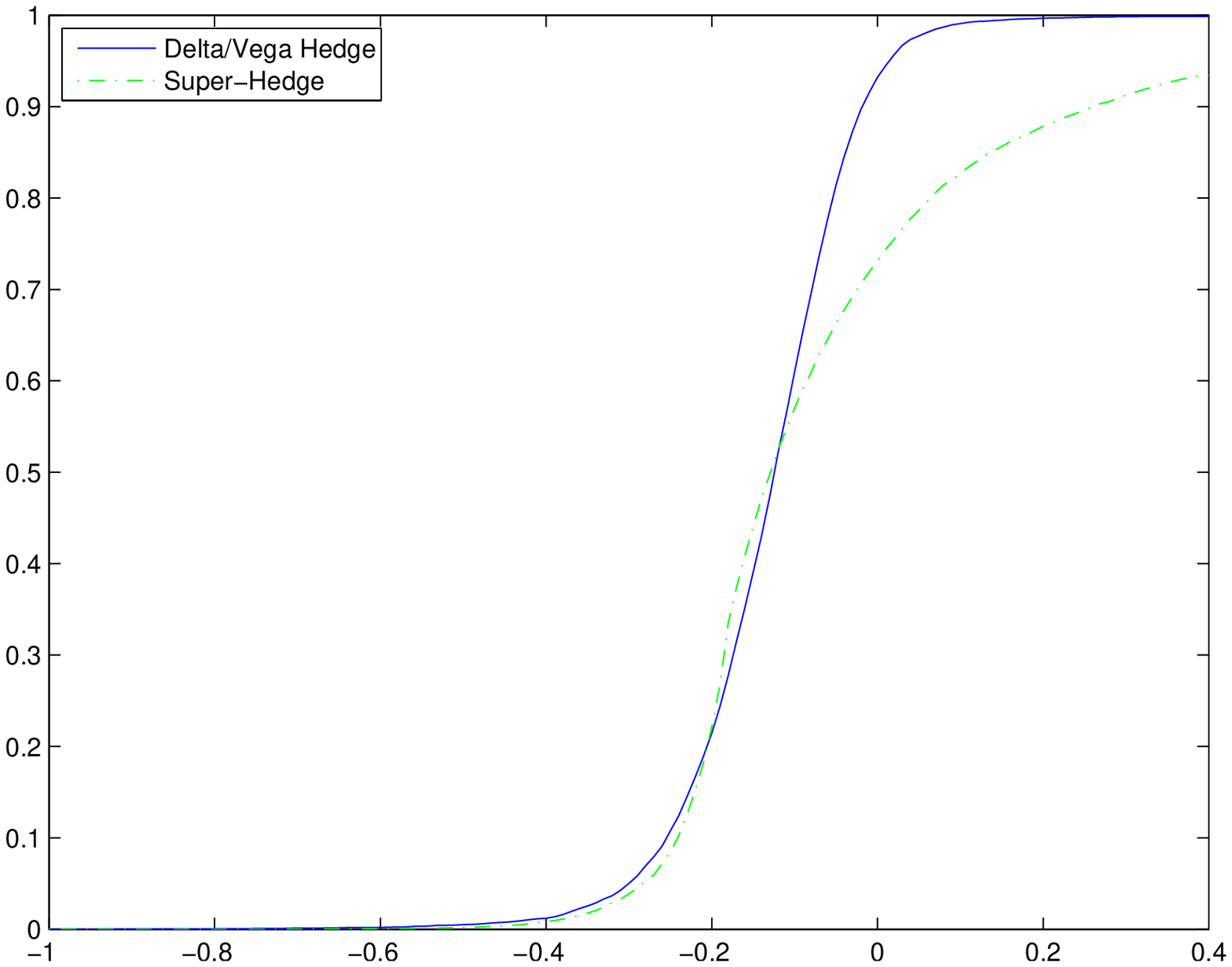}

\includegraphics[height=6.5cm]{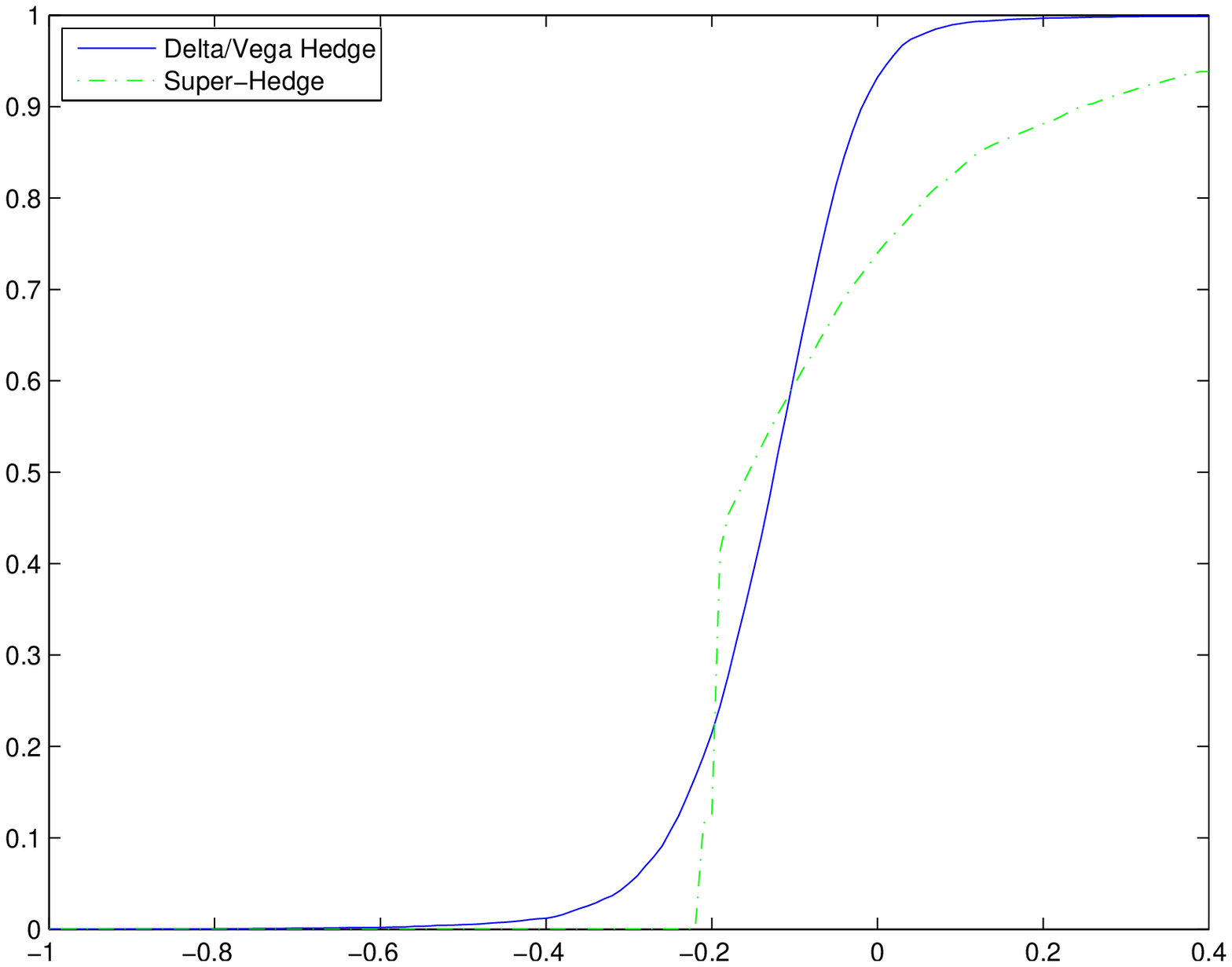}

\includegraphics[height=6.5cm]{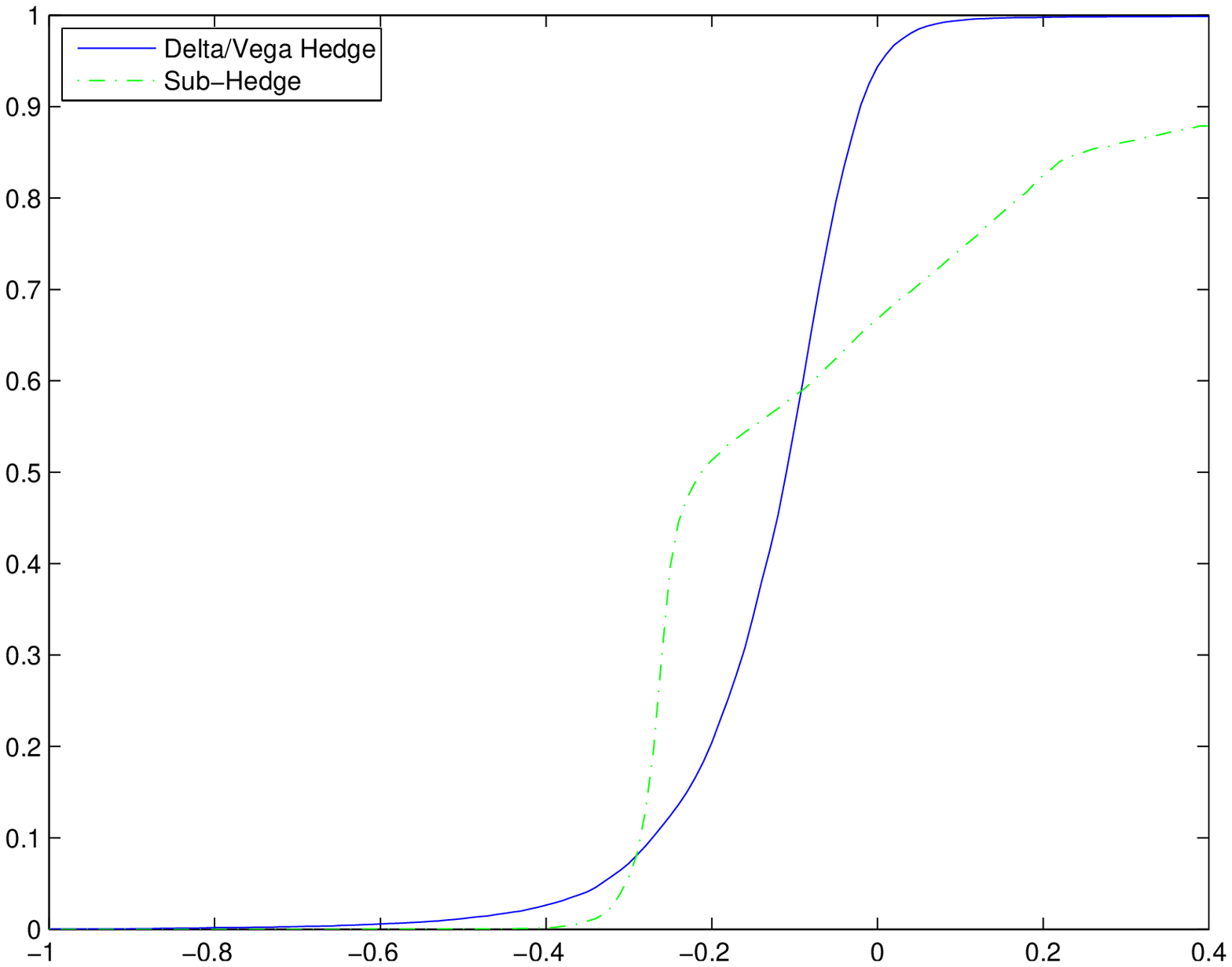}

\includegraphics[height=6.5cm]{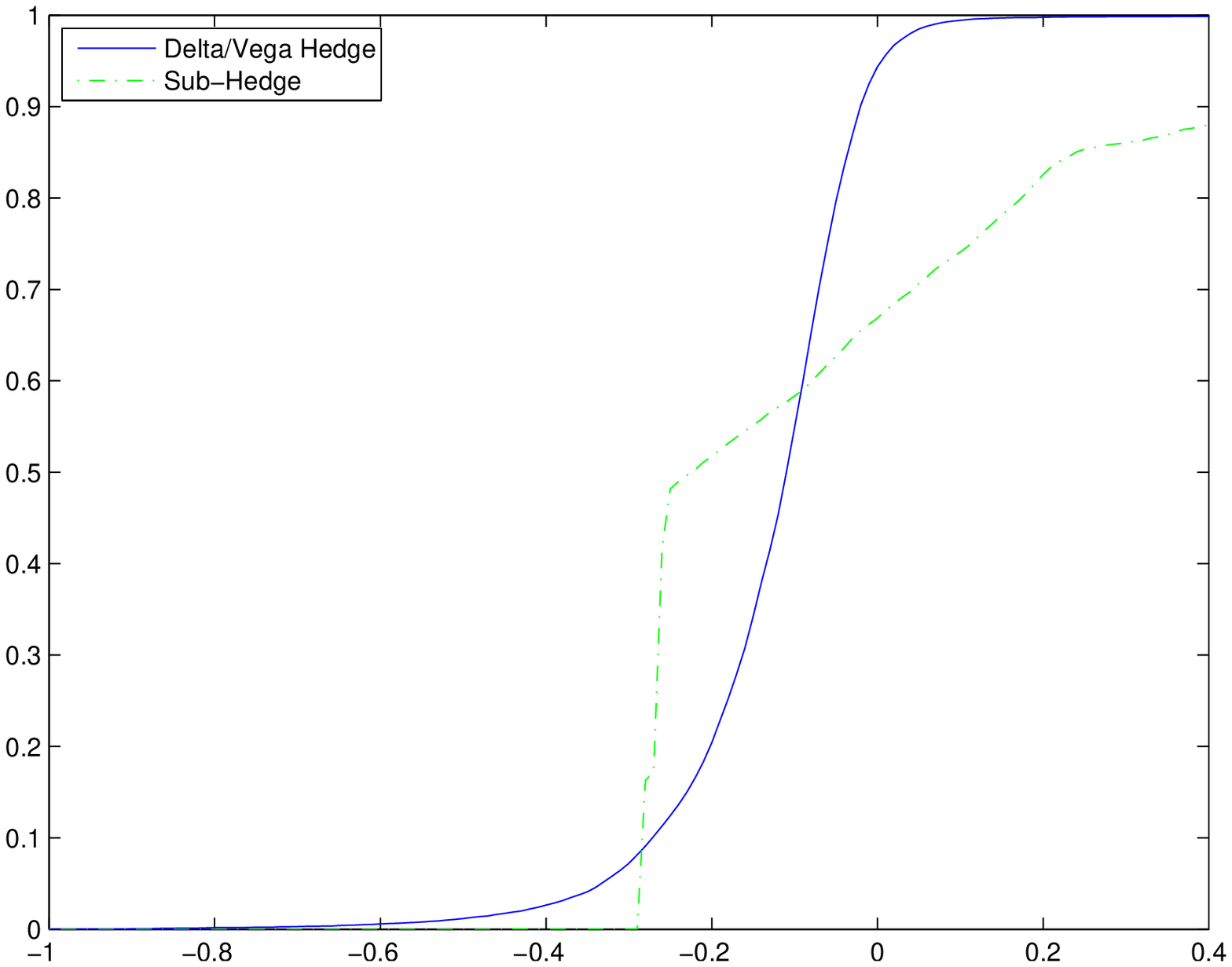}
}
\end{multicols}

\caption{Cumulative distributions of hedging errors under different scenarios of a short position
(\emph{left}) and a long position (\emph{right}) in a double touch option with barriers at $117$ and $83$ under Heston model \eqref{eq:heston}--\eqref{eq:heston_par}. In the lower graphs exact monitoring of barrier crossings was allowed.}
\label{fig:hegerr_117_83}
\end{figure}
Quasi-static super- and sub- hedges introduced in this paper also
incur large losses -- this is due to barrier crossing being monitored
daily. However, a closer inspection reveals that this comparatively
large loses are less frequent then when using delta/vega hedging
strategy. Indeed, in Table \ref{tab:exp_util} we show that an agent
with an exponential utility $U(x)=1-\exp(-x)$ would prefer the error
distribution of our hedges to that of the delta/vega hedge. We could
also ask what happens if we allow our hedges to monitor the barrier
crossings exactly.  The corresponding cumulative distributions of
hedging errors are given in the lower graphs in Figure
\ref{fig:hegerr_117_83} -- we clearly see how the losses are bounded
below. It is interesting however to note that in terms of utility of
hedging errors (cf.\ Table \ref{tab:exp_util}) this doesn't really
change the performance of our hedges - it eliminates some extremely
rare larger losses but at the cost of frequent small profits (cf.\
Section \ref{sec:extra_jumps}).

Finally, we investigate how the performance of our hedges compares if
we vary the barrier levels.  Appropriate cumulative distributions (for 20000 MC runs) are
reported in Figures \ref{fig:hegerr_130_70}-\ref{fig:hegerr_105_80} in Section \ref{sec:figures}
and exponential utilities are reported in Table
\ref{tab:exp_util}. Our superhedge consistently outperforms
delta/vega hedging. Our subhedge on the other hand performs worse as
the barriers get closer together. In fact, for barriers at $105\ \&\ 95$
or $103\ \&\ 97$, an exponential utility agent having a long position in
a double touch barrier option would prefer to use the standard
delta/vega hedging to our sub-replicating strategy. This would appear
to be due to the fact that typically the option knocks-in quickly and
delta/vega-hedging then carries smaller transactions costs.

Naturally when the barriers vary the types of super- and sub- hedges
which we use change. Figure \ref{fig:types_hedges}
shows which types are optimal depending on the values of the
barriers. This provides an illustrations of the intuitive labels we
gave to each case in Section \ref{sec:dt}. In this implementation we assumed one thousands strikes
between $0$ and $500$ are available and digital calls are not traded. This results
in non-smoothness of the borders between the cases.

\begin{table}[tbp]
  \centering
\begin{tabular}[h]{l|c|c|c|cl}
Barriers& Position & Delta/Vega & Model-free& Model-free& (barriers monitored exactly)\\
\hline
$130$--$70$ & short & $-0.076$ & \boldmath$-0.026$&$-0.025$&\\
$130$--$70$ & long &$-0.078$ & \boldmath$-0.027$&$-0.027$&\\
\hline
$120$--$80$ & short & $-0.125$ & \boldmath$-0.047$&$-0.045$&\\
$120$--$80$ &long & $-0.127$ & \boldmath$-0.069$&$-0.068$&\\
\hline
$117$--$83$ & short & $-0.137$ & \boldmath $-0.065$&$-0.062$&\\
$117$--$83$ & long &$-0.139$ & \boldmath $-0.078$&$-0.077$&\\
\hline
$115$--$85$ & short & $-0.144$ & \boldmath$-0.085$&$-0.082$&\\
$115$--$85$ & long &$-0.147$ & \boldmath$-0.086$&$-0.085$&\\
\hline
$110$--$90$ & short & $-0.139$ & \boldmath$-0.063$&$-0.061$&\\
$110$--$90$ & long &$-0.144$ & \boldmath$-0.115$&$-0.113$&\\
\hline
$105$--$95$ & short & $-0.096$ & \boldmath$-0.048$&$-0.047$&\\
$105$--$95$ & long &\boldmath$-0.102$ & $-0.171$&$-0.167$&\\
\hline
$103$--$97$ & short & $-0.070$ & \boldmath$-0.039$&$-0.038$&\\
$103$--$97$ & long &\boldmath$-0.075$ & $-0.199$&$-0.194$&\\
\hline
$120$--$95$ & short & $-0.127$ & \boldmath$-0.059$&$-0.057$&\\
$120$--$95$ & long &$-0.130$ & \boldmath$-0.099$&$-0.098$&\\
\hline
$105$--$80$ & short & $-0.155$ & \boldmath$-0.076$&$-0.075$&\\
$105$--$80$ & long &$-0.161$ & \boldmath$-0.124$&$-0.121$&
\end{tabular}
\caption{Comparison of exponential utilities of hedging errors of positions in a double touch options under Heston model \eqref{eq:heston}--\eqref{eq:heston_par} resulting from delta/vega
hedging and our model-free super- or sub- hedging strategies. In each case, the preferred
hedge is highlighted. The last column reports the change in utility when our strategies are allowed
to monitor exactly the moments of barrier crossings.}
\label{tab:exp_util}
\end{table}

\begin{figure}[htbp]
\begin{multicols}{2}{
\includegraphics[height=6cm]{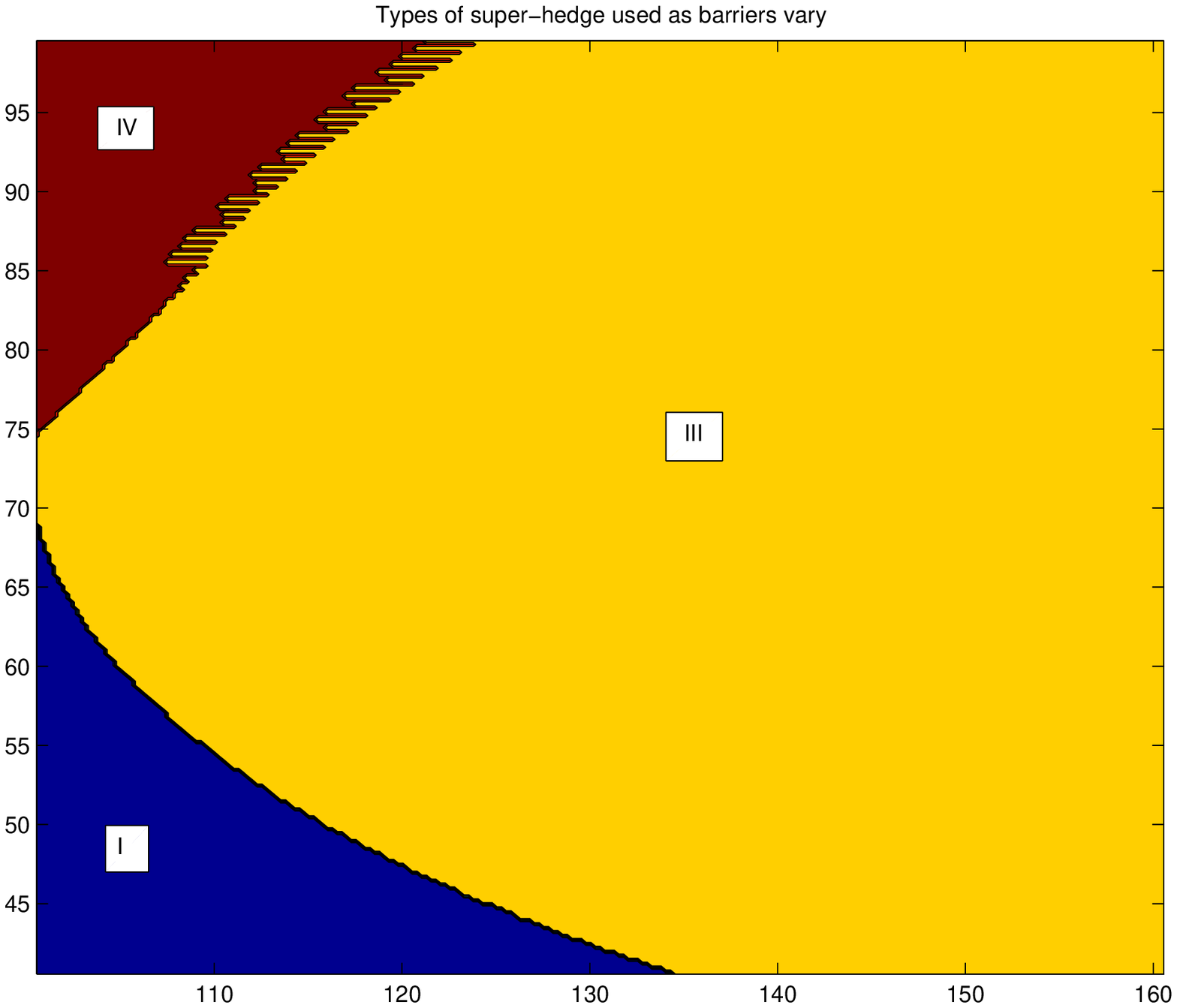}

\includegraphics[height=6cm]{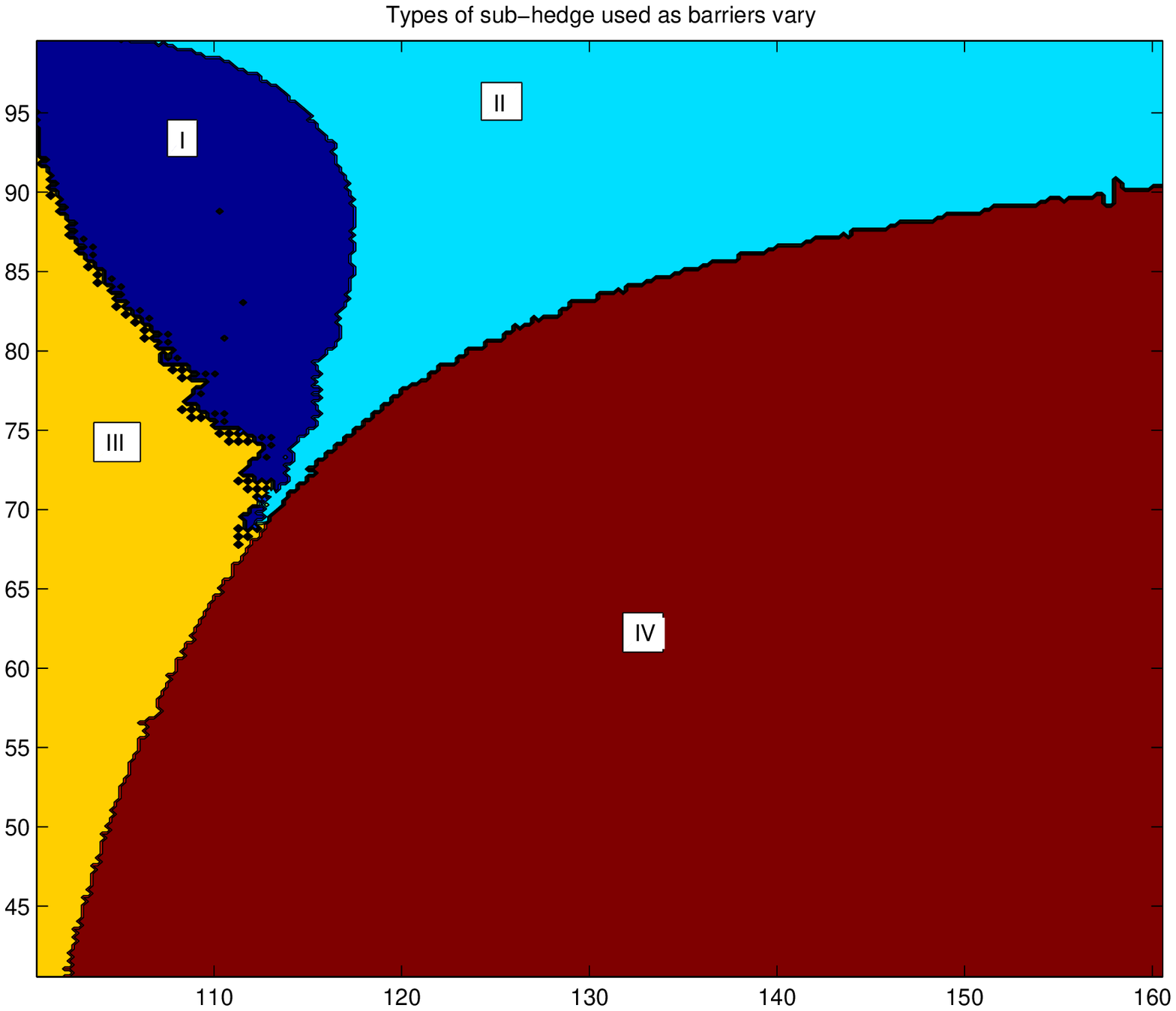}
}
\end{multicols}

\caption{Optimal types of superhedges $\uH$ (\emph{left}) and
  subhedges $\lh$ (\emph{right}) as barriers vary under the Heston
  model \eqref{eq:heston}--\eqref{eq:heston_par}. $\lh^{IV}\equiv 0$
  is the trivial subhedge. Superhedge $\uH^{II}$ is not visible as it only appears for upper barrier levels above $325$.}
\label{fig:types_hedges}
\end{figure}

The results presented in this section clearly show that the hedges we
advocate are in many circumstances an improvement on the classical
hedges. Naturally, there is a large literature (\eg{}\
\cite{CvitanicKaratzas:96}) on more sophisticated techniques that
might offer a considerable improvement over the classical hedge we
have implemented and it is possible that such improvements would
reverse the relative performance of the hedges. However, we hope that
the numerical evidence does at least convince the reader that the
quasi-static hedges are competitive with dynamic hedging, and that
their differing nature means that they might well prove a more
suitable approach in situations where market conditions are
dramatically different to the idealised Black-Scholes world --- \eg{}
large transaction costs, illiquidity or large jumps. There also seems
scope for a more sophisticated approach based on the quasi-static
hedges, but allowing for some model-based trading: for example, a
hybrid of a quasi-static portfolio and some dynamic trading could be
used to reduce some of the over-hedge in the simple quasi-static
hedge. Finally, in our simulations we did not really incorporate model
misspecification risk. We would expect that if a trader believes in a
model which is significantly different from the real world model then
our model-free hedging strategies would outperform hedging 
`using the wrong model.'

\section{Proofs}
\label{sec:proofs}

Let $(B_t)_{t \ge 0}$ be a standard real valued Brownian motion
starting from $B_0$. We recall that for any probability measure $\nu$
on $\R_+$ with $\nu_B(\R_+)=B_0$ we can find a stopping time $\tau$
such that $B_\tau\sim\nu$ and $(B_{t\wedge \tau})$ is a uniformly
integrable martingale. Such stopping is simply a solution to the
Skorokhod embedding problem and number of different explicit solutions
are known, see Ob\l \'oj \cite{Obloj:04b} for an overview of the
domain. Note also that when $\nu([a,b])=1$ then $(B_{t\wedge \tau})$
is a uniformly integrable martingale if and only if $B_t\in[a,b]$,
$t\leq \tau$, a.s.  In the sequel when speaking about embedding a
measure we implicitly mean embedding it in a UI
manner in $(B_t)$.\\
Recall that if $B_0=S_0$ and $\tau$ is an embedding of $\mu$ then
$S_t:=B_{\tau\land \frac{t}{T-t}}$, $t\leq T$, is a market model which
matches the market input \eqref{eq:market_input}. In what follows we
will be constructing embeddings $\tau$ of $\mu$ such that the
associated market model attains equality in our super- or sub- hedging
inequalities. Stopping times $\tau$ will often be compositions of
other stopping times embedding (rescaled) restrictions of $\mu$ or
some other intermediary measures. Unless specified otherwise, the
choice of particular intermediary stopping times has no importance and
we do not specify it -- one's favorite solution to the Skorokhod
embedding problem can be used.

\begin{proof}[Proof of Theorem~\ref{thm:upper_price}]
  We start with some
  preliminary lemmas and then prove Theorem~\ref{thm:upper_price}. In
  the body of the proof, cases \fbox{I} to \fbox{IV} refer to the cases
  stated in Theorem~\ref{thm:upper_price}.  We note that, by considering
  $z_0 = \gamma_+(w_0)$, \fbox{III} is equivalent to:
  \[
  \mbox{There exists $z_0 \ge \ub$ such that
    $\gamma_+(\gamma_-(z_0)) = z_0$ and $\rho_+(z_0) \ge
    \rho_-(\gamma_-(z_0))$.} 
  \]
We recall, without proof, straightforward properties of the barycentre function \eqref{eq:barycentre}
which will be useful in the sequel.
  \begin{lemma}\label{lem:baryprop} The barycentre function defined in
    \eqref{eq:barycentre} satisfies
    \begin{itemize}
    \item $\mu_B(\Gamma)\geq a$, $\Gamma_1\subset (0,a)$
      $\Longrightarrow$ $\mu_B(\Gamma\setminus \Gamma_1)\geq a$,
    \item $a\leq \mu_B(\Gamma)\leq\mu_B(\Gamma_1)\leq b$ and
      $\mu(\Gamma\cap\Gamma_1)= 0$ $\Longrightarrow$ $a\leq
      \mu_B(\Gamma\cup\Gamma_1)\leq b$. 
    \end{itemize}
  \end{lemma}
  
  %

  We separate the proof into 2 steps now. In the first step we prove
  that exactly one of \fbox{I}-\fbox{IV} holds. In the second step we
  construct the appropriate embeddings and market models which achieve
  the upper bounds on the prices. 

  \noindent
  {\bf \sc Step 1:}\\ 
  This step is divided into 4 cases. We start with technical lemmas which are proved
  after the cases are considered.

  \begin{lemma} \label{lem:incrho}
    If $\rho_+(\gamma_+(w_0)) < \rho_-(w_0)$ for some $w_0$, then
    $\rho_+(\gamma_+(w)) < \rho_-(w)$ for all $w \le w_0$. Similarly, if
    $\rho_-(\gamma_-(z_0)) > \rho_+(z_0)$ for some $z_0$, then
    $\rho_-(\gamma_-(z)) > \rho_+(z)$ for all $z \ge z_0$
  \end{lemma}

  \begin{lemma} \label{lem:notAandB}
    If $\ub\geq\rho_-(0)$ and $\lb\leq \rho_+(\infty)$ then at least
    of one of the functions $\gamma_\pm$ is bounded on its domain. In
    particular at most one of \fbox{I} and \fbox{II} may be true. 
  \end{lemma}

  \begin{lemma}\label{lem:DnotCnotAB}
    \fbox{IV} implies not \fbox{III}. \fbox{III} implies not (\fbox{I} or
    \fbox{II}). 
  \end{lemma}

  \noindent
  {\bf \sc Case (a):} $\ub \ge \rho_-(0), \lb \le \rho_+(\infty)$\\ We
  note first of all that the case \fbox{IV} is not possible, and that
  the second half of the conditions for \fbox{I} and \fbox{II} are
  trivially true. Suppose neither \fbox{I} or \fbox{II} hold. If we
  have $\gamma_-(\ub)=0$ then \fbox{III} holds with $w_0=0$ and if
  $\gamma_+(\lb)=\infty$ then \fbox{III} holds with $w_0=\lb$. We may
  thus assume that $\gamma_+(\cdot)$ is bounded above, and
  $\gamma_-(\cdot)$ is bounded away from zero, which in turn implies:
  \[
  \gamma_-(\gamma_+(\lb)) < \lb \mbox{ and } \gamma_-(\gamma_+(0))
  >0. 
  \]
  The function $\gamma_-(\gamma_+(\cdot))$ is continuous and
  increasing on $(0,\lb]$ and thus we must have $w_0$ such that
  $\gamma_-(\gamma_+(w_0)) = w_0$. Finally, suppose that for such a
  $w_0$ we in fact have $\rho_-(w_0) > \rho_+(\gamma_+(w_0))$, then
  Lemma~\ref{lem:incrho} implies $\rho_-(0) >
  \rho_+(\gamma_+(0))=\ub$, contradicting our assumptions. That only
  one of the cases $\fbox{I}$--$\fbox{III}$ holds now follows from Lemmas~\ref{lem:notAandB}
  and \ref{lem:DnotCnotAB}.


  {\bf \sc Case (b):} $\ub \ge \rho_-(0), \lb > \rho_+(\infty)$

  It follows that neither \fbox{II} or \fbox{IV} are possible. Observe
  that Lemma \ref{lem:incrho} implies that $\rho_-(w)\leq
  \rho_+(\gamma_+(w))$ for all $w\leq \lb$ --- if this were not true,
  then $\ub = \rho_+(\gamma_+(0)) < \rho_-(0)$.  Suppose further that
  \fbox{I} does not hold.  If $\gamma_-(\ub)=0$ then \fbox{III} holds
  with $w_0=0$. \\ So assume instead that $\gamma_-(\cdot)$ is bounded
  away from zero, and therefore that $w<\gamma_-(\gamma_+(w))$ for $w$
  close to zero. If we show also that
  $\gamma_-(\gamma_+(\rho_+(\infty)))\leq \rho_+(\infty)$ then by
  continuity of $\gamma_-(\gamma_+(\cdot))$ there exists a suitable
  $w_0$ for which \fbox{III} holds.  Let
  $\Gamma_1=(\rho_+(\infty),\infty)$ and
  $\Gamma_2=(\rho_+(\gamma_+(\rho_+(\infty))),\gamma_+(\rho_+(\infty)))$. We
  have by definition $\mu_B(\Gamma_1)=\ub=\mu_B(\Gamma_2)$ so that
  $\mu_B(\Gamma_1\setminus\Gamma_2)=\ub$, since $\Gamma_2\subset
  \Gamma_1$.  Let
  $\Gamma=\big(\rho_+(\infty),\rho_-(\rho_+(\infty))\big)\cup\big(\gamma_+(\rho_+(\infty)),\infty\big)$
  and note that $\mu_B(\Gamma)=\ub$ is equivalent to
  $\gamma_-(\gamma_+(\rho_+(\infty)))=\rho_+(\infty)$. Noting that
  $\rho_-(w) \le \rho_+(\gamma_+(w))$ for all $w \le \lb$ implies
  $\rho_-(\rho_+(\infty)) \le \rho_+(\gamma_+(\rho_+(\infty)))$, and
  using Lemma~\ref{lem:baryprop} have
  $$\mu_B\big(\Gamma\big)=\mu_B\Big(\big(\Gamma_1\setminus\Gamma_2\big)\setminus
  \big(\rho_-(\rho_+(\infty)),\rho_+(\gamma_+(\rho_+(\infty)))\big)\Big)\geq
  \ub,$$ which implies that $\gamma_-(\gamma_+(\rho_+(\infty)))\leq
  \rho_+(\infty)$. As previously, it remains to note that
  Lemma~\ref{lem:DnotCnotAB} implies exclusivity of $\fbox{III}$ and $\fbox{I}$.

  {\bf \sc Case (c):} $\ub < \rho_-(0), \lb \le \rho_+(\infty)$

  This case is essentially identical to Case {\sc (b)} above. 

  {\bf \sc Case (d):} $\ub < \rho_-(0), \lb > \rho_+(\infty)$

  Note that we now cannot have either of \fbox{I} or \fbox{II}. Suppose
  further that \fbox{IV} does not hold --- or rather, the weaker:
  \begin{eqnarray*}
    \rho_+(\rho_-(0)) & > & \lb,\\ \rho_-(\rho_+(\infty)) & < & \ub. 
  \end{eqnarray*}
  Let $\Gamma_w=(w,\rho_-(w))\cup (\ub,\infty)$ and observe that
  $\mu_B(\Gamma_w)$ decreases as $w$ decreases, provided
  $\rho_-(w)<\ub$. We have
  $$\mu_B(\Gamma_{\rho_+(\infty)})\geq
  \mu_B((\rho_+(\infty),\infty))=\ub.$$ Our assumption
  $\rho_-(\rho_+(\infty))<\ub<\rho_-(0)$ implies that
  $\rho_-^{-1}(\ub)\in (0,\rho_+(\infty))$ so that
  $\Gamma_{\rho_-^{-1}(\ub)}\supsetneq (\rho_+(\infty),\infty)$ and in
  consequence $\mu_B(\Gamma_{\rho_-^{-1}(\ub)})<\ub$. Using continuity
  of $w\to \mu_B(\Gamma_w)$ we conclude that there exists a $w_1 \in
  (\rho_-^{-1}(\ub),\rho_+(\infty)]$ with $\mu_B(\Gamma_{w_1})=\ub$, or
    equivalently $\gamma_-(\ub)=w_1$. A symmetric argument implies
    that $\gamma_+(\lb) >0$.  We conclude, as in Case {\sc (a)}, that
    there exists a $w_0$ such that $\gamma_-(\gamma_+(w_0))=w_0$.

  It remains to show that if \fbox{IV} does not hold, and the first half
  of the condition for \fbox{III} holds, then so too does the second
  condition. Suppose $w_0$ is a point satisfying
  $\gamma_-(\gamma_+(w_0))=w_0$ and suppose for a contradiction that
  $\rho_-(w_0) > \rho_+(\gamma_+(w_0))$. Since the sets
  $(\rho_+(\gamma_+(w_0)),\gamma_+(w_0))$ and
  $(w_0,\rho_-(w_0))\cup(\gamma_+(w_0),\infty)$ are both centred at
  $\ub$, and overlap, it follows that $\mu_B([w_0,\infty))>\ub$ and in
  consequence $\rho_+(\infty) < w_0$. Symmetric arguments imply that
  $\rho_-(0) > \gamma_+(w_0)$. Applying $\rho_-(\cdot),
  \rho_+(\cdot)$ to these inequalities, we further deduce that
  \begin{eqnarray*}
    \rho_-(w_0) & < &\rho_-(\rho_+(\infty)) \\ \rho_+(\rho_-(0)) & <
    & \rho_+(\gamma_+(w_0))
  \end{eqnarray*}
  which, together with the assumption that $\rho_-(w_0) >
  \rho_+(\gamma_+(w_0))$, implies:
  \[
  \rho_+(\rho_-(0)) < \rho_+(\gamma_+(w_0)) < \rho_-(w_0) <
  \rho_-(\rho_+(\infty)),
  \]
  contradicting \fbox{IV} not holding. 
\end{proof}

\begin{proof}[Proof of Lemma~\ref{lem:incrho}]
  Consider $w<w_0$ with $\rho_+(\gamma_+(w_0)) < \rho_-(w_0)$. The
  latter implies $\mu_B\big((w_0,\rho_+(\gamma_+(w_0)))\big)<\lb$, so
  that $\mu_B\big((0,\gamma_+(w_0))\big)<\lb$. Suppose now that
  $\rho_+(\gamma_+(w)) \geq \rho_-(w)$. As $\rho_-$ is decreasing and
  $\gamma_+$ is increasing we have $\rho_-(w_0)<\rho_-(w)\leq
  \rho_+(\gamma_+(w))$ and $\ub<\gamma_+(w)<\gamma_+(w_0)$. We then have
  \begin{eqnarray*}
    \lb&=&\mu_B\big((0,\rho_-(w))\cup(\rho_+(\gamma_+(w)),\gamma_+(w))\big)\\ &=&\mu_B\Big((0,\gamma_+(w_0))\setminus
    \big[\big(\rho_-(w_0),\rho_+(\gamma_+(\w))\big)\cup\big((\gamma_+(w),\gamma_+(w_0)\big)\big]\Big)\\ &\leq&
    \mu_B\big((0,\gamma_+(w_0))\big)<\lb,
  \end{eqnarray*}
  which gives the desired contradiction. 
\end{proof}

\begin{proof}[Proof of Lemma~\ref{lem:notAandB}]
  Define $\Gamma = (0,\lb)\cup(\ub,\infty)$ and consider
  $\mu_B(\Gamma)$. If both \fbox{I} and \fbox{II} hold, or more
  generally if $\gamma_-(z_0)=0$ and $\gamma_+(w_0)=\infty$ for
  some $z_0\geq \ub$, $w_0\leq \lb$, we have:
  \begin{equation}\label{eqn:lembarycentre1}
    \mu_B((0,\rho_-(0))\cup(\ub,\infty)) \ge
    \mu_B((0,\rho_-(0))\cup(z_0,\infty)) = \ub\\
  \end{equation}
  and
  \begin{equation}\label{eqn:lembarycentre2}
    \mu_B((0,\lb)\cup(\rho_+(\infty),\infty)) \le
    \mu_B((0,w_0)\cup(\rho_+(\infty),\infty)) = \lb. 
  \end{equation}
  Now suppose $\mu_B(\Gamma) < \ub$. Then:
  \[
  \mu_B(\Gamma \cup(\lb,\rho_-(0))) < \ub
  \]
  contradicting \eqref{eqn:lembarycentre1} and similarly, if
  $\mu_B(\Gamma) > \lb$
  \[
  \mu_B(\Gamma \cup(\rho_+(\infty),\ub)) > \lb,
  \]
  contradicts \eqref{eqn:lembarycentre2}. 
\end{proof}
\begin{proof}[Proof of Lemma~\ref{lem:DnotCnotAB}]
  \fbox{IV} $\Longrightarrow$ not \fbox{III}\\ Assume both \fbox{IV} and
  \fbox{III} hold. From the definition of $\gamma_+(w_0)$ and
  $\rho_-(w_0)$ we have that $\mu_B(\Gamma_+)=\lb$ for
  $\Gamma_+=\big(0,\rho_-(w_0)\big)\cup\big(\rho_+(\gamma_+(w_0)),\gamma_+(w_0)\big)$. 
  This implies that $\gamma_+(w_0)\geq \rho_-(0)$ since otherwise
  $$\mu_B(\Gamma_+)=\mu_B\Big(\big(0,\rho_-(0)\big)\setminus\left\{
    \big(\rho_-(w_0),\rho_+(\gamma_+(w_0))\big) \cup \big(\gamma_+(w_0),\rho_-(0)\big)\right\}\Big)>\lb,$$
    where we also used the assumption $\rho_-(w_0)\leq
    \rho_+(\gamma_+(w_0))$. Likewise, using
    $\gamma_-(\gamma_+(w_0))=w_0$ we see that $w_0\leq
    \rho_+(\infty)$. Applying $\rho_-$ to the last inequality and
    using our assumptions we obtain
  $$ \rho_+(\rho_-(0))<\rho_-(\rho_+(\infty))\leq \rho_-(w_0)\leq
  \rho_+(\gamma_+(w_0)).$$ In consequence, $\gamma_+(w_0)<\rho_-(0)$
  which gives the desired contradiction.\medskip\\ \fbox{III}
  $\Longrightarrow$ not (\fbox{I} or \fbox{II})\\ Suppose \fbox{III}
  and \fbox{II} hold together. Let $w_1<\lb$ be the point given by
  \fbox{II} such that $\gamma_+(w_1)=\infty$ and $w_0$ the point in
  \fbox{III} such that $\gamma_-(\gamma_+(w_0))=w_0$. Naturally, as
  $\gamma_-(\gamma_+(w_1))=\gamma_-(\infty)=\lb>w_1$ we have that
  $w_0<w_1$. Observe also that $\mu_B((0,\rho_-(w_1))\cup
  (\rho_+(\infty),\infty))=\lb$ and $\mu_B(\re)=S_0 \in (\lb,\ub)$
  which readily imply $\lb<\rho_-(w_1)<\rho_+(\infty)<\ub$. Let us
  further denote $r=\min\{\rho_-(w_0),\rho_+(\infty)\}$ and
  $R=\max\{\rho_-(w_0),\rho_+(\infty)\}$ so that finally, using our
  assumptions,
  \begin{equation}\label{eq:lemma_ordering}
    w_0<w_1<\lb<\rho_-(w_1)<r\leq R\leq \rho_+(\gamma_+(w_0))\leq \ub\leq
    \gamma_+(w_0). 
  \end{equation}
  By definition we have
  $$\int_{\rho_+(\infty)}^\infty (\ub-u)\mu(\td
  u)=0=\int_{w_0}^{\rho_-(w_0)}(\ub-u)\mu(\td
  u)+\int_{\rho_+(\gamma_+(w_0))}^\infty (\ub-u)\mu(\td u).$$
  Subtracting these two quantities we arrive at
  \begin{eqnarray}\label{eq:lemma_sets}
    \int_{R}^{\rho_+(\gamma_+(w_0))}(\ub-u)\mu(\td u)=\int_{w_0}^r
    (\ub-u)\mu(\td u),&\quad\textrm{and using \eqref{eq:lemma_ordering} we
      deduce}\nonumber\\ \mu\Big(\big(R,\rho_+(\gamma_+(w_0))\big)\Big)>\mu\Big((w_0,r)\Big).&
  \end{eqnarray}
  Using the properties of our functions again we have
  $$\int_{-\infty}^{w_1}(u-\lb)\mu(\td u)+\int_{\rho_+(\infty)}^\infty
  (u-\lb)\mu(\td u)=0=\int_{-\infty}^{w_0}(u-\lb)\mu(\td
  u)+\int_{\rho_+(\gamma_+(w_0))}^{\gamma_+(w_0)}(u-\lb)\mu(\td u),$$
  which after subtracting, using $\int_{w_0}^{w_1}(u-\lb)\mu(\td
  u)=-\int_{\rho_-(w_1)}^{\rho_-(w_0)}(u-\lb)\mu(\td u)$, yields
  \begin{equation}\label{eq:lemma_final}
    \int_{\rho_+(\infty)}^{\rho_+(\gamma_+(w_0))}(u-\lb)\mu(\td
    u)-\int_{\rho_-(w_1)}^{\rho_-(w_0)}(u-\lb)\mu(\td
    u)+\int_{\gamma_+(w_0)}^{\infty} (u-\lb)\mu(\td u)=0. 
  \end{equation}
  The last term in \eqref{eq:lemma_final} is positive and for the first
  two terms, using \eqref{eq:lemma_sets}, we have
  \begin{eqnarray}
    \int_{R}^{\rho_+(\gamma_+(w_0))}(u-\lb)\mu(\td
    u)&\geq&\big(R-\lb\big)\mu\Big(\big(R,\rho_+(\gamma_+(w_0))\big)\Big)
    \nonumber\\ &>&\big(R-\lb\big)\mu\Big((\rho_-(w_1),r)\Big)\geq
    \int_{\rho_-(w_1)}^r(u-\lb)\mu(\td u). 
  \end{eqnarray}
  This readily implies that the left hand side of
  \eqref{eq:lemma_final} is strictly positive leading to the desired
  contradiction.\\ The case when \fbox{III} and \fbox{I} hold together
  is similar.\medskip\\ {\bf \sc Step 2:} Construction of relevant
  embeddings.\\ Our strategy is
  now as follows. For each of the four exclusive cases
  \fbox{I}-\fbox{IV} we construct a stopping time $\tau$ which solves
  the Skorokhod embedding problem for $\mu$ and such that for the
  price process $S_t:=B_{\frac{t}{T-t}\wedge \tau}$ the appropriate
  superhedge $\uh^{I}-\uh^{IV}$ is in fact a perfect hedge. The
  stopping time $\tau$ will be a composition of stopping times, each
  of which is a solution to an embedding problem for a (rescaled)
  restriction of $\mu$ to appropriate intervals.\smallskip\\ 
  Suppose that \fbox{II} holds.\\ This embedding is closely
  related to the classical Az\'ema-Yor embedding \cite{AzemaYor:79}
  used in the work of Brown, Hobson and Rogers \cite{Brown:01b} on
  one-sided barrier options. Let $\tau_1$ be a UI embedding, in
  $(B_t)_{t \ge 0}$ with $B_0=S_0$, of
  $$\nu^1=\mu|_{(w_0,\rho_+(\infty))}+p\delta_{\lb},\quad\textrm{where
  }p=\big(1-\mu((w_0,\rho_+(\infty)))\big),$$ which is centred in
  $S_0$. Let $\nu^2=\frac{1}{p}\mu|_{\R_+\setminus
    (w_0,\rho_+(\infty))}$, which is a probability measure with
  $\nu^2_B(\R_+)=\lb$, and let $\tau_2$ be the Az\'ema-Yor embedding
  (cf.\ Ob\l\'oj \cite[Sec.~5]{Obloj:04b}) of $\nu^2$, i.e. 
  $$\tau_2=\inf\big\{t>0: \overline{B}_t\geq
  \nu^2_{B}([B_t,\infty))\big\},$$ which is a UI embedding of $\nu^2$
  when $B_0=\lb$. Note that $\nu^2_B([x,\infty))=\mu_B([x,\infty))$
  for $x\geq \rho_+(\infty)$ and that
  $$\{\overline{B}_{\tau_2}\geq \ub\}=\{B_{\tau_2}\geq
  \rho_+(\infty)\},\quad\textrm{since
  }\nu^2_B((\rho_+(\infty),\infty))=\ub.$$ We define our final
  embedding as follows: we first embed $\nu^1$ and then the atom
  in $\lb$ is diffused into $\nu^2$ using the Az\'ema-Yor
  procedure, i.e. 
  \begin{equation}
    \tau:= \tau_1\mathbf{1}_{B_{\tau_1}\neq
      \lb}+\tau_2\circ\tau_1\mathbf{1}_{B_{\tau_1}=\lb},
  \end{equation}
  where $B_0=S_0$. Clearly, $\tau$ is a UI embedding of $\mu$ and
  $S_t:=B_{\frac{t}{T-t}\wedge \tau}$ defines a model for the stock
  price which matches the given prices of calls and puts,
  i.e.~$S_T\sim\mu$. Furthermore, $\{\sS_T\geq \ub\} = \{\sS_T\geq
  \ub, \iS_T \le \lb\} = \{S_T\geq
  \rho_+(\infty)\}$ and it follows that
  $$\dbp = \uh^{II}(\rho_+(\infty)).$$
  \smallskip\\ Suppose that \fbox{I} holds.\\ This is a mirror image of
  \fbox{II}. We first embed $\nu^1=\mu|_{(\rho_-(0),z_0)}+p\delta_{\ub}$,
  with $p=1-\mu((\rho_-(0),z_0))$. Then the atom in $\ub$ is diffused
  into $\mu|_{\R_+\setminus (\rho_-(0),z_0)}$ using the reversed
  Az\'ema-Yor stopping time (cf. Ob\l\'oj
  \cite[Sec.~5.3]{Obloj:04b}). The resulting stopping time $\tau$ and
  the stock price model $S_t:=B_{\frac{t}{T-t}\wedge \tau}$ satisfy
  $\dbp=\uh^{I}(\rho_-(0))$.\smallskip\\ Suppose that \fbox{III}
  holds.\\ We describe the embedding in words before writing it
  formally. We first embed $\mu$ on
  $(\rho_-(w_0),\rho_+(\gamma_+(w_0)))$ or we stop when we hit $\ub$ or
  $\lb$. If we hit $\ub$ then we embed $\mu$ on $(\gamma_+(w_0),\infty)$
  or we run until we hit $\lb$. Likewise, if we first hit $\lb$ then we
  embed $\mu$ on $(0,w_0)$ or we run till we hit $\ub$. Finally, from
  $\lb$ and $\ub$ we embed the remaining bits of $\mu$.\\ We now
  formalise these ideas. Let
  \begin{equation}\label{eq:def_of_nu1_I}
    \nu^1=p\delta_{\lb}+\mu|_{(\rho_-(w_0),\rho_+(\gamma_+(w_0)))}+\big(1-p-\mu(\rho_-(w_0),\rho_+(\gamma_+(w_0)))\big)
  \end{equation}
  where $p$ is chosen so that $\nu^1_B(\R_+)=S_0$. Define two more
  measures
  \begin{equation}
    \begin{split}
      \nu^2&=\mu\big([w_0,\rho_-(w_0)]\big)\delta_{\lb} +
      \mu|_{(\gamma_+(w_0),\infty)}\\ \nu^3&=\mu|_{(0,w_0)}+\mu\big([\rho_+(\gamma_+(w_0)),\gamma_+(w_0)]\big)\delta_{\ub}
    \end{split}
  \end{equation}
  and note that by definition $\nu^2_B(\R_+)=\ub$ and
  $\nu^3_B(\R_+)=\lb$. Furthermore, as the barycentre of
  $$\nu^3+\mu|_{((\rho_-(w_0),\rho_+(\gamma_+(w_0)))}+\nu^2$$ is equal
  to the barycentre of $\mu$, and from
  the uniqueness of $p$ in \eqref{eq:def_of_nu1_I}, we deduce that
  $$\nu^3(\R_+)=p,\textrm{ and
  }\nu^2(\R_+)=q=\big(1-p-\mu(\rho_-(w_0),\rho_+(\gamma_+(w_0)))\big).$$
  Let $\tau_1$ be a UI embedding of $\nu^1$ (for $B_0=S_0$), $\tau_2$ be
  a UI embedding of $\frac{1}{q}\nu^2$ (for $B_0=\ub$) and $\tau_3$ be a
  UI embedding of $\frac{1}{p}\nu^3$ (for $B_0=\lb$). Further, let
  $\tau_4$ and $\tau_5$ be UI embeddings of respectively
  $$\frac{1}{\mu\big((w_0,\rho_-(w_0))\big)}\mu|_{(w_0,\rho_-(w_0))}\textrm{
    and
  }\frac{1}{\mu\big((\rho_+(\gamma_+(w_0)),\gamma_+(w_0))\big)}\mu|_{(\rho_+(\gamma_+(w_0)),\gamma_+(w_0))},$$
  where the starting points are respectively $B_0=\lb$ and
  $B_0=\ub$. We are ready to define our stopping time. Let $B_0=S_0$
  and write $H_z=\inf\{t:B_t=z\}$. We put
  \begin{equation}
    \begin{split}
      \tau:=&\tau_1\mathbf{1}_{\tau_1<H_{\lb}\wedge
        H_{\ub}}\\
      &+\tau_2\circ\tau_1\mathbf{1}_{H_{\ub}=\tau_1}\mathbf{1}_{\tau_2\circ\tau_1<H_{\lb}}\\
      &+\tau_4\circ\tau_2\circ\tau_1\mathbf{1}_{H_{\ub}=\tau_1}\mathbf{1}_{H_{\lb}=\tau_2\circ\tau_1}\\
      &+\tau_3\circ\tau_1\mathbf{1}_{H_{\lb}=\tau_1}\mathbf{1}_{\tau_3\circ\tau_1<H_{\ub}}\\
      &+\tau_5\circ\tau_3\circ\tau_1\mathbf{1}_{H_{\lb}=\tau_1}\mathbf{1}_{H_{\ub}=\tau_3\circ\tau_1},
    \end{split}
  \end{equation}
  and it is immediate from the properties of our measures that
  $B_\tau\sim\mu$ and $(B_{t\wedge\tau})$ is a UI martingale. 
  Furthermore, with $S_t:=B_{\frac{t}{T-t}\wedge \tau}$, we see that
  $$\dbp =
  \uh^{III}\big(\gamma_+(w_0),\rho_+(\gamma_+(w_0)),\rho_-(w_0),w_0\big),\quad
  a.s.$$

  Finally, suppose that \fbox{IV} holds.\\
  In this case, we initially run to $\{\lb,\ub\}$ without stopping any
  mass. Then, from $\ub$, we either run to $\lb$ or embed $\mu$ on
  $(\rho_-(0),\infty)$. The mass which is at $\lb$ after the first
  step is run to either $\ub$ or used to embed $\mu$ on
  $(0,\rho_+(\infty))$. The mass which remains at $\lb$ and $\ub$ is
  then used to embed the remaining part of $\mu$ on
  $(\rho_+(\infty),\rho_-(0))$.

  To begin with, we define the measures
  \begin{eqnarray*}
    \nu^1 & = & \left[\frac{S_0 - \lb}{\ub-\lb} -
      \mu((\rho_-(0),\infty))\right] \delta_{\lb} + \mu
    |_{(\rho_-(0),\infty)},\\
    \nu^2 & = & \left[\frac{\ub - S_0}{\ub-\lb} -
      \mu((0,\rho_+(\infty))\right] \delta_{\ub} + \mu
    |_{(0,\rho_+(\infty))}.
  \end{eqnarray*}
  Then $\nu^1$ is a measure, since $\mu((\rho_-(0),\infty)) <
  \frac{S_0 - \lb}{\ub-\lb}$: noting that $\ub < \rho_-(0)$, we get
  \begin{eqnarray*}
    0 & = & \int_0^{\rho_-(0)} (u-S_0) \mu(\td u) +
    \int_{\rho_-(0)}^\infty (u-S_0) \mu(\td u)\\
    & \ge & (\lb - S_0) \mu((0,\rho_-(0))) + (\ub-S_0)
    \mu((\rho_-(0),\infty))
  \end{eqnarray*}
  and the statement follows. Moreover, we can see that $\nu^1_B(\R_+)
  = \ub$:
  \begin{equation*}
    \begin{split}
      \int_{\rho_-(0)}^\infty & (u-\ub) \mu(\td u) + \frac{S_0 -
        \lb}{\ub-\lb} (\lb-\ub) - \mu((\rho_-(0),\infty))(\lb-\ub)\\
      & = \int_{\rho_-(0)}^\infty (u-\lb) \mu(\td u) + (\lb - S_0) \\
      & = (S_0 - \lb) - \int_0^{\rho_-(0)} (u-\lb) \mu(\td u) + (\lb
      -S_0) = 0.
    \end{split}
  \end{equation*}
  Similar results hold for $\nu^2$.

  Consequently, we can construct the first stages of the embedding. The
  final stage is to run from $\lb$ and $\ub$ to embed the remaining
  mass. Of course, it does not matter exactly how we do this from the
  optimality point of view, since these paths have already struck both
  barriers, but we do need to check that the embedding is possible. It
  is clear that the means and probabilities match, but unlike the
  previously considered cases, we now have initial mass in two places,
  and the existence of a suitable embedding is not trivial. To resolve
  this, we note the following: suppose we can find a point $z^* \in
  (\rho_+(\rho_-(0)),\rho_-(\rho_+(\infty)))$ such that
  $\nu^1(\{\lb\}) = \mu((\rho_+(\infty),z^*))$. Then because
  $\mu_B((\rho_+(\infty),\rho_-(\rho_+(\infty)))) = \lb$, we can find
  $z_1 \in (\rho_+(\infty), z^*)$ such that the measure
  \[
  \nu^3 = \mu|_{(\rho_+(\infty), z_1)} + (\nu^1(\{\lb\}) -
  \mu((\rho_+(\infty), z_1))) \delta_{z^*}
  \]
  has barycentre $\lb$. There is a similar construction
  for $\nu^4$ and a point $z_2$ which will embed mass from $\nu^2$ at
  $\ub$ to $\mu$ on $(z_2,\rho_-(0))$, and an atom at $z^*$. In the
  final stage, we can then embed the mass from $z^*$ to $(z_1,z_2)$.

  It remains to show that we can find such a point $z^*$. To do this,
  we check that there is sufficient mass being stopped at $\lb$ at the
  end of the second step (\ie{} which has already hit $\ub$.)
  Specifically, we need to show that
  \[
  \frac{S_0-\lb}{\ub-\lb} - \mu((\rho_-(0),\infty)) \ge
  \mu((\rho_+(\infty),\rho_+(\rho_-(0)))).
  \]
  Rearranging, and using the definitions of the functions $\rho_+$ and
  $\rho_-$, this is equivalent to
  \begin{eqnarray*}
    (S_0-\lb) & \ge & \int_{\rho_+(\infty)}^\infty(u-\lb)\, \mu(\td u) -
    \int_{\rho_+(\rho_-(0))}^{\rho_-(0)} (u-\lb) \mu(\td u)\\
    & \ge & \int_{(\rho_+(\infty),\rho_+(\rho_-(0))) \cup
      (\rho_-(0),\infty)}(u-\lb) \, \mu(\td u)\\
    & \ge & (S_0-\lb) - \int_{(0,\rho_+(\infty)) \cup
      (\rho_+(\rho_-(0)),\rho_-(0))} (u-\lb) \, \mu(\td u).
  \end{eqnarray*}
  Using the definitions of the appropriate functions, this can be seen
  to be equivalent to
  \[
  0 \le \int_{\rho_+(\rho_-(0))}^{\rho_-(\rho_+(\infty))} (u-\lb) \,
  \mu(dy),
  \]
  which follows since $\rho_+(\rho_-(0)) > \lb$. The construction of
  the appropriate stopping time, and its optimality follow as previously. This ends the proof
  of Theorem \ref{thm:upper_price}.
\end{proof}

In order to prove Theorem \ref{thm:sub_prices} we start with an auxiliary lemma.
\begin{lemma} \label{lem:subhmaxv}
  Either we may construct an embedding of $\mu$ under which the process never
  hits both $\ub$ and $\lb$, or
  \begin{eqnarray}
    \inf\{v \in [\lb,\ub]:\psi(v) < \infty\} & \ge & \inf\{v \in
    [\lb,\ub]: \theta(v)>-\infty\}
    \label{eq:infpsigeinftheta}\\
    \sup\{v\in [\lb,\ub]:\psi(v) < \infty\} & \ge & \sup\{v \in
    [\lb,\ub] : \theta(v)>-\infty\}
    \label{eq:suppsigesuptheta}
  \end{eqnarray}
  and we may then write
  \begin{equation} \label{eqn:vbarlevbar}
    \lv = \inf\{v\in [\lb,\ub]:\psi(v) < \infty\} \le \sup\{v \in
    [\lb,\ub] : \theta(v)>-\infty\}
    = \uv,
  \end{equation}
  where $\lv,\uv$ are given in \eqref{eqn:vvK_3}.
\end{lemma}
\begin{proof}
  We begin by showing that if $\theta(v) = -\infty$ for all $v \in
  [\lb,\ub]$ then there exists an embedding of $\mu$ which does not
  hit both $\lb$ and $\ub$. 

  For $w \ge \ub$, define $\alpha_*(w)$ to be the mass that must be
  placed at $\lb$ in order for the barycentre of this mass plus $\mu$
  on $(\ub,w)$ to be $\ub$, so $\alpha_*(w)$ satisfies
  \begin{equation*}
    \alpha_*(w) \lb + \int_{\ub}^w u \, \mu(\td u) = \ub\left(\alpha_*(w)+
      \mu((\ub,w))\right). 
  \end{equation*}
  If follows that $\alpha_*(w)$ exists, although there is no guarantee
  that it is less than $1-\mu((\ub,w))$. In addition, define
  $\beta^*(w)$ to be
  \begin{equation*}
    \beta^*(w) = \inf\left\{ \beta \in [\lb,\ub] : \int_{(\lb,\beta)
        \cup (\ub,w)} u \, \mu(\td u) = \ub \int_{(\lb,\beta)
        \cup (\ub,w)} \, \mu(\td u)\right\}. 
  \end{equation*}
  If this is finite, then it is the point at which
  $\mu_B((\lb,\beta^*(w))\cup (\ub,w)) = \ub$. Note also that
  $\beta^*(w)$ is increasing as a function of $w$, and is
  continuous when $\beta^*(w) < \infty$. Define
  \begin{eqnarray*}
    p_*(w) & = & \mu((\ub,w)) + \alpha_*(w) \\
    p^*(w) & = & \mu((\lb,\beta^*(w))\cup (\ub,w)). 
  \end{eqnarray*}

  Suppose initially that $\beta^*(w) \le \ub$ for all $w \ge
  \ub$. Then we may assign the following interpretations to these
  quantities: $p_*(w)$ is the smallest amount of mass that we can
  start at $\ub$ and run to embed $\mu$ on $(\ub,w)$, $(\lb,\cdot)$
  and an atom at $\lb$, and $p^*(w)$ is the largest amount of mass
  that we may do this with: the smallest amount is attained by running
  all the mass below $\ub$ to $\lb$, while the largest probability is
  attained by running all this mass to $(\lb,\beta^*(w))$. The
  assumption that $\beta^*(w) \le \ub$ implies that this upper bound
  does not run out of mass to embed. Moreover, by adjusting the size
  of the atom at $\lb$, we can embed an atom of any size between
  $p_*(w)$ and $p^*(w)$ from $\ub$ in this way. Recalling the
  definition of $\theta(v)$, we conclude that there exists $v$ such
  that $\theta(v) = w$ if and only if $p_*(w) \le \frac{S_0 -
    \lb}{\ub-\lb} \le p^*(w)$. Finally, note that the functions
  $p_*(w)$ and $p^*(w)$ are both increasing in $w$, and further that
  $p_*(\ub) = 0$. Consequently, if there is no $v$ such that
  $\theta(v) >-\infty$, and $\beta^*(w) \le \ub$ for all $w \ge \ub$,
  we must have $p^*(\infty) := \lim_{w \to \infty} p^*(w) < \frac{S_0
    - \lb}{\ub-\lb}$. 

  So suppose $p^*(\infty) < \frac{S_0 - \lb}{\ub-\lb}$. We now
  construct an embedding as follows: from $S_0$, we initially run to
  either $\ub$ or
  \begin{equation*}
    b_* = \frac{S_0 - p^*(\infty) \ub}{1-p^*(\infty)}. 
  \end{equation*}
  Since $p^*(\infty) < \frac{S_0 - \lb}{\ub-\lb}$, then $b_* \in
  (\lb,S_0)$, and the probability that we hit $\ub$ before $b_*$ is
  $p^*(\infty)$. In addition, by the definition of $\beta^*(w)$, we
  deduce that the set $(\lb,\beta^*(\infty)] \cup [\ub,\infty)$ is
  given mass $p^*(\infty)$ by $\mu$, and that the barycentre of $\mu$
  on this set is $\ub$.  We may therefore embed the paths from $\ub$
  to this set, and the paths from $b_*$ to the remaining intervals,
  $[0,\lb] \cup (\beta^*(\infty),\ub)$ and we note no paths will hit
  both $\ub$ and $\lb$. 

  So suppose instead that $\beta^*(w_0) = \ub$ for some $w_0$, with
  $p^*(w_0) \le \frac{S_0 - \lb}{\ub-\lb}$. (If the latter condition
  does not hold, then using the fact that $\beta^*(w)$ is
  left-continuous and increasing, we can find a $w$ such that
  $\beta^*(w) < \ub$ and $p^*(w) = \frac{S_0 - \lb}{\ub-\lb}$, and
  therefore, by the arguments above, there exists $v$ with $\theta(v)
  > -\infty$.) We may then continue to construct measures with
  barycentre $\ub$, which are equal to $\mu$ on $(\lb,w)$ for $w>w_0$,
  and have a compensating atom at $\lb$. As we increase $w$,
  eventually either $w$ reaches $\infty$, or the mass of the measure
  reaches $\frac{S_0 - \lb}{\ub-\lb}$. In the latter case, we know
  $\theta(\ub) = w$, contradicting $\theta(v) = -\infty$ for all $v
  \in [\lb,\ub]$. So consider the former case: we obtained that the
  measure which is $\mu$ on $(\lb,\infty)$ with a further atom at
  $\lb$ to give barycentre $\ub$ has total mass ($p$ say) less than
  $\frac{S_0 - \lb}{\ub-\lb}$. We show that this is impossible: divide
  $\mu$ into its restriction to $(0,z)$ and $[z,\infty)$, where $z$ is
  chosen so that $\mu([z,\infty)) = p$. Then $z<\lb$ and the
  barycentre of the restriction to $[z,\infty)$ is strictly smaller
  than the barycentre of the measure with the mass on $[z,\lb)$ placed
  at $\lb$, which is the measure described above, and which has
  barycentre $\ub$. Additionally, the barycentre of the lower
  restriction of $\mu$ must be strictly smaller than $\lb$. Moreover, we
  may calculate the barycentre of $\mu$ by considering the barycentre
  of the two restrictions; since $\mu$ has mean (and therefore
  barycentre) $S_0$, we must have:
  \[
  S_0  = (1-p) \mu_B((0,z)) + p \mu_B([z,\infty)) < (1-p) \lb + p \ub
  < \lb \frac{\ub - S_0}{\ub-\lb} + \ub \frac{S_0 - \lb}{\ub-\lb} = S_0,
  \]
  which is a contradiction.

  We conclude that, if $\{v \in [\lb,\ub]: \theta(v) > -\infty\}$ is
  empty, there is an embedding of $\mu$ which does not hit both $\ub$
  and $\lb$. A similar result follows for $\psi(v)$. In particular, if
  we assume that there is no such embedding, then there exists $v$
  such that $\psi(v) < \infty$, and (not necessarily the same) $v$
  such that $\theta(v) > -\infty$. We now wish to show that
  \eqref{eq:infpsigeinftheta} holds. Suppose not. Then:
  \begin{equation*}
    v_* := \inf \{ v \in [\lb,\ub]: \psi(v) < \infty\} < \inf \{ v \in
    [\lb,\ub]: \theta(v) > -\infty\} =: v^*. 
  \end{equation*}
  Moreover, we can deduce from the definition of $\theta(v)$ that
  since $v^* > \lb$, we must have $\theta(v^*) = \infty$.  Now
  consider the barycentre of the measure which is taken by running
  from $S_0$ to $\lb$ and $\ub$, and then from $\lb$ to
  $(\psi(v_*),\lb) \cup (v_*,\ub)$, with a compensating mass at $\ub$,
  so that the measure has barycentre $\lb$, and from $\ub$ to
  $(\lb,v^*) \cup (\ub,\theta(v^*)) = (\lb,v^*) \cup (\ub,\infty)$,
  with a compensating mass at $\lb$, so that the measure has
  barycentre $\ub$. Then the whole law of the resulting process must
  have mean $S_0$, since this can be done in a uniformly integrable
  way, but the resulting distribution is at least $\mu$ on
  $(\psi(v_*),\infty)$ (it is twice $\mu$ on $(v_*,v^*)$, has atoms at
  $\lb$ and $\ub$ and is $\mu$ elsewhere), and zero on $(0,\psi(v_*))$,
  so must have mean greater than $S_0$, which is a contradiction. A
  similar argument shows \eqref{eq:suppsigesuptheta}. Hence, we may
  conclude (still under the assumption that there is no embedding
  which never hits both $\lb$ and $\ub$) that the equalities in
  \eqref{eqn:vbarlevbar} hold. It remains to show the
  inequality when $\lv,\uv\in (\lb,\ub)$. However this is now almost immediate: the forms of $\uv, \lv$ imply that $\mu$ gives mass $\frac{\ub-S_0}{\ub-\lb}$ to the
  set $(\psi(\lv),\lb)\cup (\lv,\ub)$ and it gives mass $\frac{S_0 -
    \lb}{\ub-\lb}$ to the set $(\lb,\uv)\cup (\ub,\theta(\uv))$. If
  $\uv < \lv$, this implies that $\mu$ gives mass 1 to the set
  $(\psi(\lv),\uv) \cup (\lv,\theta(\uv)) \subsetneq [0,\infty)$,
  contradicting the positivity of $\mu$.
 \end{proof}

\begin{proof}[Proof of Theorem \ref{thm:sub_prices}]
  From Lemma \ref{lem:subhmaxv} case $\fbox{IV}$ and the last
  statement of the theorem follow.  Assume from now on that $\lv\leq
  \uv$. We note firstly that $\psi(v)$ and $\theta(v)$ are both
  continuous and decreasing on $[\lv, \uv]$, and consequently $\kappa(v)$
  is also continuous and decreasing as a function of $v$ on $[\lv,
  \uv]$. It follows that the three cases $\kappa(\lv) < \lv$, $\kappa(\uv)
  > \uv$ and the existence of $v_0\in [\lv,\uv]$ such that $\kappa(v_0)
  = v_0$ are exclusive and exhaustive. We consider each case
  separately:
  \begin{enumerate}
  \item[$\fbox{I}$]
    Suppose that there exists $v_0\in [\lv,\uv]$ such that $\kappa(v_0) =
    v_0$. By the definition of $\psi(v)$, we can run all the mass
    initially from $S_0$ to $\{\lb, \ub\}$ and then embed (in a uniformly
    integrable way) from $\lb$ to $(\psi(v_0),\lb) \cup (v_0,\ub)$ and a
    compensating atom at $\ub$ with the remaining mass, and similarly from
    $\ub$ to $(\lb,v_0) \cup (\ub,\theta(v_0))$ with an atom at $\lb$. 
    The mass now at $\lb$ and
    $\ub$ can now be embedded in the remaining tails in a suitable way ---
    the means and masses must agree, since the initial stages were
    embedded in a uniformly integrable manner, and the remaining mass all
    lies outside $[\lb,\ub]$. We denote $\tau$ the stopping time which achieves the embedding.

    Now we compare both sides of the inequality in
    \eqref{eqn:asineqsH1}, where we choose $K_1 = \theta(v_0), K_2 =
    \psi(v_0)$ and therefore, as a consequence of the definition of
    $\kappa(v)$, we also have $K_3 = v_0$. The key observation is now
    that the mass is stopped only at points where the inequality is an
    equality: mass which hits $\ub$ initially either stops in the
    interval $(\lb,K_3) \cup (\ub,K_1)$, when there is equality in
    \eqref{eqn:asineqsH1}, or it goes on to hit $\ub$, and from this
    point also continues to the tails $(0,K_2)\cup(K_1,\infty)$, where
    there is again equality between both sides of
    \eqref{eqn:asineqsH1}. Taking expectations on the right of
    \eqref{eqn:asineqsH1}, we get the terms on the right of
    \eqref{eqn:subhedge1}, and we conclude that  \eqref{eqn:subhedge1} holds in the market model
    $S_t:=B_{\tau\land \frac{t}{T-t}}$.

  \item[$\fbox{II}$] Suppose now that $\kappa(\uv) > \uv$. Then we
    must have $\uv = \sup\{v \in [\lb,\ub]: \theta(v) > -\infty\}$ by
    Lemma~\ref{lem:subhmaxv}, and then $\uv < \kappa(\uv) \le
    \ub$. So, by the definition of $\theta(v)$, since $\uv$ exists and
    is less than $\ub$, we must have
    \[
    \int_{(\lb,\uv)\cup(\ub,\theta(\uv))} u \, \mu(\td u) = \ub
    \frac{S_0-\lb}{\ub-\lb} \mbox{ and }
    \mu((\lb,\uv)\cup(\ub,\theta(\uv))) = \frac{S_0-\lb}{\ub-\lb},
    \]
    and we can embed from $\ub$ (having initially run to
    $\{\lb,\ub\}$) to $(\lb,\uv)\cup(\ub,\theta(\uv))$ {\it without}
    leaving an atom at $\lb$. Similarly, we can also run from $\lb$ to
    $(\psi(\uv), \lb) \cup (\uv,\ub)$ with an atom at $\ub$. The
    atom can then be embedded in the tails
    $(0,\psi(\uv))\cup(\theta(\uv),\infty)$ in a uniformly integrable
    manner. We now need to show that when we take $K_3 = \uv, K_1 = \theta(\uv)$ and $K_2 =
    \psi(\uv)$ we get the required equality in \eqref{eqn:subhedge2}. 
    The main difference from the above case occurs in the
    case where we hit $\ub$ initially and then hit $\lb$: we no longer
    need equality in \eqref{eqn:asineqsH1}, since this no longer
    occurs in our optimal construction, however what remains to be
    checked is that the inequality does hold on this
    set. Specifically, we need to show that:
    \[
    1 \ge \alpha_0 + \alpha_1 (K_2 - S_0) - (\alpha_3 - \alpha_3 + \alpha_1)(K_2 -
    \ub) + (\alpha_3 -\alpha_2) (K_2 - \lb)
    \]
    Using \eqref{eq:subh2b} and \eqref{eqn:subh2soln}, we see that
    this occurs when
    \[
    \frac{(K_3 - K_2)(K_1 - \lb)}{(\ub-K_2)(K_1 - K_3)} \le 1
    \]
    which rearranges to give:
    \[
    K_3 \le \ub \frac{K_1-\lb}{(K_1 - \lb) + (\ub - K_2)} + \lb \frac{\ub
      -K_2}{(K_1 - \lb) + (\ub - K_2)}. 
    \]
    This is satisfied by our choice of $\uv$ as $K_3$, and $K_1 =
    \theta(\uv), K_2 = \psi(\uv)$.
    \item[$\fbox{III}$] This is symmetric to case $\fbox{II}$.
  \end{enumerate}
\end{proof}

\section{Additional Figures}
\label{sec:figures}

\begin{figure}[htbp]
\begin{multicols}{2}{
\includegraphics[height=6cm]{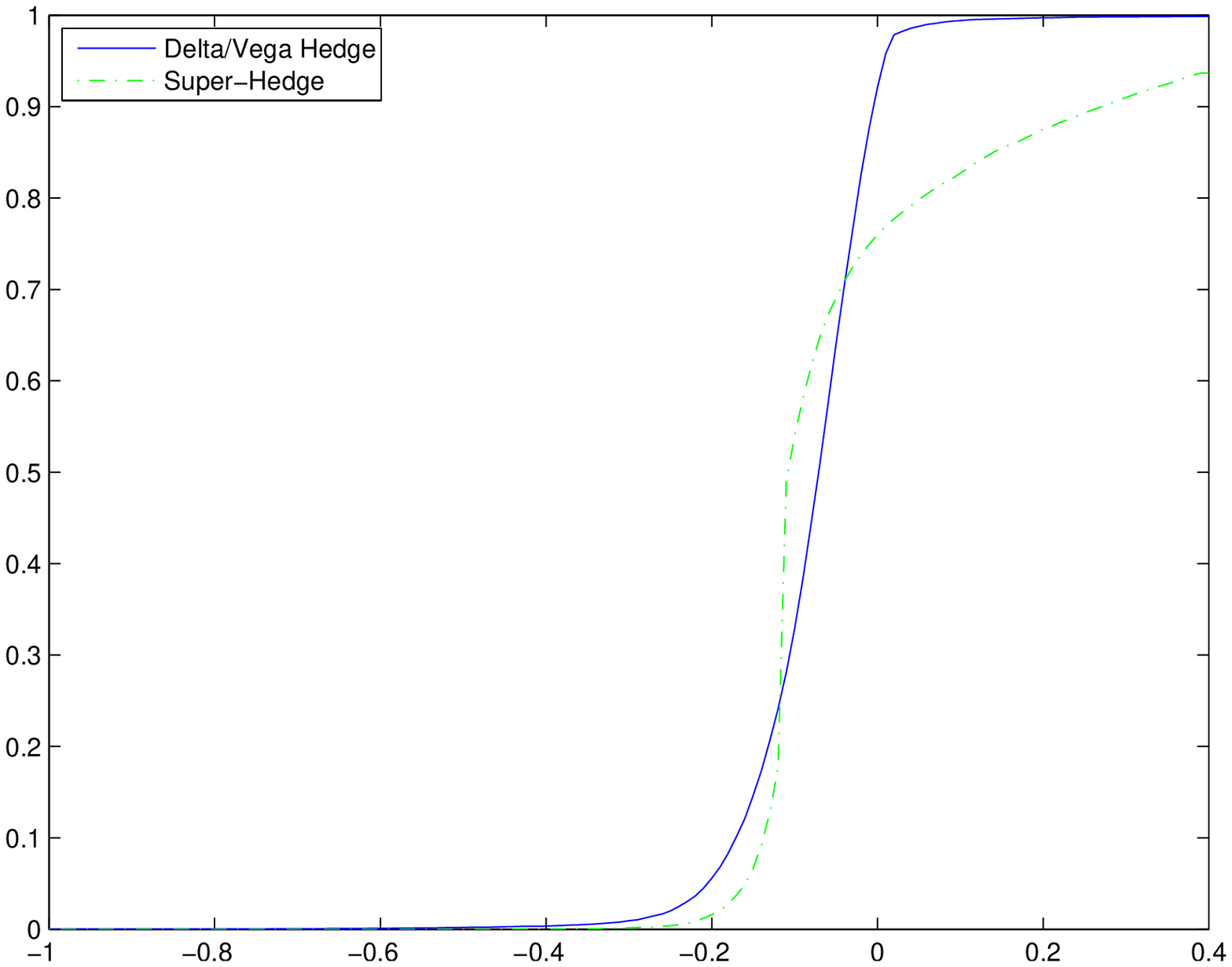}

\includegraphics[height=6cm]{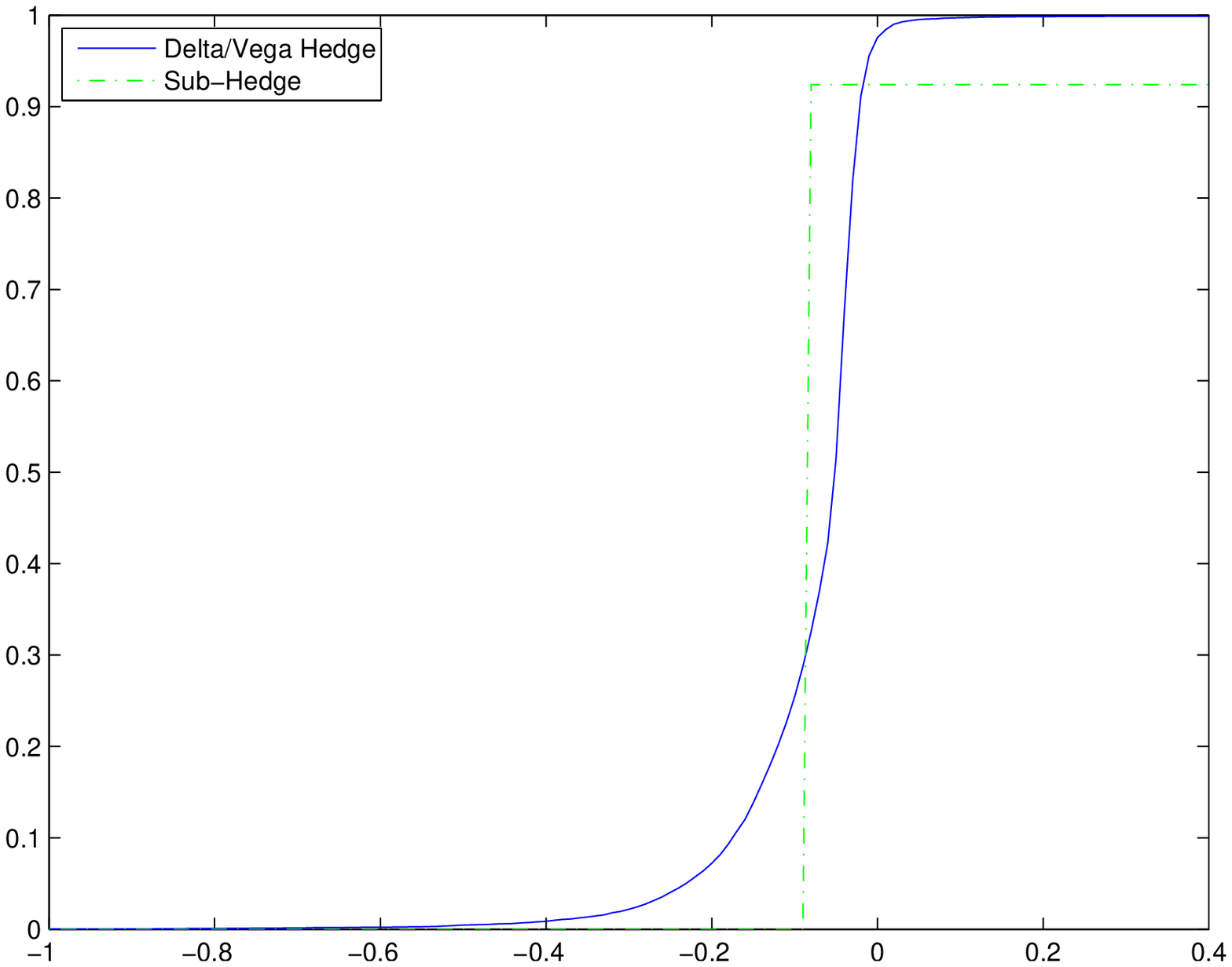}
}
\end{multicols}

\caption{Cumulative distributions of hedging errors under different
  scenarios of a short position (\emph{left}) and a long position
  (\emph{right}) in a double touch option with barriers at $130$ and
  $70$ under the Heston model
  \eqref{eq:heston}--\eqref{eq:heston_par}.}
\label{fig:hegerr_130_70}
\end{figure}


\begin{figure}[htbp]
\begin{multicols}{2}{
\includegraphics[height=6cm]{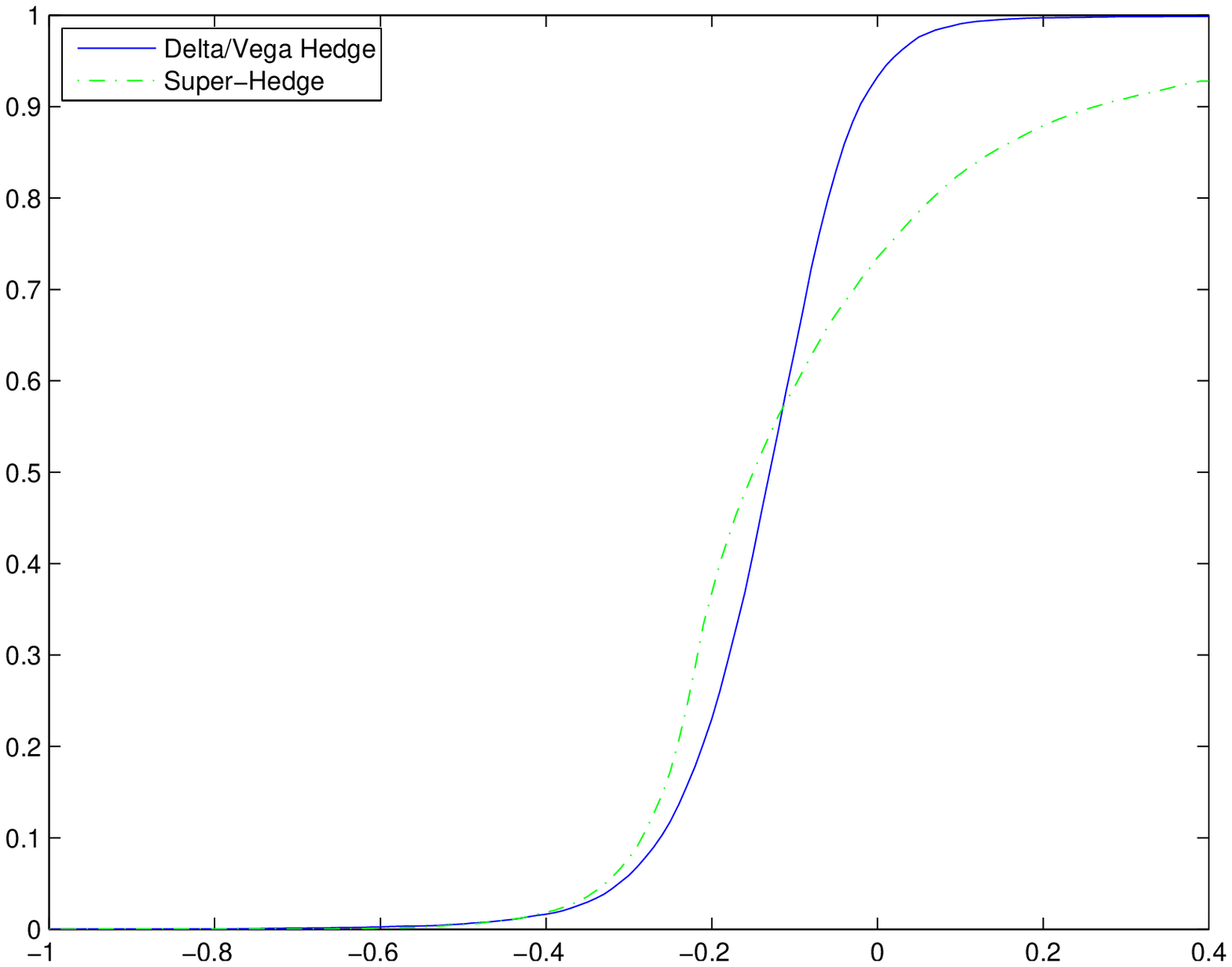}

\includegraphics[height=6cm]{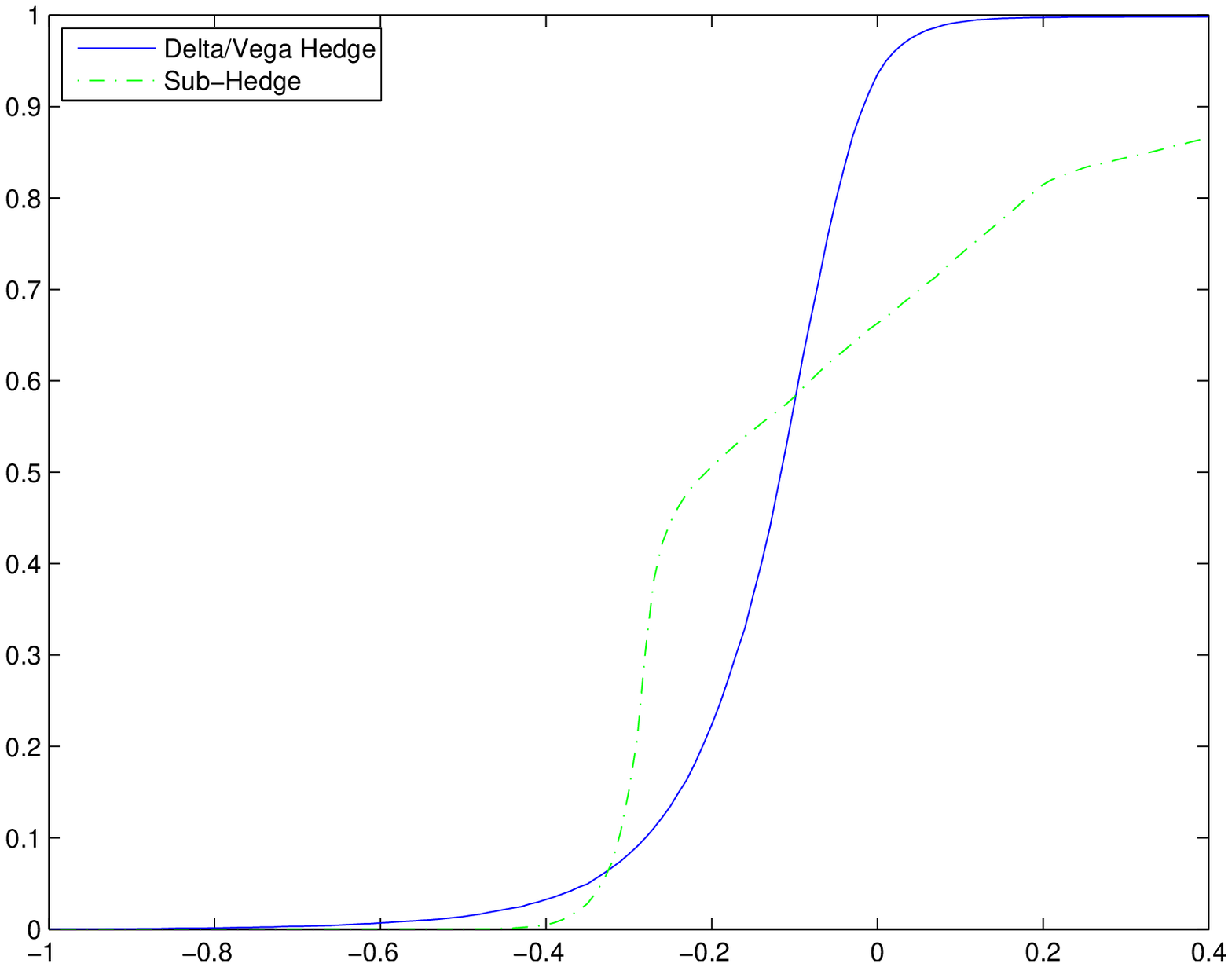}
}
\end{multicols}

\caption{Cumulative distributions of hedging errors under different
  scenarios of a short position (\emph{left}) and a long position
  (\emph{right}) in a double touch option with barriers at $115$ and
  $85$ under the Heston model \eqref{eq:heston}--\eqref{eq:heston_par}.}
\label{fig:hegerr_115_85}
\end{figure}




\begin{figure}[htbp]
\begin{multicols}{2}{
\includegraphics[height=6cm]{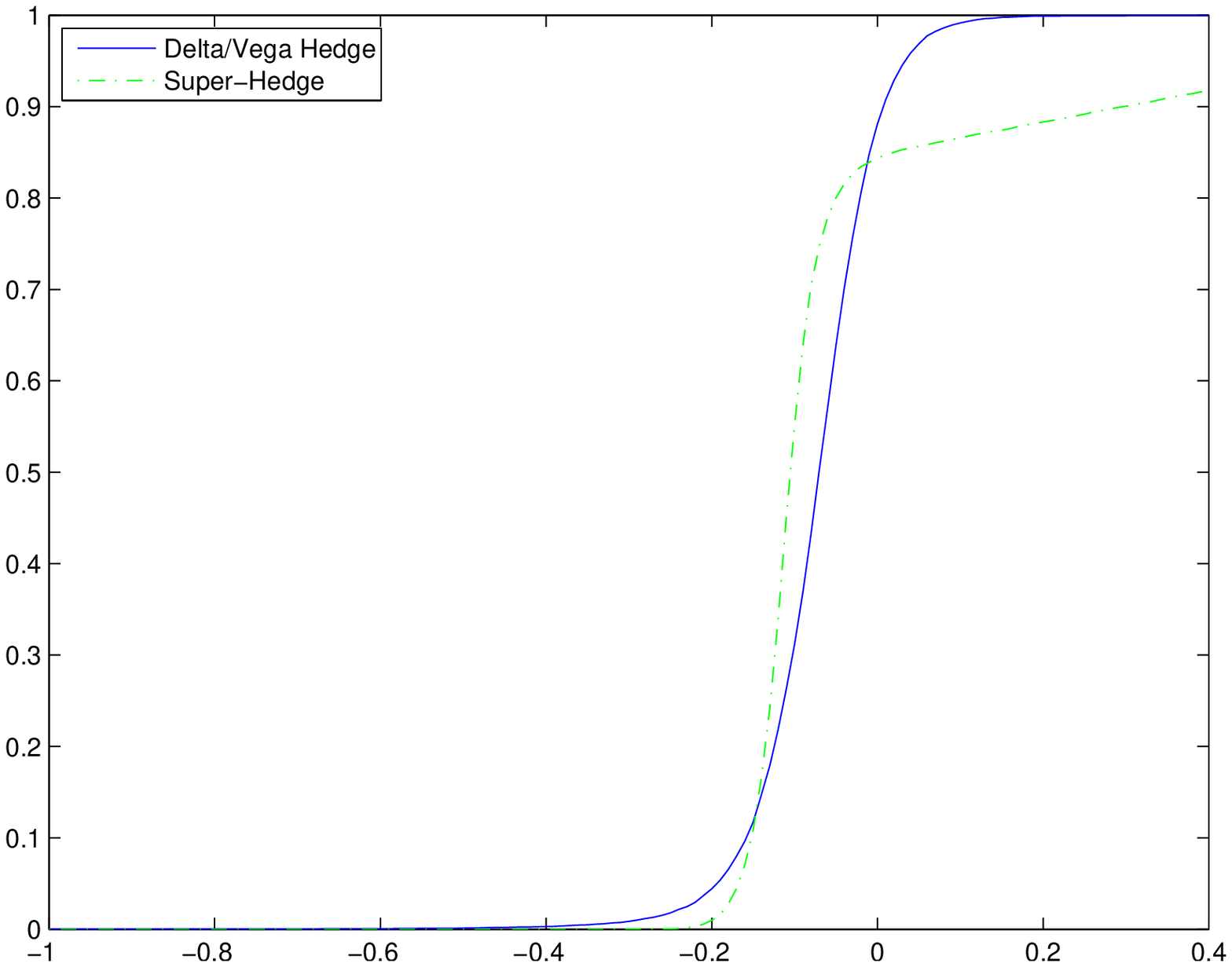}

\includegraphics[height=6cm]{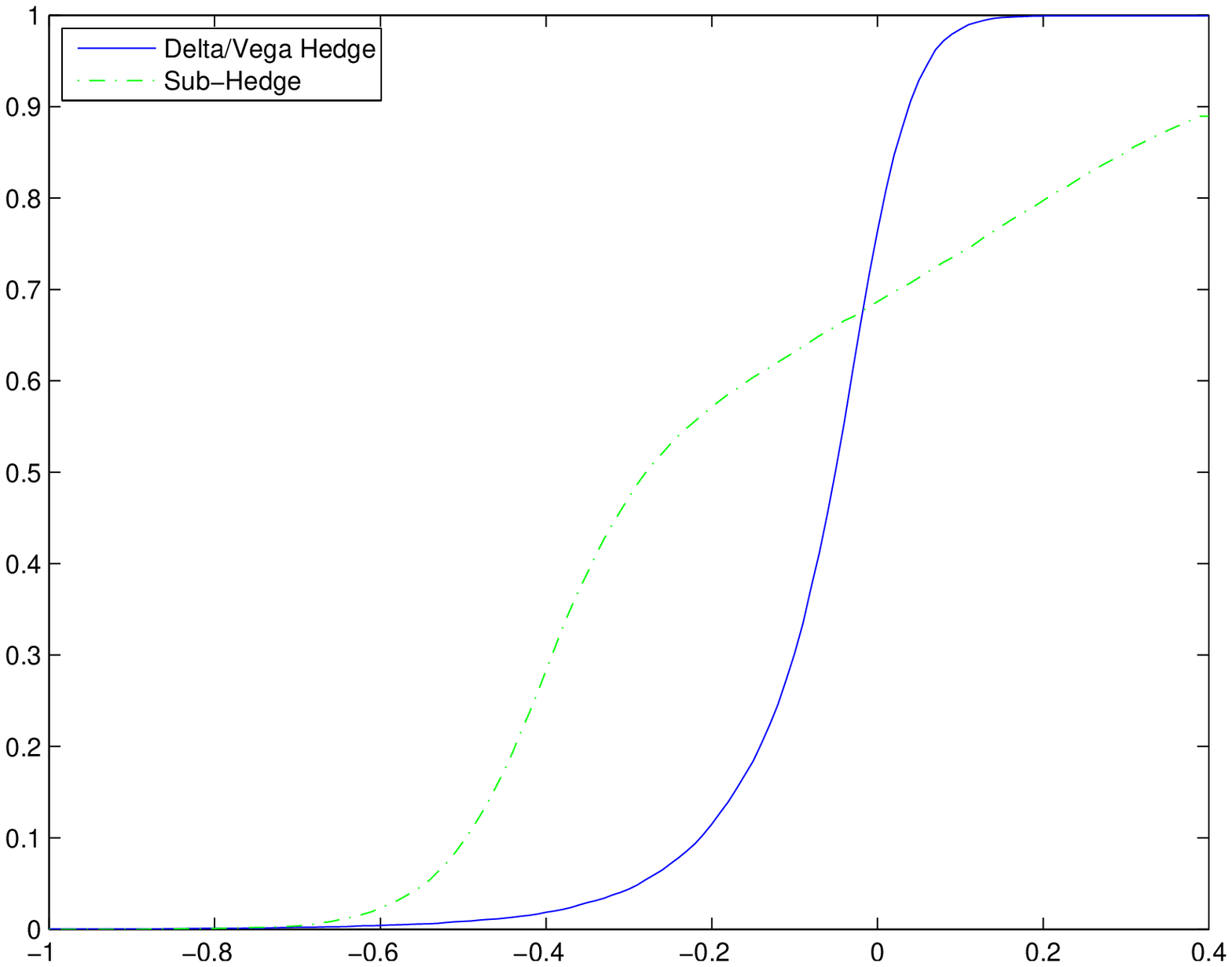}
}
\end{multicols}

\caption{Cumulative distributions of hedging errors under different
  scenarios of a short position (\emph{left}) and a long position
  (\emph{right}) in a double touch option with barriers at $103$ and
  $97$ under the Heston model \eqref{eq:heston}--\eqref{eq:heston_par}.}
\label{fig:hegerr_103_97}
\end{figure}
\begin{figure}[htbp]
\begin{multicols}{2}{
\includegraphics[height=6cm]{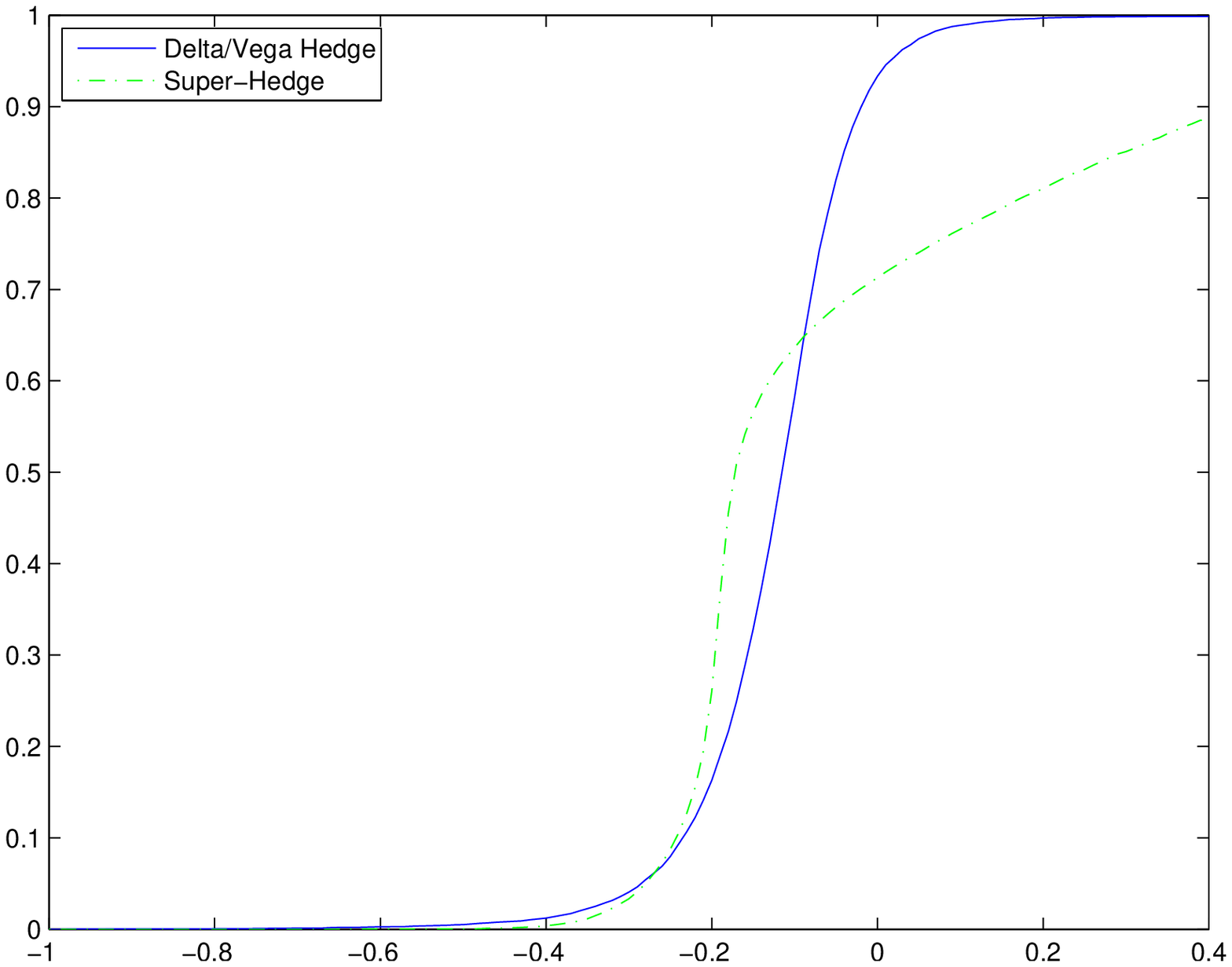}

\includegraphics[height=6cm]{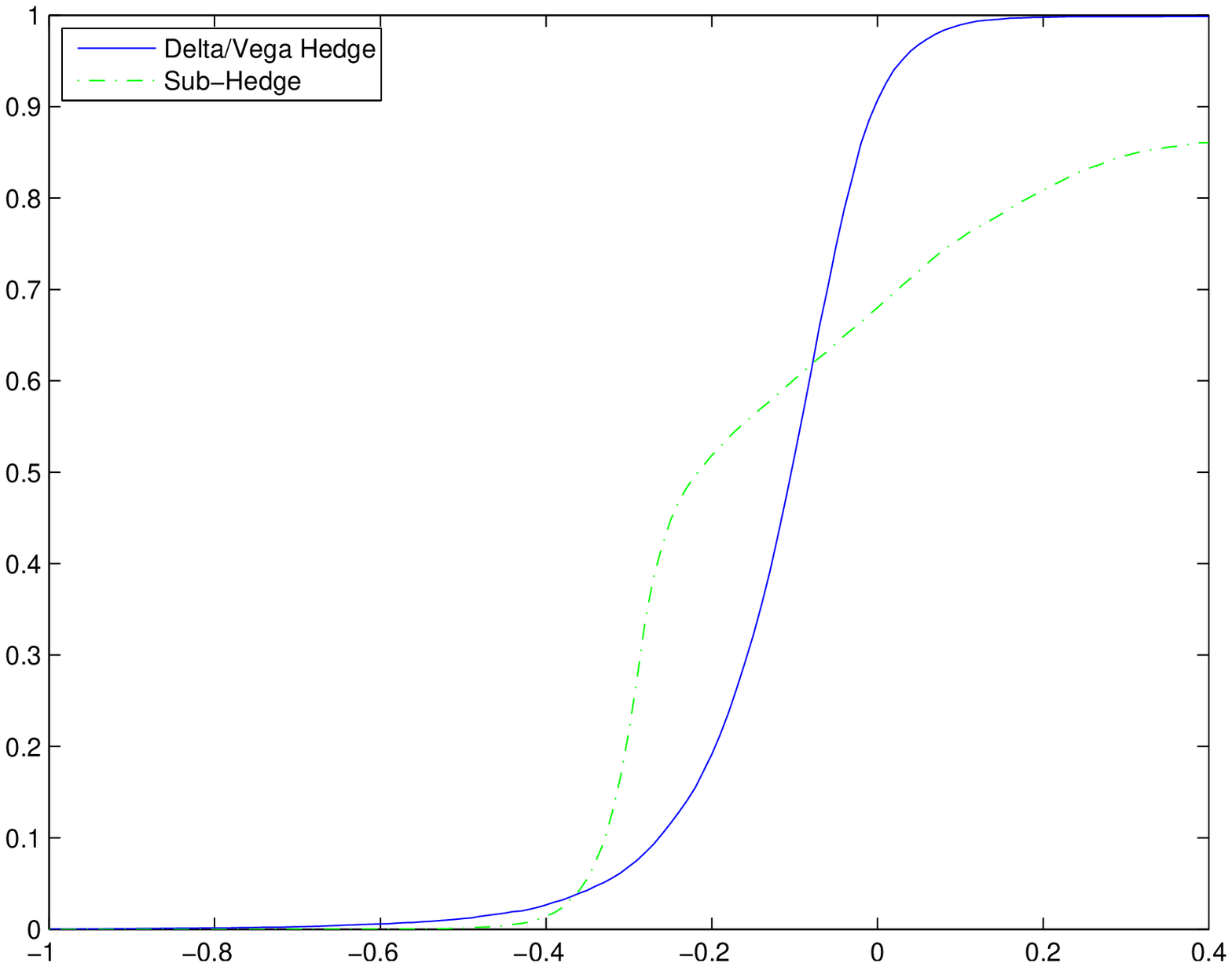}
}
\end{multicols}

\caption{Cumulative distributions of hedging errors under different
  scenarios of a short position (\emph{left}) and a long position
  (\emph{right}) in a double touch option with barriers at $120$ and
  $95$ under the Heston model \eqref{eq:heston}--\eqref{eq:heston_par}.}
\label{fig:hegerr_120_95}
\end{figure}
\begin{figure}[htbp]
\begin{multicols}{2}{
\includegraphics[height=6cm]{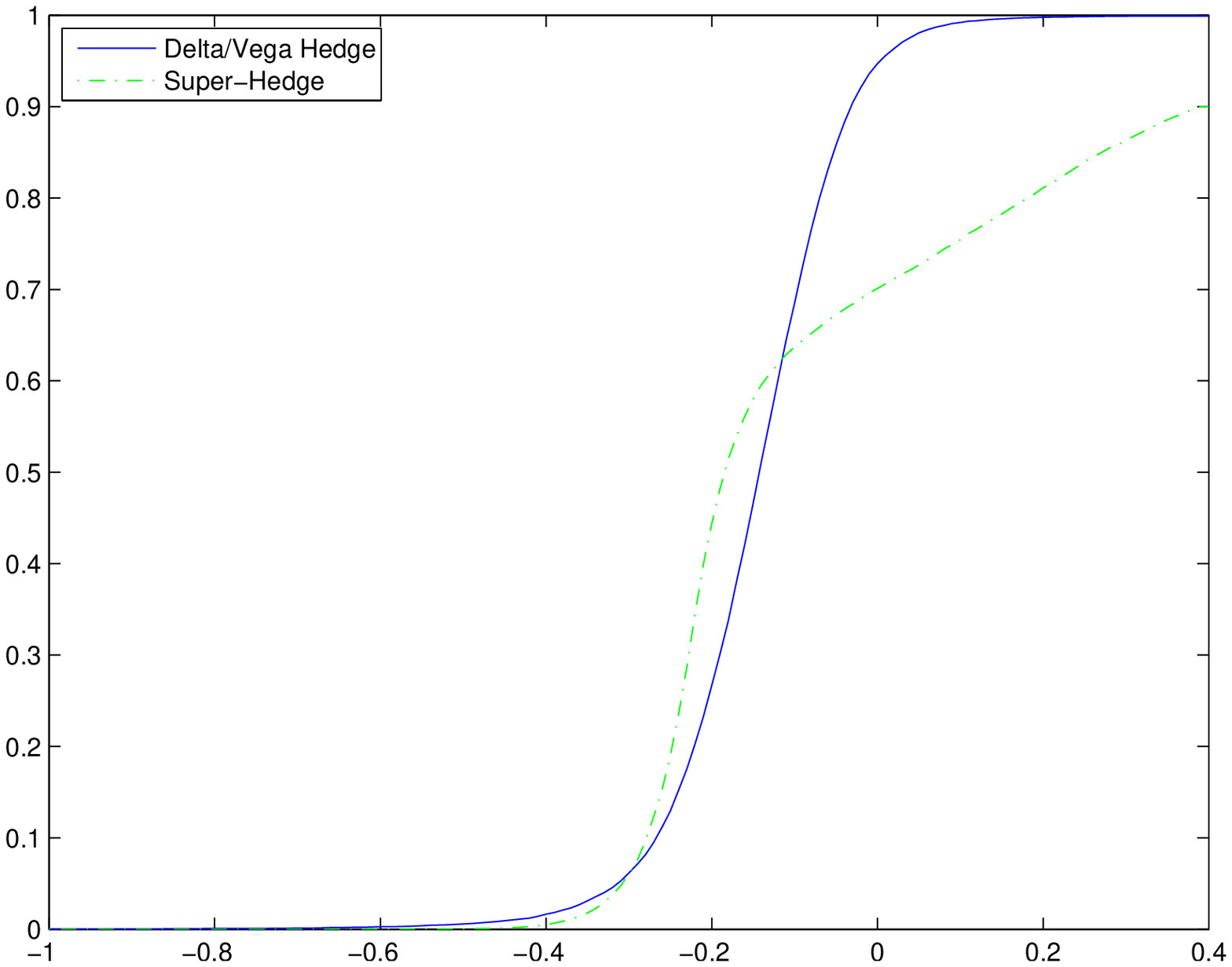}

\includegraphics[height=6cm]{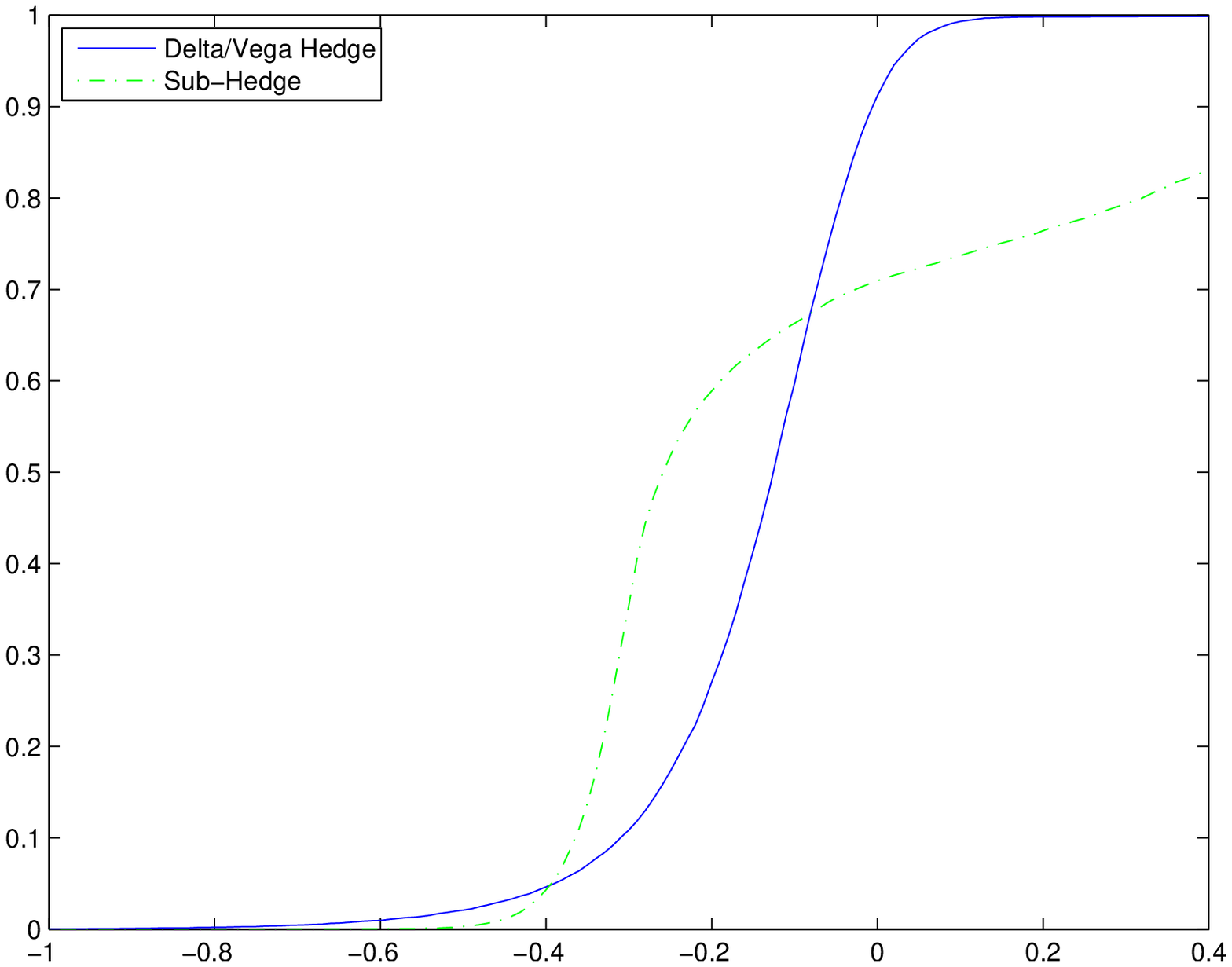}
}
\end{multicols}

\caption{Cumulative distributions of hedging errors under different
  scenarios of a short position (\emph{left}) and a long position
  (\emph{right}) in a double touch option with barriers at $105$ and
  $80$ under the Heston model \eqref{eq:heston}--\eqref{eq:heston_par}.}
\label{fig:hegerr_105_80}
\end{figure}

\newpage

\bibliographystyle{alpha}
\bibliography{general}

\end{document}